\documentclass[12pt]{amsart}
\usepackage{a4wide}
\usepackage[latin9]{inputenc}
\usepackage{amsmath}
\usepackage{amsfonts}
\usepackage{amssymb}
\usepackage{amsthm}
\usepackage{subfigure}
\usepackage{xcolor}
\newtheorem{theo}{Theorem}[section]
\newtheorem{prop}[theo]{Proposition}
\newtheorem{coro}[theo]{Corollary}

\newtheorem{conj}[theo]{Conjecture}

\theoremstyle{definition}
\newtheorem{def1}[theo]{Definition}
\theoremstyle{remark}
\newtheorem{rema}[theo]{Remark}
\newtheorem{ex}[theo]{Example}
\newcommand{\Op}{\operatorname{Op}}

\newcommand{\nwc}{\newcommand}
\nwc{\eps}{\epsilon}
\nwc{\ep}{\epsilon}
\nwc{\vareps}{\varepsilon}
\nwc{\Oph}{\operatorname{Op}_\hbar}
\nwc{\la}{\langle}
\nwc{\ra}{\rangle}

\nwc{\mf}{\mathbf} 
\nwc{\blds}{\boldsymbol} 
\nwc{\ml}{\mathcal} 

\nwc{\defeq}{\stackrel{\rm{def}}{=}}

\nwc{\cE}{\ml{E}}
\nwc{\cN}{\ml{N}}
\nwc{\cO}{\ml{O}}
\nwc{\cP}{\ml{P}}
\nwc{\cU}{\ml{U}}
\nwc{\cV}{\ml{V}}
\nwc{\cW}{\ml{W}}
\nwc{\tU}{\widetilde{U}}
\nwc{\IN}{\mathbb{N}}
\nwc{\IR}{\mathbb{R}}
\nwc{\IS}{\mathbb{S}}
\nwc{\IZ}{\mathbb{Z}}
\nwc{\IC}{\mathbb{C}}
\nwc{\IT}{\mathbb{T}}
\nwc{\be}{\mathbf{e}}
\nwc{\tP}{\widetilde{P}}
\nwc{\tPi}{\widetilde{\Pi}}
\nwc{\tV}{\widetilde{V}}
\nwc{\supp}{\operatorname{supp}}
\nwc{\rest}{\restriction}

\date{\today}

\begin{document}

\title[Semiclassical behaviour of quantum eigenstates]{Semiclassical behaviour of quantum eigenstates}

\author[Gabriel Rivi\`ere]{Gabriel Rivi\`ere}

\address{Laboratoire de math\'ematiques Jean Leray (U.M.R. CNRS 6629), Universit\'e de Nantes, 2 rue de la Houssini\`ere, BP92208, 
44322 Nantes Cedex 3, France}
\email{gabriel.riviere@univ-nantes.fr}

\begin{abstract} Given a quantum Hamiltonian, we explain how the dynamical properties of 
the underlying classical system affect the behaviour of quantum eigenstates in the semiclassical limit. We study this problem 
via the notion of semiclassical measures. We mostly focus on two opposite dynamical paradigms: completely integrable systems 
and chaotic ones. We recall standard tools from microlocal analysis and from dynamical systems. We show how to use them in order 
to illustrate the classical-quantum correspondance and to compare properties of completely integrable and chaotic systems. 

\end{abstract}

\maketitle

\tableofcontents

\section{Introduction}

The semiclassical regime of quantum mechanics is a regime where the Planck constant is negligible compared with the 
other physical actions of the system. In this limit, the physical system under consideration is governed by the equations of 
classical mechanics and one expects that the nature (e.g. chaotic versus integrable) of the corresponding classical system
is reflected in the behaviour of the quantum system. From the mathematical point of view, this type of asymptotics is at the 
heart of microlocal analysis which has a wide range of applications: spectral theory, partial differential equations, 
symplectic topology, hyperbolic dynamics, random geometry, etc. This course will be mostly devoted to the 
application of microlocal analysis to the study of quantum eigenstates in the semiclassical limit 
which can also be viewed as a question of spectral geometry. Our main emphasis will be on
the use of tools from semiclassical analysis in order to study the 
impact of classical dynamics on this problem. These lectures are complementary to the ones of Fr\'ed\'eric Faure 
which are devoted to the application of semiclassical methods to classical chaotic dynamics~\cite{Faure19}.

The plan of these lectures is as follows. In section~\ref{s:background}, we collect some standard facts from 
dynamical systems and semiclassical analysis that will be used in the rest of these notes. The expert can easily 
skip (or read quickly) this part and we hope that the motivated non-expert will find enough material to understand 
what follows. In this section, we also point several classical books that can be used to find some complementary material 
and also the proofs of the results described in this part. In section~\ref{s:observability}, we discuss the question of observability 
and relate it to the quantum ergodicity property. The 
last three parts focus on three important examples: Zoll manifolds where the geodesic flow is periodic 
(section~\ref{s:zoll}), flat tori which are the simplest examples of nondegenerate completely integrable systems 
(section~\ref{s:torus}) and negatively curved manifolds where the geodesic flow is strongly chaotic 
(section~\ref{s:QE}). For these three models, we try to give a review of what is currently known on the 
semiclassical behaviour of quantum eigenstates with a strong emphasis on the properties derived from semiclassical methods.

To end this introduction, we mention that several important and interesting aspects on the study of semiclassical 
quantum eigenstates will not be described in these lecture notes. In fact, we chose to focus on their description via the notion 
of semiclassical measures which is natural from the point of view of quantum mechanics. We think that this point of view 
already illustrates well the interplay between classical and quantum mechanics which is at the heart of this 
school. Among the subjects that we will not (or only very briefly) discuss are the asymptotic 
distribution of eigenvalues, the relation with analytic number theory, the growth of $L^p$ norms, the geometry of 
nodal sets, random waves models, graph models, etc.

\section{Background material}\label{s:background}

In this preliminary section, we define Laplace eigenfunctions and discuss a couple of important examples. We also describe 
relevant tools from symplectic geometry and semiclassical analysis that we apply to define the so-called 
semiclassical measures and to derive some of their first properties. Details on this background material can for instance be 
found in Zworski's book on semiclassical analysis~\cite{Zworski12}.

\subsection{Laplace-Beltrami operators}

In all these lectures, $(M,g)$ denotes a compact, oriented, connected, smooth and Riemannian manifold which has no boundary and 
dimension $n\geq 2$. The Laplace-Beltrami operator is then defined as the unique operator $\Delta_g$ verifying
$$\forall \psi_1,\psi_2\in\mathcal{C}^{\infty}(M),\quad-\int_M\psi_1 \Delta_g\psi_2d\text{Vol}_g
=\int_M\langle d\psi_1,d\psi_2\rangle_{g^*}d\text{Vol}_g,$$
where $d$ is the usual differential, $g^*$ is the induced metric on $T^*M$ and $\text{Vol}_g$ is the 
Riemannian volume on $M$. In local coordinates $(x_1,\ldots,x_n)$, this operator reads
$$\Delta_g \psi=\frac{1}{\sqrt{|\det g|}}\sum_{i,j}\partial_{x_i} \left(g^{ij}\sqrt{|\det g|}\partial_{x_j}\psi\right).$$
In the following, we will focus our attention on solutions to the Laplace equation, i.e.
\begin{equation}\label{e:Laplace}-\Delta_g\psi_{\lambda}=\lambda^2\psi_{\lambda},\quad\|\psi_{\lambda}\|_{L^2}=1,\end{equation}
where $\lambda\geq 0$. Such solutions are called (normalized) Laplace eigenfunctions and they correspond to stationary solutions 
$u(t,x)=e^{-\frac{it\lambda^2}{2}}\psi_\lambda(x)$ of 
the Schr\"odinger equation
$$i\partial_tu=-\frac{1}{2}\Delta_gu.$$
The eigenvalues $\lambda^2$ represent the energy of the quantum eigenstate $\psi_\lambda$~\cite[\S 2.5]{FeynmanQM}.
Solutions to~\eqref{e:Laplace} are smooth by elliptic regularity~\cite[Prop.~2.2.4]{Sogge14}. 
Following~$\S 3.1$ from the same reference, one also knows that there exists a nondecreasing sequence 
$$0=\lambda_0<\lambda_1\leq\lambda_2\leq\ldots\leq\lambda_j\rightarrow+\infty,$$
and an orthonormal basis $(\psi_{\lambda_j})_{j\geq 0}$ of $L^2(M)$ such that
$$\forall j\geq 0,\quad-\Delta_g\psi_{\lambda_j}=\lambda_j^2\psi_{\lambda_j}.$$
The H\"ormander-Weyl's law~\cite[Th.~3.3.1]{Sogge14} also tells us:
\begin{equation}\label{e:weyl}
 N(\lambda):=\left|\left\{j\in\IZ_+:\lambda_j\leq\lambda\right\}\right|=
 \frac{\pi^{\frac{n}{2}}}{(2\pi)^n\Gamma\left(\frac{n}{2}+1\right)}\text{Vol}_g(M)\lambda^n+\ml{O}(\lambda^{n-1}),\ \text{as}\  \lambda\rightarrow+\infty.
\end{equation}

The purpose of these notes is to give some insights on the properties of Laplace eigenfunctions in the high energy 
limit $\lambda\rightarrow +\infty$. This means that we are looking at quantum eigenstates 
which are more and more ``excited''. From the correspondence principle of quantum 
mechanics~\cite[\S 7.4]{FeynmanQM}, one expects that the properties of these eigenmodes are related to the 
properties of the underlying classical Hamiltonian, here the geodesic flow (see below). This regime is called the 
\emph{semiclassical regime} of quantum mechanics. 

As we shall see, except for certain specific 
examples\footnote{Even in these cases, the asymptotic description of these solutions 
turns out to be a subtle question.}, there is no explicit description of these eigenfunctions. Yet, we would like 
to measure if these quantum eigenstates can develop singularities in the large eigenvalue limit and to 
understand how this is related to the properties of the corresponding classical Hamiltonian. There are several natural ways to 
attack this question and let us list three of them:
\begin{itemize}
 \item Given $\beta\in\mathbb{C}$, one can try to describe the level sets of these solutions:
 $$\ml{Z}_{\psi_{\lambda}}(\beta):=\left\{x\in M:\psi_{\lambda}(x)=\beta\right\}.$$
 This is related to the question of nodal sets that we will not discuss very precisely in these notes (except in 
 section~\ref{ss:application}). The interested reader can have a look at the survey by Zelditch~\cite{Zelditch13}.
 \item Given $p\geq 2$, one can try to bound $\|\psi_\lambda\|_{L^p(M)}$ by a function depending 
 only on $\lambda$ and on $(M,g)$. In the case $p=+\infty$, it follows from the localized version of~\eqref{e:weyl} 
 that $\|\psi_{\lambda}\|_{L^{\infty}(M)}\leq C_{M,g}(1+\lambda)^{\frac{n-1}{2}}$ -- see Eq.~(3.3.2) in~\cite{Sogge14}. 
 This is a typical question in harmonic analysis which is related to the 
 so-called Strichartz estimates for dispersive PDE. Again, we shall only describe very briefly this approach and, 
 for a description of recent results in this direction, we refer to~\cite{Sogge15}.
 \item Finally, one can try to study the probability measures $d\nu_\lambda(x)=|\psi_\lambda(x)|^2d\text{Vol}_g(x)$. 
 From the point of view of quantum mechanics, given a subset $\omega$ of $M$,
 $$\nu_{\lambda}(\omega)=\int_{\omega}|\psi_\lambda(x)|^2d\text{Vol}_g(x)$$
 represents the probability of finding a quantum particle in the state $\psi_{\lambda}$ inside the set $\omega$. 
\end{itemize}
We will mostly focus our attention on the last problem and let us describe some natural questions that can be raised 
in this direction. First, given any solution $\psi_{\lambda}$ to~\eqref{e:Laplace}, 
it follows from~\cite{Aronszajn56, Cordes56} that, for any \emph{open} set $\omega$, 
\begin{equation}\label{e:unique-continuation}\omega\neq\emptyset\quad\Longrightarrow\quad\nu_{\lambda}(\omega)>0.\end{equation}
It is then natural to look for conditions on $(M,g)$ and $\omega$ under which one can find $c_{\omega,M,g}>0$ such that 
\begin{equation}\label{e:observability}\forall\ \psi_{\lambda}\ \text{solution to~\eqref{e:Laplace}},\quad 0<c_{\omega,M,g}\leq \int_{\omega}|\psi_\lambda(x)|^2d\text{Vol}_g(x).\end{equation}
In that case, we say that Laplace eigenfunctions are uniformly \emph{observable} in the subset $\omega$. 
According to~\eqref{e:unique-continuation}, it is sufficient to restrict ourselves to the case of large eigenvalues, i.e.
\begin{def1} Let $\nu$ be a probability measure on $M$. We say that $\nu$ is a quantum limit if there exists a sequence 
$(\tilde{\psi}_j)_{j\geq 1}$ of normalized Laplace eigenfunctions (not necessarly an orthonormal basis) with eigenvalues 
$\tilde{\lambda}_j\rightarrow+\infty$ such that
$$\forall a\in\mathcal{C}^0(M),\quad\lim_{j\rightarrow+\infty}\int_M a(x)|\tilde{\psi}_j(x)|^2d\text{Vol}_g(x)=\int_Ma(x)d\nu(x).$$
We denote by $\ml{N}(\Delta_g)$ the set of quantum limits.
\end{def1}
\begin{rema}
Recall that, as $M$ is compact, the set of probability measures on $M$ is compact for the weak-$\star$ topology. 
In particular, given any sequence of probability measures on $M$, we can extract a converging subsequence.
\end{rema}

By a contradiction argument, proving~\eqref{e:observability} follows then from showing that there exists a constant 
$\tilde{c}_{\omega,M,g}>0$ such that for every $\nu\in\ml{N}(\Delta_g)$, $\nu(\omega)\geq \tilde{c}_{\omega,M,g}$. Now that we have 
defined $\ml{N}(\Delta_g)$, several other natural questions appear in order to keep track of the regularity of 
Laplace eigenfunctions as $\lambda\rightarrow +\infty$. How big is $\ml{N}(\Delta_g)$ inside the set of all probability 
measure on $M$? What is the Sobolev regularity of an element inside $\ml{N}(\Delta_g)$? These lecture notes are an attempt 
to describe what can be said depending on the geometric framework. 

As we shall use semiclassical conventions, we 
set $\hbar=\lambda^{-1}$ to be the effective Planck constant of our problem and, for $\hbar>0$, we rewrite~\eqref{e:Laplace} as
\begin{equation}\label{e:eigenvalue}
 -\hbar^2\Delta_g\psi_\hbar=\psi_{\hbar},\quad\|\psi_{\hbar}\|_{L^2}=1.
\end{equation}
\begin{rema}
We will often make the small (but standard) abuse of notations to write $\hbar\rightarrow 0^+$ when we mean that 
$\hbar_{j}\rightarrow 0^+$ as $j\rightarrow+\infty$. At some points, we will consider orthonormal basis 
$(\psi_j)_{j\geq 0}$ of $L^2(M)$ made of solutions to~\eqref{e:eigenvalue} with the small abuse of notations that 
$\psi_0=1$ corresponds formally to the parameter $\hbar_0=+\infty$ and that $\psi_j=\psi_{\hbar_j}$. In that case, we 
will say that $(\psi_j)_{j\geq 1}$ is an orthonormal basis of Laplace eigenfunctions. 
\end{rema}

\begin{rema}\label{r:general-operator}
 For the sake of simplicity, we will focus on the case of the (semiclassical) quantum Hamiltonian 
 $$\widehat{H}_{\hbar}:=-\frac{1}{2}\hbar^2\Delta_g.$$
 Yet, some (but not all) of the results presented below can be extended to more general semiclassical 
 operators of the form $-\frac{1}{2}\hbar^2\Delta_g+\eps_{\hbar}V(x)$, where $0\leq \eps_{\hbar}\leq 1$ and where $V$ is a smooth and real-valued function on $M$.
\end{rema}

\begin{rema}\label{r:quasimodes} Most of the results that we will present will also be valid in the more general setting of quasimodes, 
i.e. solutions to
$$-\hbar^2\Delta_g\psi_\hbar=\psi_{\hbar}+\hbar r_{\hbar},$$
where $\|r_{\hbar}\|_{L^2}\rightarrow 0$ at a convenient rate. In order to alleviate the presentation, we will not discuss the size of 
$r_{\hbar}$ and we refer the reader to the corresponding references for more precise statements.

\end{rema}

\subsection{Some important examples} Before describing the analytical and dynamical tools needed for our study, let us describe 
three important examples.

\subsubsection{The flat torus}\label{sss:torus}

The simplest example is given by the torus $\mathbb{T}^n=\IR^n/(2\pi\IZ)^n$ 
endowed with its canonical (Euclidean) metric. In that case, $\Delta_g=\Delta$ is the usual Laplacian on $\IR^n$ that we 
restrict to $2\pi\IZ^n$-periodic functions. For $\hbar>0$, solutions to~\eqref{e:eigenvalue} are 
then given by
$$\psi_{\hbar}(x)=\sum_{k\in\IZ^n:\|k\|\hbar=1}\hat{c}_ke^{ik.x},\quad(2\pi)^n\sum_{k\in\IZ^n:\|k\|\hbar=1}|\hat{c}_k|^2=1.$$
\begin{rema} In that case, the H\"ormander-Weyl's law~\eqref{e:weyl} has a simple proof by observing that
$$N(\hbar^{-1})=\left|\left\{k\in\IZ^n:\|k\|\leq\hbar^{-1}\right\}\right|,$$
and that this quantity can be compared with the volume of the ball of radius $\hbar^{-1}$ centered at $0$. 
\end{rema}
It is then easy to construct some elements inside $\ml{N}(\Delta)$. For instance, by taking 
$\psi_{\hbar}(x)=e^{ik.x}/(2\pi)^{\frac{n}{2}}$ with $\hbar=\|k\|^{-1}\rightarrow 0^+$, one finds that 
$d\nu(x)=\frac{dx}{(2\pi)^n}$ belongs to $\ml{N}(\Delta)$. One can also construct more complicated examples. For instance, let 
$p$ be an element in $\IZ^{n-1}$ and write $\psi_{\hbar}(x)=C\sin(p.x')e^{inx_1}$ where $x=(x_1,x')\in\IT^n$, 
$C$ is a normalizing constant and $\hbar_n=(\|p\|^2+n^2)^{-1}\rightarrow 0^+$. Then, 
$d\nu(x)=C^2|\sin(p.x')|^2dx$  belongs to $\ml{N}(\Delta)$. 

From these simple examples, one can already feel that, 
even if Laplace eigenfunctions have an explicit expression, their asymptotic description may lead to subtle 
arithmetic questions as their structure is related to the distribution of lattice points on circles of large radius. Exploiting 
this arithmetic structure, one can in fact show the following result on the regularity of elements 
inside $\ml{N}(\Delta)$\cite{Cooke71, Zygmund74, Jakobson97}:

\begin{theo}\label{t:zygmund} Let $\IT^n$ be the torus endowed with its canonical metric. Then, 
every element $\nu$ of $\ml{N}(\Delta)$ is absolutely continuous with respect to the Lebesgue measure. Moreover, if $n=2$, the density of 
$\nu$ belongs to $L^2(\IT^2)$. 
\end{theo}
 
\begin{proof} Let us give the proof in the case $n=2$ which follows from results due to Cooke~\cite{Cooke71} and 
Zygmund~\cite{Zygmund74}. In higher dimension, this Theorem is due to Bourgain and this was 
reported in~\cite{Jakobson97}. As the proof for $n\geq 3$ is more involved, we will not discuss it here. In dimension $2$, 
it was shown in~\cite{Cooke71, Zygmund74} that, for any solution $\psi_{\hbar}$ 
to~\eqref{e:eigenvalue} on $\IT^2$, one has $\|\psi_{\hbar}\|_{L^4(\IT^2)}\leq 3^{\frac{1}{4}}/\sqrt{2\pi}$. Once we have 
this result, one can verify via the Cauchy-Schwarz inequality that, for any $\nu\in\ml{N}(\Delta)$,
$$\forall a\in\ml{C}^0(\IT^2),\quad\left|\int_{\IT^2}a(x)d\nu(x)\right|\leq \sqrt{3}\|a\|_{L^2(\IT^2)}.$$
In particular, $\nu\in L^2(\IT^2)$ by Riesz representation Theorem. 
Hence, it remains to prove the upper bound on the $L^4$ norm of $\psi_{\hbar}$. For that purpose, we write
$$|\psi_{\hbar}(x)|^2=\sum_{k,l:\|k\|\hbar=\|l\|\hbar=1}\hat{c}_k\overline{\hat{c}_l}e^{i(k-l).x}
=\sum_{p\in\IZ^n}e^{ip.x}\sum_{k:\|k\|\hbar=\|k-p\|\hbar=1}\hat{c}_k\overline{\hat{c}_{k-p}}.$$
For $p=0$, the sum over $k$ is equal to $1/(2\pi)^2$ as $\psi_{\hbar}$ is normalized in $L^2$. 
For $p\neq 0$, one can verify that there are at most two terms in the sum over $k$. Hence, by Plancherel Theorem and by 
Cauchy-Schwarz inequality, one finally has the expected inequality
$$\int_{\IT^2}|\psi_{\hbar}(x)|^4dx\leq \frac{1}{(2\pi)^2}+2(2\pi)^2\sum_{p\in\IZ^n\setminus 0}
\sum_{k:\|k\|\hbar=\|k-p\|\hbar=1}|\hat{c}_k|^2|\overline{\hat{c}_{k-p}}|^2\leq\frac{3}{(2\pi)^2}.$$
\end{proof}
The fact that quantum limits are absolutely continuous can be recovered by semiclassical methods as the ones 
described in paragraph~\ref{ss:semiclassical-torus}. This semiclassical proof was first obtained by 
Maci\`a~\cite{Macia10} in dimension $2$ and generalized later on to higher dimensions by Anantharaman and 
Maci\`a~\cite{AnantharamanMacia14}. Compared with the approach from~\cite{Cooke71, Zygmund74, Jakobson97}, 
much less can be deduced on the regularity of elements inside $\ml{N}(\Delta)$. Yet, it has the advantage that 
it can be generalized to more general completely integrable systems~\cite{AnantharamanFermanianMacia15, AnantharamanLeautaudMacia16}.

\subsubsection{The round sphere}

Consider now the $2$-sphere
$$\mathbb{S}^2:=\left\{(x,y,z)\in\IR^3:x^2+y^2+z^2=1\right\}.$$
In spherical coordinates, the canonical metric can be written as $g(\theta,\phi)=d\theta^2+\sin^2\theta d\phi^2$. Hence, 
one has
$$\Delta_{\IS^2}=\frac{1}{\sin^2\theta}\partial_\phi^2+\frac{1}{\sin\theta}\partial_\theta\left(\sin\theta\partial_{\theta}\right).$$
Solutions to~\eqref{e:eigenvalue} are linear combinations of the spherical harmonics~\cite[\S 3.4]{Sogge14}. 
In that case, one has $\hbar^{-2}=l(l+1)$ for some integer $l\geq 0$, and for every such eigenvalue, the algebraic 
multiplicity is $2l+1$. Again, we have exact expressions which can be written in terms of Legendre polynomials. This allows us to construct explicit 
examples of elements lying in $\ml{N}(\Delta_g)$. For instance, given $l\in\IZ_+$, one can consider the following example of Laplace eigenfunction
$$\psi_{\hbar_l}(\theta,\phi)=c_le^{il\phi}(\sin\theta)^l,$$
where $c_l$ is some normalizing constant. A calculation shows that this constant is of order $l^{1/4}\asymp\hbar_l^{-1/4}$ 
as $l\rightarrow+\infty$. Similarly, one can verify that
$$\forall a\in\mathcal{C}^0(\IS^2),\quad\lim_{l\rightarrow+\infty}\int_{\IS^2}a|\psi_{\hbar_l}|^2d\text{Vol}_g
=\int_0^{2\pi}a\left(\frac{\pi}{2},\phi\right)\frac{d\phi}{2\pi}.$$
Hence, the Lebesgue measure along the closed geodesic 
$\gamma_0=\{\left(\frac{\pi}{2},\phi\right):0\leq\phi\leq 2\pi\}$ belongs to $\ml{N}(\Delta_{\IS^2})$. In particular, we do not 
have any analogue of Theorem~\ref{t:zygmund} in the case of the $2$-sphere. By spherical symmetry, one can deduce that, 
given any closed geodesic $\gamma$ on $\IS^2$, the Lebesgue measure $\delta_{\gamma}$ along it 
belongs to $\ml{N}(\Delta_{\IS^2})$. More generally, a similar calculation shows that any finite convex combination of 
such measures is still inside $\ml{N}(\Delta_{\IS^2})$. Hence, if we denote by $\ml{N}(\IS^2)$ the closed convex hull of these measures, 
one has
\begin{equation}\label{e:QL-sphere}
 \ml{N}(\IS^2)\subset\ml{N}(\Delta_{\IS^2}).
\end{equation}
As a consequence of the results from paragraph~\ref{ss:semiclassical-measure}, we shall see that one has in fact 
equality. Hence, in that case, we have a complete description of the set of quantum limits. Again, compared with the case of the 
torus, it contains singular elements as eigenfunctions can concentrate along closed geodesics. 

\subsubsection{The hyperbolic plane}\label{sss:hyperbolic-plane} In the previous paragraphs, we saw two models corresponding respectively to manifolds 
with $0$ curvature (the torus) and with constant positive curvature (the sphere). Let us now discuss the case 
of hyperbolic surfaces which are the simplest models of manifolds with negative curvature. We set 
$$\mathbb{H}:=\{x+iy\in\IC:y>0\},$$
that we endow with the Riemannian metric $\frac{dx^2+dy^2}{y^2}$. One can verify that the sectional curvature is constant and equal to $-1$. 
In this system of coordinates, the Laplace-Beltrami operator (or hyperbolic Laplacian) can be written as
$$\Delta_{\mathbb{H}}=y^2\left(\partial_x^2+\partial_y^2\right).$$
If we fix a subgroup $\Gamma$ of $SL(2,\mathbb{R})$ acting by isometry on $\mathbb{H}$ via the map $\left(\begin{array}{cc}
         a & b\\
         c & d
        \end{array}\right).z=\frac{az+b}{cz+d},$
then we can consider the quotient $M=\mathbb{H}/\Gamma$ endowed with the induced Riemannian structure. If $M$ is compact, we are in the framework described above 
and eigenfunctions are eigenfunctions of the hyperbolic Laplacian which are $\Gamma$-invariant. Compared with the 
previous examples, we do not have an explicit expression for such solutions, or at least solutions which are 
easy to manipulate. An important class of such manifolds are compact arithmetic surfaces which correspond to 
subgroups $\Gamma$ that are derived from a quaternion division algebra $\ml{A}$~\cite{RudnickSarnak94}. 
In that case, one has more structure and one can expect to say more on the structures of elements inside 
$\ml{N}(\Delta_{g})$ thanks to the arithmetic nature of the problem -- see Section~\ref{s:QE}. 


\subsection{Hamiltonian flows} Given a smooth function $H:T^*M\rightarrow \mathbb{R}$, we recall how to define its 
corresponding Hamiltonian vector field $X_H$. To that aim, one needs to introduce the Liouville $1$-form, which is 
an element $\theta$ of $\Omega^1(T^*M)$ defined as
$$\forall (x,\xi)\in T^*M,\quad\theta_{x,\xi}: V\in T_{x,\xi}T^*M\mapsto\xi(d_{x,\xi}\pi V),$$
where $\pi: (x,\xi)\in T^*M\mapsto x\in M$ is the canonical projection. In local coordinates 
$(x_1,\ldots,x_n,\xi_1,\ldots,\xi_n)$, it reads
$$\theta_{x,\xi}=\sum_{i=1}^n\xi_idx_i.$$
Then, we define the symplectic form on $T^*M$ as $\omega=d\theta,$ which can be written in local coordinates as
$$\omega_{x,\xi}=\sum_{i=1}^nd\xi_i\wedge dx_i.$$ The Hamiltonian vector field of $H$ is then the unique vector field $X_H$ 
such that
$$dH=\omega(.,X_H).$$
If $H$ is smooth, then $X_H$ induces a smooth flow $\varphi_H^t:T^*M\rightarrow T^*M$. In local coordinates, one has the 
\emph{Hamilton equations}:
$$\forall 1\leq i\leq n,\quad\dot{x}_i=\partial_{\xi_i}H,\quad\dot{\xi}_i=-\partial_{x_i}H.$$ 
Given two smooth functions $a$ and $b$ on $T^*M$, we can define their Poisson bracket:
$$\{a,b\}=X_a(b)=db(X_a)=\omega(X_a,X_b)=\sum_{i=1}^n\left(\partial_{\xi_i}a\partial_{x_i}b-\partial_{x_i}a\partial_{\xi_i}b\right),$$
which will be useful when establishing the semiclassical correspondence for pseudodifferential operators in the following. By construction, the 
flow $\varphi_H^t$ preserves $H$ and $\omega$ along the evolution. In particular, it preserves the 
Liouville measure $L=\frac{|\omega^n|}{n!}.$ Even if we will encounter other Hamiltonian functions at certain points, we 
will mostly focus on the one corresponding to the free evolution on $(M,g)$, i.e.
$$H_{0}(x,\xi)=\frac{\|\xi\|_{g^{*}(x)}^2}{2}.$$
In that case, we will write $\varphi^t=\varphi_{H_0}^t$ which is nothing else but the geodesic flow 
on $T^*M$. We will also denote by $L_1$ the desintegration of the Liouville measure along the energy 
layer $S^*M=H_0^{-1}(1/2)$ that we normalize to have $L_1(S^*M)=1$. 
We still have that, for every $t\in\mathbb{R}$, $\varphi^t_*L_1=L_1$. Equivalently, for every 
continuous function $a$ on $S^*M$ and for every $t$ in $\IR$,
\begin{equation}\label{e:invariance-Liouville}
 \int_{S^*M} a\circ\varphi^tdL_1=\int_{S^*M}adL_1.
\end{equation}
For more details and background facts on symplectic geometry, 
we refer to~\cite[Ch.~2]{Zworski12}.

\subsection{Birkhoff ergodic Theorem} We now collect a few basic facts from ergodic theory that will be needed 
in our analysis. We only state them for the geodesic flow $\varphi^t:S^*M\mapsto S^*M$ but they are in fact valid 
in the more general context of continuous flows acting on a compact metric space. Ergodic theory is concerned with 
the study of flows (or maps) which preserve some measure $\mu$. We already encounter such an example of measure with the geodesic 
flow and the Liouville measure -- see~\eqref{e:invariance-Liouville}. More generally, we define
\begin{def1} Let $\mu$ be a probability measure on $S^*M$. We say that $\mu$ is a $\varphi^t$-invariant probability measure 
on $S^*M$ if, for any continuous function $a$ on $S^*M$,
$$\forall t\in\IR,\quad \int_{S^*M} a\circ\varphi^td\mu=\int_{S^*M}ad\mu.$$
 We denote by $\ml{M}(\varphi^t)$ the set of all such measures.
\end{def1}
As $L_1\in\ml{M}(\varphi^t)$, this set is nonempty. Other examples of elements in $\ml{M}(\varphi^t)$ are given by measures carried by 
closed geodesics. More precisely, given a closed orbit $\gamma$ of $\varphi^t$ with minimal period $T_{\gamma}$ and given a continuous 
function $a$ on $S^*M$, we define
$$\int_{S^*M} ad\mu_{\gamma}=\frac{1}{T_{\gamma}}\int_0^{T_{\gamma}}a\circ\varphi^t(x,\xi)dt,$$
where $(x,\xi)$ is any point on the closed orbit. 

The set $\ml{M}(\varphi^t)$ is endowed with the weak topology induced by $\ml{C}^0(S^*M,\IR)$ on its dual. One can verify 
that it is a compact, convex and nonempty subset of $\ml{M}(\varphi^t)$. By definition, we say that a measure $\mu$ is 
ergodic if it is an extremal point of this convex set. Equivalently, it means that, if $\mu=t\mu_1+(1-t)\mu_2$ 
with $\mu_1$, $\mu_2$ in $\ml{M}(\varphi^t)$ and $0<t<1$, then $\mu=\mu_1=\mu_2$. 
Again, $\mu_{\gamma}$ is the simplest example of such a measure. 
In the case of negatively curved manifolds, the Liouville measure $L_1$ was shown to be ergodic by Anosov~\cite{Anosov67}. 
An equivalent manner to define ergodic measures is to require that, if $A$ is a subset such that, 
for every $t$ in $\IR$, $\varphi^t(A)=A$, then $\mu(A)=0$ or $1$. In the case of the geodesic flow on 
$S^*\mathbb{S}^d$ (for the round metric), one can then verify that the ergodic measures are exactly given by the 
measures $\mu_{\gamma}$. A central result in ergodic theory is the Birkhoff 
ergodic Theorem\footnote{We do not state this Theorem in its full generality and we only mention 
what will be needed in our analysis.}~\cite[\S 8.6]{EinsiedlerWard11}:
\begin{theo}[Birkhoff Ergodic Theorem]\label{t:birkhoff} Let $\mu$ be an element of $\ml{M}(\varphi^t)$. Then, for every 
continuous function $a:S^*M\rightarrow \IR$ and for $\mu$-a.e. $(x,\xi)$ in $S^*M$,
\begin{equation}\label{e:birkhoff}\lim_{T\rightarrow+\infty}\frac{1}{T}\int_0^Ta\circ\varphi^t(x,\xi)dt=
\mu_{x,\xi}(a),\end{equation}
 where $(x,\xi)\mapsto\mu_{x,\xi}(a)$ belongs to $L^{\infty}(S^*M,\mu)$. Moreover, $(x,\xi)\mapsto\mu_{x,\xi}(a)$ is $\varphi^t$-invariant and one has
$$\int_{S^*M}\mu_{x,\xi}(a)d\mu(x,\xi)=\int_{S^*M}ad\mu,$$
Finally, if $\mu$ is ergodic then, for $\mu$-a.e. $(x,\xi)$ in $S^*M$,
$$\mu_{x,\xi}(a)=\int_{S^*M}ad\mu.$$
\end{theo}
\begin{rema}\label{r:ergodic} Note that one can choose a full measure subset $\Omega$ (w.r.t. $\mu$) such that~\eqref{e:birkhoff} 
holds for a dense and countable subset of functions $(a_k)_{k\geq 1}$ inside $\mathcal{C}^0(S^*M,\IR)$. 
By construction, one has $|\mu_{x,\xi}(a_k)|\leq\|a_k\|_{\ml{C}^0}$ for every $k\geq 1$. Hence, by the Riesz 
representation Theorem, $\mu_{x,\xi}$ can be identified with a Radon measure on $S^*M$. One can also
verify that $\mu_{x,\xi}$ is an ergodic measure in $\ml{M}(\varphi^t)$ for $\mu$-a.e. $(x,\xi)$ in $S^*M$. 
Thus, we can write the ergodic decomposition of any measure $\mu\in\ml{M}(\varphi^t)$:
$$\mu=\int_{S^*M}\mu_{x,\xi}d\mu(x,\xi).$$ 
\end{rema}

Let us now illustrate this Theorem in the three geometries that we have encountered so far.
\begin{ex} In the case of a negatively curved manifold, Anosov proved that $L_1$ is an ergodic measure~\cite{Anosov67}. 
This tells us that, for $L_1$-a.e. $(x,\xi)$ in $S^*M$,
$$\forall a\in \ml{C}^0(S^*M,\IR),\quad\lim_{T\rightarrow+\infty}\frac{1}{T}\int_0^Ta\circ\varphi^t(x,\xi)dt=\int_{S^*M}adL_1.$$
 It means that if we look at the average of $a$ along the orbit of a typical point of $S^*M$ (w.r.t. 
 the natural volume measure on $S^*M$), then it converges to the average of $a$ over $S^*M$.
\end{ex}
\begin{ex} Let $(\mathbb{S}^n,g_0)$ be the $n$-dimensional sphere endowed with its canonical metric $g_0$. In that case, 
the geodesic flow on $S^*\IS^n$ is $2\pi$-periodic, and one has, \emph{for every} $(x,\xi)\in S^*\IS^n$ and for 
every continuous function $a$, 
$$\lim_{T\rightarrow+\infty}\frac{1}{T}\int_0^Ta\circ\varphi^t(x,\xi)dt=\mu_{\gamma_{x,\xi}}(a),$$
where $\gamma_{x,\xi}:=\{\varphi^t(x,\xi):0\leq t\leq 2\pi\}$ is the closed orbit issued from $(x,\xi)$. 
Observe that the map $(x,\xi)\mapsto \mu_{\gamma_{x,\xi}}(a)$ has the same regularity as the function $a$. Ergodic measures are 
exactly the measures $\mu_{\gamma}$.
\end{ex}
\begin{ex}\label{ex:torus} Let $(\mathbb{T}^n,g_0)$ be the $n$-dimensional torus endowed with its Euclidean metric. Given 
$\xi\in\IS^{n-1}$, we define 
$$\Lambda_{\xi}:=\{k\in\IZ^{n}:k.\xi=0\},$$
and $\Lambda_{\xi}^\perp$ which is the orthononormal complement of the vector space generated by $\Lambda_{\xi}$. 
If the rank of $\Lambda_{\xi}$ is equal to $n-j$, then the orbit of the point $(x,\xi)$ under the geodesic flow 
$\varphi^t$ fills a torus of dimension $j$. 
More precisely, one has, \emph{for every} $(x,\xi)$ in $S^*\IT^n$ and for every continuous function $a$, 
 $$\lim_{T\rightarrow+\infty}\frac{1}{T}\int_0^Ta\circ\varphi^t(x,\xi)dt=\int_{\IT(\xi)}a(x+y,\xi)d\mathfrak{h}_\xi(y),$$
 where $\IT(\xi)=\Lambda_{\xi}^{\perp}/(2\pi\IZ^n\cap\Lambda_\xi^{\perp})$ and $\mathfrak{h}_\xi$ is the normalized Haar measure 
 on $\IT(\xi)$. These are exactly the ergodic measures of the geodesic flow on the flat torus. 
 \end{ex}

From these three examples, we can remark that the dynamical properties of the geodesic flow are of very 
different nature depending on the choice of the manifold $(M,g)$ and this will be responsible for the different 
behaviour of Laplace eigenfunctions that we will encounter. 

\subsection{Semiclassical calculus}\label{ss:semiclassical} So far, we have constructed two sets of probability measures. 
One set is built from the eigenfunctions
of the Laplacian, and the other one from the geodesic flow. One of our objective is to compare these two sets of 
measures but one is made of measures on $M$ and the other of measures on $S^*M$. In order to overcome this difference, we 
will explain how to lift elements of $\ml{N}(\Delta_g)$ into measures defined on $S^*M$ in a natural manner from the point of view 
of quantum mechanics. This will also connect $\ml{N}(\Delta_g)$ to the dynamics on $S^*M$. 
To that aim, we need to introduce some tools from semiclassical analysis. We will do that in a 
somewhat condensed manner and we invite the reader to look at~\cite[Ch.~4 and 14]{Zworski12} for a more detailed 
presentation. First of all, the goal is to associate to a function $a_{\hbar}(x,\xi)$ on $T^*M$ 
(that may also depend on $\hbar$) an operator acting on $\mathcal{C}^{\infty}(M)$ (ideally on the 
Hilbert space $L^2(M)$). This is called a quantization procedure. 

If we consider the case of $T^*\IR^n\simeq\IR^{2n}$, then a natural way to proceed is to set, for every 
$a$ in the Schwartz class $\ml{S}(\IR^{2n})$ and for every $\hbar>0$,
\begin{equation}\label{e:weyl-quant}\forall u\in\ml{S}(\mathbb{R}^n),\quad\Oph^w(a)u(x):=\frac{1}{(2\pi\hbar)^{n}}
\iint_{\IR^{2n}}a\left(\frac{x+y}{2},\xi\right)e^{\frac{i}{\hbar}\langle x-y,\xi\rangle}u(y)dyd\xi.\end{equation}
This is referred to the Weyl quantization of the observable $a$, and it can be extended to more general classes 
of observables. For instance, if $a(x,\xi)=a(x)\in\mathcal{S}(\mathbb{R}^n)$, then $\Oph^w(a)u(x)=a(x)u(x).$ Similarly, if 
$a(x,\xi)=P(\xi)$ with $P$ a polynomial, then $\Oph^w(a)=P(-i\hbar\partial_x).$ Following~\cite[\S 14.2.2]{Zworski12}, this 
definition can be extended to manifolds. More precisely, we say that $A_{\hbar}:\ml{C}^{\infty}(M)\rightarrow\ml{C}^{\infty}(M)$ 
is an $\hbar$-pseudodifferential operator of order $m$ if\footnote{One also needs to impose some regularizing assumption between distinct charts -- see~\cite{Zworski12} for details.} 
it can be written in local charts as $\Oph^{w}(\tilde{a}_\hbar)$ where $\tilde{a}_\hbar$ is the 
pullback on $\IR^{2n}$ of some function $a_{\hbar}$ (that may depend on $\hbar$) belonging to the class of symbols
$$S^m(T^*M):=\left\{a_{\hbar}(x,\xi)\in\mathcal{C}^{\infty}(T^*M):\forall(\alpha,\beta)\in\IZ_+^{2n},\ 
|\partial_x^{\alpha}\partial_\xi^{\beta}\tilde{a}_{\hbar}(x,\xi)|\leq C_{\alpha,\beta}(1+\|\xi\|_x^2)^{\frac{m-|\beta|}{2}}\right\},$$
with $m$ belonging to $\mathbb{R}$. Here, the constant $C_{\alpha,\beta}$ is uniform for $0<\hbar\leq 1$. The set of 
such operators is denoted by $\Psi^m(M)$ and one has for instance $-\frac{1}{2}\hbar^2\Delta_g\in \Psi^2(M)$. 
We also denote by $\Psi^{-\infty}(M)$ (resp. $S^{-\infty}(T^*M)$) 
the intersection of all the spaces $(\Psi^m(M))_{m\in\IR}$ (resp. $(S^m(T^*M))_{m\in\IR}$). Note that 
$\mathcal{C}^{\infty}_c(T^*M)\subset S^{-\infty}(T^*M).$ According to~\cite[Th.~14.1]{Zworski12}, there exist two maps 
$$\sigma_\hbar:\Psi^m(M)\rightarrow S^{m}(T^*M)/ \hbar S^{m-1}(T^*M),\ \text{and}\ \Oph:S^m(T^*M)\rightarrow \Psi^m(T^*M)$$
such that $\sigma_\hbar(\Oph(a_\hbar))=a_\hbar\ \text{mod}\ \hbar S^{m-1}(T^*M).$ Note that this result says 
in particular that one can associate to each observable $a_{\hbar}$ in $S^m(T^*M)$ an operator 
$\Oph(a_\hbar)$ acting on smooth functions.
\begin{rema}
We say that 
$a_{\hbar}=\sigma_{\hbar}(A_{\hbar})$ is the principal symbol of the operator $A_{\hbar}$. In the case of 
$-\frac{1}{2}\hbar^2\Delta_g$, the principal symbol is equal to $H_0(x,\xi)=\frac{1}{2}\|\xi\|_x^2.$ 
\end{rema}
Let us review (without proof) some of the fundamental properties of these operators. First, we can choose $\Oph$ such that 
$\Oph(a)u=a\times u$ if $a$ depends only on the variable $x$ and such that, for every $u$ and $v$ in $\ml{C}^{\infty}(M)$,
$$\langle\Oph(a)u,v\rangle_{L^2(M)}=\langle u,\Oph(\overline{a})v\rangle_{L^2(M)}.$$
These assumptions will simplify our exposition but they are not crucial and they can be weakened by requiring that there are only true 
modulo error terms of order $\ml{O}(\hbar)$. We now have the following properties that 
relate the properties of the operators obtained through the quantization $\Oph$ to the properties of the 
corresponding functions:
\begin{itemize}
 \item (Composition formula) Let $a_\hbar\in S^{m_1}(T^*M)$ and let $b_\hbar\in S^{m_2}(T^*M)$. One has 
 $$\Oph(a_\hbar)\circ\Op_{\hbar}(b_\hbar)\in\Psi^{m_1+m_2}(M),$$
 and its principal symbol is given by
 \begin{equation}\label{e:composition}\sigma_\hbar\left(\Oph(a_\hbar)\circ\Op_{\hbar}(b_\hbar)\right)=a_\hbar b_\hbar.\end{equation}
 \item (Bracket properties) Let $a_\hbar\in S^{m_1}(T^*M)$ and let $b_\hbar\in S^{m_2}(T^*M)$. One has 
 $$[\Oph(a_\hbar),\Op_{\hbar}(b_\hbar)]\in\hbar \Psi^{m_1+m_2-1}(M),$$ and its principal symbol is given by 
 \begin{equation}\label{e:bracket}\sigma_\hbar\left([\Oph(a_\hbar),\Op_{\hbar}(b_\hbar)]\right)=-i\hbar\{a_\hbar, b_\hbar\}.\end{equation} 
 In the case of the flat torus (or of $\IR^n$) and of the Weyl quantization, one has also an equality when $a=H_0$ 
 in the sense that $[\Oph^w(H_0),\Op_{\hbar}^w(b_\hbar)]=-i\hbar\Oph^w(\{H_0,b_{\hbar}\})$.
 \item (Calder\'on-Vaillancourt Theorem) Let $a_\hbar\in S^{0}(T^*M)$. The operator
 $$\Oph(a_h):L^2(M)\rightarrow L^2(M)$$
 is bounded and its norm is bounded by
 \begin{equation}\label{e:calderon}
  \left\|\Oph(a_h)\right\|_{L^2\rightarrow L^2}\leq C_M\sum_{|\alpha|\leq D_M}\hbar^{\frac{|\alpha|}{2}}\|\partial^{\alpha}a_\hbar\|_{L^{\infty}},
 \end{equation}
where $C_M$ and $D_M$ are constants depending only on $M$.
 \item (Disjoint supports) Let $a_\hbar\in S^{m_1}(T^*M)$ and let $b_\hbar\in S^{m_2}(T^*M)$. Suppose that there exists two 
 disjoint subsets $\Omega_1$ and $\Omega_2$ such that $\text{dist}(\Omega_1,\Omega_2)>0$ 
 and such that, for every $0<\hbar\leq 1$, $\text{supp}(a_\hbar)\subset\Omega_1$ and $\text{supp}(b_\hbar)\subset\Omega_2$. 
 Then,
 \begin{equation}\label{e:disjoint}
  \Oph(a_h)\Oph(b_h)=\ml{O}_{L^2\rightarrow L^2}(\hbar^{\infty}),
 \end{equation}
where the convention $\ml{O}_{L^2\rightarrow L^2}(\hbar^{\infty})$ means that, for every $N\geq 1$, there exists a 
constant $C_N>0$ such that the norm of the operator is less than $C_N\hbar^N$.
 \item (Egorov Theorem) Let $T>0$ and let $a_\hbar\in S^{-\infty}(T^*M)$ which is compactly supported. Then, there exists some constant $C_{T,a}>0$ such that, 
 for every $|t|\leq T$,
 \begin{equation}\label{e:egorov}\left\|e^{-\frac{it\hbar\Delta_g}{2}}\Oph(a_h)e^{\frac{it\hbar\Delta_g}{2}}-\Oph(a_h\circ\varphi^t)\right\|_{L^2\rightarrow L^2}\leq C_{a,T}\hbar.\end{equation}
 This result establishes a correspondence between the classical dynamics induced by the geodesic flow and the quantum dynamics induced 
 by the semiclassical Schr\"odinger equation
 \begin{equation}\label{e:schrodinger}
  i\hbar\partial_tu_\hbar=-\frac{1}{2}\hbar^2\Delta_g u_\hbar.
 \end{equation}
 Determining the range of validity (in terms of $T$) of this kind of approximation is a key issue when studying the 
 fine properties of Laplace eigenfunctions~\cite{Berard77, Zelditch94, Anantharaman08, AnantharamanNonnenmacher07b, AnantharamanMacia14, 
 DyatlovJin17}. In general, the best one can expect is that such a correspondence 
 remains true up to times of order $c|\log\hbar|$ with $c>0$ which can be expressed in terms of certain expansion rates of the geodesic 
 flow~\cite{BambusiGraffiPaul99, BouzouinaRobert02}. Yet, in the case of the torus (or of $\IR^n$) and of the Weyl quantization, the exact 
 bracket formula allows to prove that, for every $t\in\IR$,
 \begin{equation}\label{e:egorov-torus}
  e^{-\frac{it\hbar\Delta}{2}}\Oph^w(a_h)e^{\frac{it\hbar\Delta}{2}}=\Oph^w(a_h\circ\varphi^t),
 \end{equation}
for every smooth function on $T^*\IT^n$ (or $T^*\IR^n$) all of whose derivatives are bounded.
 \item (G\aa{}rding inequality) Let $a_\hbar\in S^{0}(T^*M)$ such that, for every $0<\hbar\leq 1$, $a_{\hbar}\geq 0$. Then, there exists 
 a constant $C_a>0$ such that
 \begin{equation}\label{e:garding}
  \Oph(a_h)\geq -C_a\hbar.
 \end{equation}
\end{itemize}
These properties are the fundamental rules of our semiclassical quantization and they establish a connection between the 
operations we have on functions (product, Poisson bracket, pullback by the geodesic flow, etc.) and 
the ones on operators (composition, bracket, evolution by the Schr\"odinger flow, etc.).

\subsection{Semiclassical measures}\label{ss:semiclassical-measure}

As a first application of these semiclassical rules, we will define the notion of semiclassical 
measures~\cite{Gerard91b, LionsPaul93, Zworski12} and use it to relate elements of $\ml{N}(\Delta_g)$ with 
elements of $\ml{M}(\varphi^t)$. The idea is to lift the measure
$$\nu_{\hbar}(a)=\int_Ma(x)|\psi_\hbar(x)|^2d\text{Vol}_g(x)$$
into a distribution on $T^*M$ by making use of the quantization we have just defined. Recall that $\psi_{\hbar}$ 
is a solution to~\eqref{e:eigenvalue}, i.e. a Laplace eigenfunction. For such a sequence $(\psi_{\hbar})_{\hbar\rightarrow 0^+}$, 
we set
$$w_{\hbar}:a\in\mathcal{C}^{\infty}_c(T^*M)\mapsto\left\la \psi_{\hbar},\Oph(a)\psi_{\hbar}\right\ra_{L^2}.$$
Thanks to~\eqref{e:calderon}, this defines a \emph{bounded sequence} in $\ml{D}'(T^*M)$ which can be extended to 
more general symbols belonging to $S^0(T^*M)$. In particular, when $a$ is independent of the $\xi$ variable, we recover 
the measure $\nu_{\hbar}$. This kind of quantity measures the distribution of the quantum state $\psi_{\hbar}$ in the 
phase space $T^*M$. It encodes more information than $\nu_{\hbar}$ but it is not anymore a measure. 
\begin{ex}\label{ex:lagrangian} Let us illustrate the content of this new distribution compared with 
$\nu_{\hbar}$. To that aim, consider a sequence of Lagrangian states on $\IR^n$:
$$\forall 0<\hbar\leq 1,\quad \psi_\hbar(x)=\chi(x)e^{\frac{iS(x)}{\hbar}},$$
where $\chi$ and $S$ are smooth functions and where $\chi$ is compactly supported on $\IR^{n}$. The corresponding 
measure $d\nu_{\hbar}(x)=|\chi(x)|^2dx$ is constant and it does not capture the oscillations due to the phase 
factor. Now, if we consider the Wigner distribution
$$a\in\mathcal{C}^{\infty}_c(T^*\IR^n)\mapsto\left\la \psi_{\hbar},\Oph^w(a)\psi_{\hbar}\right\ra_{L^2},$$
then an application of the stationary phase Lemma shows that the limit distribution is given by the measure
$$ a\in\mathcal{C}^{\infty}_c(T^*\IR^n)\mapsto\int_{\IR^n}a(x,d_xS)|\chi(x)|^2dx,$$
which is carried by the Lagrangian manifold $\ml{L}_S:=\{(x,d_xS)\}$ -- see~\cite[Chap.~5]{Zworski12} for details. 
\end{ex}

The appearance of 
this quantity in the literature can at least be traced back to the works of 
Wigner~\cite{Wigner32}. In the mathematics literature, it implicitely appears many years later in the works on quantum 
ergodicity~\cite{Snirelman73, Snirelman74, Zelditch87, ColindeVerdiere85, HelfferMartinezRobert87} and their systematic study was 
initiated by Tartar~\cite{Tartar90} and G\'erard~\cite{Gerard91}.
\begin{def1} Let $\mu$ be an element of $\ml{D}'(T^*M)$. We say that $\mu$ is a semiclassical measure if there exists a 
sequence $(\psi_{\hbar_n})_{\hbar_n\rightarrow 0^+}$ satisfying
$$-\hbar^2\Delta_g\psi_{\hbar_n}=\psi_{\hbar_n},\quad\|\psi_{\hbar_n}\|_{L^2}=1,$$
and such that, for every $a\in\mathcal{C}^{\infty}_c(T^*M)$, 
 $$\lim_{\hbar_n\rightarrow 0^+}\left\la \psi_{\hbar_n},\Op_{\hbar_n}(a)\psi_{\hbar_n}\right\ra_{L^2}=\la\mu,a\ra.$$
We denote by $\ml{M}(\Delta_g)$ the set of semiclassical measures.
\end{def1}
\begin{rema} Using the Calder\'on-Vaillancourt Theorem~\ref{e:calderon}, we can verify that any sequence
$$w_{\hbar_n}: a\in\mathcal{C}^{\infty}_c(T^*M)\mapsto\left\la \psi_{\hbar_n},\Op_{\hbar_n}(a)\psi_{\hbar_n}\right\ra_{L^2},\quad n\geq 1$$
is bounded in $\mathcal{D}'(T^*M)$. Hence, it admits converging subsequences. 
\end{rema}

We have then the following Theorem which among other things justifies the terminology measure:
\begin{theo}\label{t:semiclassical-measure} One has
$$\ml{M}(\Delta_g)\subset\ml{M}(\varphi^t),$$
and
$$\ml{N}(\Delta_g):=\left\{\int_{S^*_xM}\mu(x,d\xi):\mu\in\ml{M}(\Delta_g)\right\}.$$
\end{theo}
Thanks to this Theorem and by passing to the semiclassical limit $\hbar\rightarrow 0^+$, 
we can identify the (Wigner) distributions $(w_{\hbar})_{\hbar\rightarrow 0^+}$ of the 
quantum states $(\psi_{\hbar})_{\hbar\rightarrow 0^+}$ with objects from classical mechanics, i.e. 
invariant measures of the geodesic flow. Moreover, the description of elements inside $\ml{N}(\Delta_g)$ 
can be recovered from the informations we have on the set $\ml{M}(\Delta_g)$. For instance, we can already 
deduce that $\nu\in\ml{N}(\Delta_g)$ cannot be the Dirac measure at some point as the geodesic flow has 
no fixed point on $S^*M$. An obvious question to ask is: how big is the set $\ml{M}(\Delta_g)$ inside 
$\ml{M}(\varphi^t)$? For instance, in the case of the $2$-sphere, one can deduce 
from~\eqref{e:QL-sphere} that these two sets are equal. However, Theorem~\ref{t:zygmund} together with 
example~\ref{ex:torus} show that this is no longer true in the case of the flat torus. We could have 
thought from the correspondance principle that all classical states can be recovered from quantum states 
in the semiclassical limit. Yet, this simple example shows us that this is not true in general and we will 
try to discuss what can be said under various geometric assumptions.

\begin{proof} Let $\mu$ 
be an element of $\ml{M}(\Delta_g)$ issued from a sequence $(\psi_{\hbar})_{\hbar\rightarrow 0^+}$ of Laplace eigenfunctions. 
First, from G\aa{}rding inequality~\eqref{e:garding}, we can verify that $\mu$ is a positive 
distribution, hence a measure~\cite[Ch.~1]{Schwartz66}. Let 
now $a$ be an element in $\ml{C}^{\infty}_c(T^*M)$ whose support does not intersect $S^*M$. Thanks 
to~\eqref{e:composition} and to~\eqref{e:calderon}, one has 
$$\la \psi_{\hbar},\Oph(a)\psi_{\hbar}\ra=\la \psi_{\hbar},\Oph\left(a(2H_0-1)^{-1}\right)(-\hbar^2\Delta_g-1)\psi_{\hbar}\ra
+\ml{O}(\hbar)=\ml{O}(\hbar),$$
where we also used~\eqref{e:eigenvalue}. Hence, by letting $\hbar\rightarrow 0^+$, one finds that $\mu(a)=0$. As this is 
valid for any $a$ whose support does not intersect $S^*M$, this implies that 
the support of $\mu$ is contained in $S^*M$.  In order to show 
that $\mu$ is a probability measure, it is sufficient to verify that there is no escape of mass 
at infinity. To that aim, we reproduce the same argument as before except that we do not pick $a$ to be compactly 
supported but rather in $S^0(T^*M)$. Take now $a$ in $\ml{C}^{\infty}(M)$ and a smooth cutoff function 
$\chi\in\ml{C}^{\infty}(\IR,[0,1])$ which is equal to $1$ on $[-1,1]$ and to $0$ outside $[-2,2]$. One has then
$$\int_Mad\nu_{\hbar}=\la w_{\hbar}, a\chi(H_0)\ra+\la \psi_{\hbar},\Oph(a(1-\chi(H_0)))\psi_{\hbar}\ra,$$
where the second term on the right hand side tends to $0$ from the previous discussion. Hence, by letting 
$\hbar\rightarrow 0^+$, the pushforward of $\mu$ on $M$ is an element of $\ml{N}(\Delta_g)$.

Let now $a$ be an element in $\ml{C}^{\infty}_c(T^*M)$ without any 
assumption on its support. From the eigenvalue equation~\eqref{e:eigenvalue}, one knows that
$$\forall t\in\IR,\quad\la\psi_{\hbar},\Oph(a)\psi_{\hbar}\ra=\la\psi_{\hbar},e^{-\frac{it\hbar\Delta_g}{2}}
\Oph(a)e^{\frac{it\hbar\Delta_g}{2}}\psi_{\hbar}\ra.$$
Combined with~\eqref{e:egorov}, this yields
\begin{equation}\label{e:semiclassical-measure_invariance}
\forall t\in\IR,\quad\la\psi_{\hbar},\Oph(a)\psi_{\hbar}\ra=\la\psi_{\hbar},\Oph(a\circ\varphi^t)\psi_{\hbar}\ra
+\ml{O}_{t,a}(\hbar).
\end{equation}
Letting $\hbar\rightarrow 0^+$, we finally get that $\varphi^t_*\mu=\mu$ for every $t$ in $\IR$. 
\end{proof}
We conclude this background section by mentionning without proof a version of the microlocal Weyl 
law~\cite[Th.~15.3]{Zworski12}:
\begin{theo}[Microlocal Weyl law]\label{t:weyl} Let $(\psi_{\hbar_j})_{j\geq 0}$ 
be an orthonormal basis of Laplace eigenfunctions (i.e. solutions to~\eqref{e:eigenvalue}). 
Then, for every $a$ in $\mathcal{C}^{\infty}_c(T^*M)$,
$$\lim_{\hbar\rightarrow 0}\frac{1}{N_{\hbar}}\sum_{j:\hbar\leq\hbar_j<+\infty}\la\psi_{\hbar_j},\Op_{\hbar_j}(a)\psi_{\hbar_j}\ra=\int_{S^*M}adL_1,$$
where
\begin{equation}\label{e:number-eigenvalue}N_\hbar:=\left|\left\{j\geq 0:\hbar\leq\hbar_j<+\infty\right\}\right|.\end{equation} 
\end{theo}

\begin{rema} When $L_1$ is an ergodic measure, then this Theorem tells us that we have a convex 
combination of distributions (which are asymptotically elements of $\ml{M}(\varphi^t)$) that converges 
to $L_1$ in $\ml{D}'(T^*M)$. Ergodicity exactly means that $L_1$ is an extreme point of $\ml{M}(\varphi^t)$. Thus, 
from an heuristic point of view, this microlocal Weyl law should imply that most of the elements in this convex 
sum  converges to $L_1$. This idea can be made rigorous and it is at the heart of the Quantum Ergodicity 
Theorem that we will 
state (and prove) in section~\ref{s:observability}.
\end{rema}


\section{Semiclassical observability}\label{s:observability}


We are in position to discuss when~\eqref{e:observability} holds, i.e. we want to determine under which 
geometric (or dynamical) conditions, a sequence $(\psi_{\hbar})_{\hbar\rightarrow 0^+}$ 
verifying~\eqref{e:eigenvalue} puts some mass on a given open set of $M$ (or $S^*M$). For that purpose, we can 
observe that a direct consequence of Theorem~\ref{t:semiclassical-measure} is
\begin{prop}\label{p:observability} Let $a$ be an element in $\ml{C}^0(S^*M,\IR)$ and set 
$$A_-:=\lim_{T\rightarrow+\infty}\frac{1}{T}\inf\left\{\int_{0}^Ta\circ\varphi^t(x,\xi)dt:(x,\xi)\in S^*M\right\},$$
and
$$A_+:=\lim_{T\rightarrow+\infty}\frac{1}{T}\sup\left\{\int_{0}^Ta\circ\varphi^t(x,\xi)dt:(x,\xi)\in S^*M\right\}.$$
Then, for every $\mu\in\ml{M}(\Delta_g)$, one has
$$A_-\leq \mu(a)\leq A_+.$$ 
\end{prop}
These bounds appear naturally when studying controllability or stabilization 
questions of certain partial differential equations~\cite{Sjostrand00}. As a Corollary of this Proposition, one gets
\begin{theo}[Uniform observability of Laplace eigenfunctions]\label{t:observability} Let $\omega$ be an open set of $M$ such that, 
for every geodesic $\gamma$ of $(M,g)$, 
$$\gamma\cap\omega\neq\emptyset.$$ Then, there exists some contant $c_{\omega,M,g}>0$ such 
that, for every $\psi_{\lambda}$ solution to
$$-\Delta_g\psi_{\lambda}=\lambda^2\psi_{\lambda},\quad\|\psi_{\lambda}\|_{L^2}=1,$$
 one has 
 $$\int_{\omega}|\psi_\lambda(x)|^2d\operatorname{Vol}_g(x)\geq c_{\omega,M,g}>0.$$
\end{theo}
The geometric assumption of this Theorem is in general sharp thanks to~\eqref{e:QL-sphere}. Yet, we shall see later on that this result can 
be improved under appropriate assumptions on the manifold $(M,g)$ -- see Theorems~\ref{t:observability-torus} 
and~\ref{t:FUP}. Note that the geometric condition appearing in this Theorem is often reffered as a 
geometric control condition which is a terminology coming from control theory~\cite{BardosLebeauRauch92}.
\begin{proof} We proceed by contradiction and suppose that there exists a sequence of solutions $(\psi_{\hbar_n})_{\hbar_n>0}$ 
to~\eqref{e:eigenvalue} such that the conclusion does not hold. Thanks to~\eqref{e:unique-continuation}, one finds that 
$\hbar_n\rightarrow 0^+$. Up to an extraction, we denote by $\mu$ the corresponding semiclassical measure. We also set 
$\pi_1:(x,\xi)\in S^*M \mapsto x\in M$ to be the canonical projection. Let now $(x_0,\xi_0)$ be an element of $S^*M$. Thanks to our 
hypothesis on $\omega$, there exists
some $T_{0}\in\IR$ such that $\pi_1\circ \varphi^{T_0}(x_0,\xi_0)\in \omega$. By continuity, this remains true on a small ball of 
radius $r_0>0$ centered at $(x_0,\xi_0)$. By compactness, we can then extract a finite number of balls 
$(B(x_i,\xi_i;r_i))_{i=1,\ldots N}$ such that this property holds true. Fix now $a$ to be a continuous function on 
$M$ which is compactly supported inside $\omega$ and which is equal to $1$ on the union of all the open sets 
$\pi_1\circ\varphi^{T_i}(B(x_i,\xi_i;r_i))$. One has
$$\int_{\omega}|\psi_{\hbar_n}(x)|^2d\operatorname{Vol}_g(x)\geq\int_{M}a(x)|\psi_{\hbar_n}(x)|^2d\operatorname{Vol}_g(x)=
\la\psi_{\hbar_n},\Op_{\hbar_n}(a)\psi_{\hbar_n}\ra,$$ 
from which we can deduce the expected contradiction thanks to Proposition~\ref{p:observability} 
as $\mu(a)\geq A_->0$.
\end{proof}

\subsection{Improving the observability bounds along density one subsequences}

We will now show how to improve Proposition~\ref{p:observability} along a typical subsequence of 
Laplace eigenfunctions. More precisely, following~\cite{Riviere13} -- see also~\cite{Sjostrand00} for 
earlier related results of Sj\"ostrand in the case of the damped wave equation, one has
\begin{theo}\label{t:weak-QE} Let $(\psi_{\hbar_j})_{j\geq 0}$ 
be an orthonormal basis of Laplace eigenfunctions (i.e. solutions to~\eqref{e:eigenvalue}). 
Then, for every $a$ in $\mathcal{C}^{\infty}_c(T^*M,\IR)$, there exists $S\subset \IZ_+$ such that
\begin{equation}\label{e:density1}\lim_{\hbar\rightarrow 0^+}\frac{1}{N_\hbar}\left|\{j\in S:\hbar\leq\hbar_j<+\infty\}\right|=1,\end{equation}
and
$$\operatorname{ess inf}\{L_{x,\xi}(a)\}\leq 
\liminf_{j\rightarrow +\infty, j\in S}\la\psi_{\hbar_j},\Op_{\hbar_j}(a)\psi_{\hbar_j}\ra\leq
\limsup_{j\rightarrow +\infty, j\in S}\la\psi_{\hbar_j},\Op_{\hbar_j}(a)\psi_{\hbar_j}\ra\leq \operatorname{ess sup}\{L_{x,\xi}(a)\},$$
where 
$$\operatorname{ess sup}\{L_{x,\xi}(a)\}:=\operatorname{inf}\left\{C: L_{x,\xi}(a)\leq C\ \text{for}\ L_1-\text{a.e.}\ (x,\xi)\in S^*M\right\},$$
with $L_{x,\xi}$ being the ergodic component of $(x,\xi)$ w.r.t. $L_1$ and where
$$\operatorname{ess inf}\{L_{x,\xi}(a)\}=-\operatorname{ess sup}\{-L_{x,\xi}(a)\}.$$
\end{theo}
In the following, we will say that a subset $S$ inside $\IZ_+$ is of density $1$ if it satisfies property~\eqref{e:density1}. 
Using the conventions of Proposition~\ref{p:observability}, we note that 
$$A_-\leq\operatorname{ess inf}\{L_{x,\xi}(a)\}\leq \operatorname{ess sup}\{L_{x,\xi}(a)\}\leq A_+,$$
and that these inequalities may be strict. Hence, this Theorem shows that the bounds of Proposition~\ref{p:observability} can be improved along a typical subsequence of 
an orthonormal basis of Laplace eigenfunctions. Before giving some applications of this Theorem in the next section, 
let us give a proof of it. The key inputs compared with Proposition~\ref{p:observability} are the Birkhoff 
ergodic Theorem and the microlocal Weyl law.
\begin{proof} Up to replacing $a$ by $-a$, it is sufficient to prove the upper bound. To that aim, we need to go through 
the proof of the upper bound of Proposition~\ref{p:observability} which was a direct consequence of the invariance 
property proved in Theorem~\ref{t:semiclassical-measure}. 
This Theorem itself relied implicitely on~\eqref{e:semiclassical-measure_invariance} which itself followed from the Egorov 
Theorem~\eqref{e:egorov}. This tells us that, as $j\rightarrow+\infty$,
$$\forall T>0,\quad\la\psi_{\hbar_j},\Op_{\hbar_j}(a)\psi_{\hbar_j}\ra=\left\la\psi_{\hbar_j},\Op_{\hbar_j}\left(\frac{1}{T}\int_0^T
a\circ\varphi^tdt\right)\psi_{\hbar_j}\right\ra+o_{T,a}(1).$$
Arguing as in the proof of Theorem~\ref{t:semiclassical-measure}, we may also suppose without loss of generality that $a$ is 
supported in a small neighborhood of $S^*M$ and that it is $0$-homogeneous in $\xi$ near $S^*M$. We now fix a smooth cutoff 
function $\chi\in\mathcal{C}^{\infty}_c(\IR,[0,1])$ which is equal to $1$ in a small neighborhood of $1/2$. 
Fix also some $\varepsilon>0$ and set $A_0=\operatorname{ess sup}\{L_{x,\xi}(a)\}$. We define a smooth function 
$\tilde{a}_T(x,\xi)\in \ml{C}^{\infty}_c(T^*M)$ such that
\begin{itemize}
 \item $\tilde{a}_T(x,\xi)=\frac{\chi(H_0(x,\xi))}{T}\int_0^T
a\circ\varphi^t(x,\xi)dt$ if $\chi(H_0(x,\xi))\neq 0$ and 
$$\frac{1}{T}\int_0^T
a\circ\varphi^t\left(x,\xi\right)dt\leq A_0+\sqrt{\varepsilon}.$$
 \item $\tilde{a}_T(x,\xi)\leq \chi(H_0(x,\xi))(A_0+2\sqrt{\varepsilon})$ otherwise.
\end{itemize}
Thanks to the Garding inequality~\eqref{e:garding} and for some fixed $T>0$, 
we find that, as $j\rightarrow+\infty$,
$$\la\psi_{\hbar_j},\Op_{\hbar_j}(a)\psi_{\hbar_j}\ra\leq A_0+2\sqrt{\varepsilon} +\left\la\psi_{\hbar_j},\Op_{\hbar_j}\left(\frac{\chi(H_0)}{T}\int_0^T
a\circ\varphi^tdt-\tilde{a}_T\right)\psi_{\hbar_j}\right\ra+o_{T,a}(1).$$
Then, from the Cauchy-Schwarz inequality, one gets
$$\la\psi_{\hbar_j},\Op_{\hbar_j}(a)\psi_{\hbar_j}\ra\leq A_0+2\sqrt{\varepsilon} +\left\|\Op_{\hbar_j}\left(\frac{\chi(H_0)}{T}\int_0^T
a\circ\varphi^tdt-\tilde{a}_T\right)\psi_{\hbar_j}\right\|_{L^2}+o_{T,a}(1).$$
In particular, one has
$$\limsup_{\hbar\rightarrow 0^+}\frac{1}{N_\hbar}\left|\left\{j:\hbar_j\geq\hbar\text{ and }\la\psi_{\hbar_j},\Op_{\hbar_j}(a)\psi_{\hbar_j}\ra\geq A_0+3\sqrt{\varepsilon}\right\}\right|$$
$$\leq\limsup_{\hbar\rightarrow 0^+}\frac{1}{N_\hbar}\left|\left\{j:\hbar_j\geq\hbar\text{ and }\left\|\Op_{\hbar_j}\left(\frac{\chi(H_0)}{T}\int_0^T
a\circ\varphi^tdt-\tilde{a}_T\right)\psi_{\hbar_j}\right\|_{L^2}\geq \sqrt{\varepsilon}\right\}\right|.$$
By Bienaym\'e-Tchebychev inequality, we obtain that
$$\frac{1}{N_\hbar}\left|\left\{j:\hbar_j\geq\hbar\text{ and }\left\|\Op_{\hbar_j}\left(\frac{\chi(H_0)}{T}\int_0^T
a\circ\varphi^tdt-\tilde{a}_T\right)\psi_{\hbar_j}\right\|_{L^2}\geq \sqrt{\varepsilon}\right\}\right|$$
$$\leq \frac{1}{\varepsilon N_\hbar}\sum_{j:\hbar_j\geq\hbar}\left\|\Op_{\hbar_j}\left(\frac{\chi(H_0)}{T}\int_0^T
a\circ\varphi^tdt-\tilde{a}_T\right)\psi_{\hbar_j}\right\|_{L^2}^2.$$
Combining the microlocal Weyl law with the composition formula~\eqref{e:composition}, these two bounds yields 
yields
$$\limsup_{\hbar\rightarrow 0^+}\frac{1}{N_\hbar}\left|\left\{j:\hbar_j\geq\hbar\text{ and }\la\psi_{\hbar_j},\Op_{\hbar_j}(a)\psi_{\hbar_j}\ra\geq A_0+3\sqrt{\varepsilon}\right\}\right|
\leq\frac{1}{\varepsilon}\int_{S^*M}\left(\frac{1}{T}\int_0^T
a\circ\varphi^tdt-\tilde{a}_T\right)^2dL_1.$$ 
From our construction of the function $\tilde{a}_T$, one can deduce that
$$\limsup_{\hbar\rightarrow 0^+}\frac{1}{N_\hbar}\left|\left\{j:\hbar_j\geq\hbar\text{ and }\la\psi_{\hbar_j},\Op_{\hbar_j}(a)\psi_{\hbar_j}\ra\geq A_0+3\sqrt{\varepsilon}\right\}\right|
$$
$$\leq \frac{4\|a\|_{\infty}}{\varepsilon}\text{L}_1\left(\left\{(x,\xi):\frac{1}{T}\int_0^T
a\circ\varphi^t\left(x,\xi\right)dt> A_0+\sqrt{\varepsilon}\right\}\right).$$
As this last inequality is valid for any $T>0$, we can apply the Birkhoff ergodic Theorem to get
$$\lim_{\hbar\rightarrow 0^+}\frac{1}{N_\hbar}\left|\left\{j:\hbar_j\geq\hbar\text{ and }\la\psi_{\hbar_j},\Op_{\hbar_j}(a)\psi_{\hbar_j}\ra\geq A_0+3\sqrt{\varepsilon}\right\}\right|=0.
$$
This gives us a subset $S_{\varepsilon}$ of density $1$ along which one has
$$\limsup_{j\rightarrow +\infty, j\in S}\la\psi_{\hbar_j},\Op_{\hbar_j}(a)\psi_{\hbar_j}\ra\leq \operatorname{ess sup}\{L_{x,\xi}(a)\}+3\sqrt{\varepsilon}.$$
In order to construct a subset $S$ of density $1$ verifying the conclusion of the Theorem, one can for instance follow the 
argument of~\cite{ColindeVerdiere85} to build $S$ out of the subsets $S_{1/k}$ with $k\in\IZ_+^*$ -- see also~\cite[\S 4]{Riviere13}.
\end{proof}
As an application of this Theorem, we can mention the following Corollary~\cite{Riviere13}:
\begin{coro} Let $\Lambda$ be an invariant subset of $S^*M$ such that 
$$L_1=\int_{\Lambda}L_{x,\xi}dL_1(x,\xi).$$
Let 
$(\psi_{\hbar_j})_{j\geq 0}$ be an orthonormal basis of Laplace eigenfunctions (i.e. solutions to~\eqref{e:eigenvalue}). 
Then, there exists $S\subset \IZ_+$ such that any accumulation point (in $\ml{D}'(T^*M)$) of the sequence
$$w_{\hbar_j}:a\in\mathcal{C}^{\infty}_c(T^*M)\mapsto\la\psi_{\hbar_j},\Op_{\hbar_j}(a)\psi_{\hbar_j}\ra,\quad j\in S,$$
belongs to the closure of the convex hull of $\{L_{x,\xi}:(x,\xi)\in\Lambda\}.$
\end{coro}


\subsection{Quantum ergodicity}

We will now apply Theorem~\ref{t:weak-QE} in our three main geometric paradigms: round spheres, flat tori and negatively 
curved manifolds. In the case of the sphere $\IS^d$ endowed with its 
canonical metric, Theorem~\ref{t:weak-QE} is in fact weaker than Proposition~\ref{p:observability} which holds without 
extracting a subsequence. A more interesting case is when the Liouville measure $L_1$ is ergodic, e.g. 
on negatively curved manifolds. In that case, $L_{x,\xi}=L_1$ for 
$L_1$-almost every $(x,\xi)\in S^*M$ and we recover the Quantum Ergodicity Theorem~\cite{Snirelman73, Snirelman74, 
Zelditch87, ColindeVerdiere85}:
\begin{theo}[Quantum Ergodicity]\label{t:QE} Suppose that $L_1$ is ergodic and let $(\psi_{\hbar_j})_{j\geq 0}$ 
be an orthonormal basis of Laplace eigenfunctions (i.e. solutions to~\eqref{e:eigenvalue}). 
Then, there exists $S\subset \IZ_+$ such that
$$\lim_{\hbar\rightarrow 0^+}\frac{1}{N_\hbar}\left|\{j\in S:\hbar\leq\hbar_j<+\infty\}\right|=1,$$
and, for every $a$ in $\mathcal{C}^{\infty}_c(T^*M)$,
$$ 
\lim_{j\rightarrow +\infty, j\in S}\la\psi_{\hbar_j},\Op_{\hbar_j}(a)\psi_{\hbar_j}\ra=\int_{S^*M}adL_1.$$
\end{theo}
In other words, this Theorem states that most eigenfunctions of an orthonormal basis become equidistributed on $S^*M$
in the semiclassical limit. In particular, it says that $L_1$ is an element of $\ml{M}(\Delta_g)$ which may not 
be the case if $L_1$ is not ergodic. The fact that the density $1$ subset can be chosen independently of $a$ requires some additional density 
argument which is for instance explained in~\cite{ColindeVerdiere85} -- see also~\cite[Proof of Th.~15.5]{Zworski12}. This Theorem raises several questions including the question of the size of $S$ inside 
$\IZ_+$ and we will discuss them more precisely in section~\ref{s:QE}. As above, we only stated this 
Theorem in the context of Laplace Beltrami operators on closed manifolds. Yet, it can be extended in many other 
contexts enjoying ergodic features. This includes semiclassical Hamiltonian operators~\cite{HelfferMartinezRobert87}, manifolds 
with boundary~\cite{GerardLeichtnam93, ZelditchZworski96}, quantization of ergodic toral 
automorphisms~\cite{BouzouinaDeBievre96}, discontinuous metrics~\cite{JakobsonSafarovStrohmaier15}, 
sequences of compact hyperbolic 
surfaces~\cite{LeMassonSahlsten17}, subriemannian Laplacians~\cite{ColindeVerdiereHillairetTrelat18}, etc.


In the case of flat tori $\IT^n$, the pushforward of $L_{x,\xi}$ on $\IT^n$ is equal to $dx/(2\pi)^n$ for $L_1$-almost 
every $(x,\xi)$ in $S^*\IT^n$. Hence, we can deduce from Theorem~\ref{t:weak-QE} the following Theorem which was initially 
proved\footnote{The result in~\cite{MarklofRudnick12} also encompasses the case of rational polygons.} 
in~\cite{MarklofRudnick12} -- see also~\cite{Taylor15}.
\begin{theo}\label{t:QE-torus} Suppose that $\IT^n$ is endowed with the Euclidean metric and let $(\psi_{\hbar_j})_{j\geq 0}$ 
be an orthonormal basis of Laplace eigenfunctions (i.e. solutions to~\eqref{e:eigenvalue}). Then, there exists $S\subset \IZ_+$ such that
$$\lim_{\hbar\rightarrow 0^+}\frac{1}{N_\hbar}\left|\{j\in S:\hbar\leq\hbar_j<+\infty\}\right|=1,$$
and, for every $a$ in $\mathcal{C}^{0}(\IT^n)$,
$$ 
\lim_{j\rightarrow +\infty, j\in S}\int_{\IT^n}a(x)|\psi_{\hbar_j}(x)|^2dx=\int_{\IT^n}a(x)\frac{dx}{(2\pi)^n}.$$
\end{theo}
Thanks to the 
examples of paragraph~\ref{sss:torus}, one knows that it is in general necessary to extract a subsequence to get convergence. 
Compared with the case where 
the Liouville measure $L_1$ is ergodic, we emphasize that equidistribution only holds in the configuration space 
$\IT^n$ and not in the phase space $S^*\IT^n$. In fact, if one considers the orthonormal basis given by 
$\left(e^{2i\pi k.x}/(2\pi)^{\frac{n}{2}}\right)_{k\in\IZ^n}$, then an application of 
formula~\eqref{e:weyl-quant} and of the stationary phase Lemma allows to show that any semiclassical measure 
issued from this basis is of the form $\frac{1}{(2\pi)^n}\times\delta_0(\xi-\xi_0)$ where $\xi_0$ is an element 
of $\IS^{n-1}$.

For more general geometric situations, it is not clear what the typical situation should be. In that direction, 
let us mention the following questions due to {\v{S}}nirel'man~\cite[App.]{Lazutkin}. \emph{If} $\varphi^t$ \emph{is 
not ergodic on} 
$S^*M$\emph{, is it true that ``almost all'' eigenfunctions are asymptotically uniformly distributed on ``almost all'' of 
the ergodic components? What conditions are necessary or sufficient for such uniform distribution?} Very little is know on this natural 
question except for a recent work of Gomes~\cite{Gomes18}. In fact, these kind of 
questions also appear in the physics literature and it is sometimes referred as the Percival 
conjecture~\cite{Percival79}. Following some ideas that go back to the works of Einstein~\cite{Einstein17}, 
this conjecture roughly states that quantum eigenmodes split into two families~\cite{Berry77}: regular states that 
concentrate in the part of phase space where the classical flow is regular and irregular ones concentrating where 
the classical flow is chaotic. 

\section{The case of Zoll manifolds}\label{s:zoll}

In the case of the sphere $\IS^d$ endowed with its canonical metric, we already saw that 
\begin{equation}\label{e:scm-sphere}\ml{M}(\Delta_{\IS^d})=\ml{M}(\varphi^t).\end{equation} 
In other words, the set of semiclassical measures 
is as big as it can be. This result was first proved by Jakobson and Zelditch~\cite{JakobsonZelditch99} and 
extended to more general compact rank one symmetric spaces by Maci\`a~\cite{Macia08} -- see also~\cite{HumbertPrivatTrelat16}, 
and to quotients of $\IS^d$ by certain groups of isometries by Azagra and Maci\`a~\cite{AzagraMacia10}. The round sphere 
is in fact the simplest example of a Zoll manifold $(M,g)$, i.e. a manifold all of whose geodesics 
are closed~\cite{Besse78}. In the general context of Zoll manifolds, thanks to a Theorem of Wadsley~\cite[\S 7.B]{Besse78}, there exists some $l>0$ such 
that $\varphi^l=\text{Id}_{S^*M}$, i.e. the geodesic flow is periodic. The existence of nontrivial 
manifolds of this type is due to Zoll~\cite{Zoll03}. Many years later, Guillemin showed that, 
near the round metric $g_0$ on $\IS^2$, there exists an infinite dimensional submanifold of Zoll metrics and he 
characterized its tangent space at $g_0$~\cite{Guillemin76}. 

A natural question to ask is whether~\eqref{e:scm-sphere} remains true on more general Zoll manifolds. The answer is no 
in general~\cite{MaciaRiviere16, MaciaRiviere17}. Before stating a precise result, we need to introduce a couple 
of conventions. We will in fact study the slightly more general framework of Schr\"odinger eigenfunctions, i.e. 
solutions to
\begin{equation}\label{e:eigenvalue-schrodinger}
 -\frac{1}{2}\hbar^2\Delta_g\psi_\hbar+\hbar^2V\psi_{\hbar}=\frac{1}{2}\psi_{\hbar},\quad\|\psi_{\hbar}\|_{L^2}=1,
\end{equation}
where $V$ is a smooth real valued function on $(M,g)$ and $\hbar$ is some positive parameter. As before, we can 
build from these sequences of eigenfunctions some sets of measures $\ml{M}(\Delta_g-2V)$ and 
$\ml{N}(\Delta_g-2V).$ Following the same proof\footnote{The term involving the potential only adds error terms which are of 
order $\ml{O}(\hbar^2)$.} as for Theorem~\ref{t:semiclassical-measure}, 
elements in $\ml{N}(\Delta_g-2V)$ are exactly the pushforwards on $M$ of elements in 
$\ml{M}(\Delta_g-2V)$ and $\ml{M}(\Delta_g-2V)\subset \ml{M}(\varphi^t)$. Moreover, given a smooth function $b$ on 
$T^*M\setminus 0$, one can define its Radon transform
$$\forall (x,\xi)\in T^*M\setminus 0,\quad\ml{R}_g(b)(x,\xi):=\frac{\|\xi\|_x}{l}\int_0^{\frac{l}{\|\xi\|_x}}b\circ\varphi^s(x,\xi)ds.$$
\begin{rema}
With the conventions of the Birkhoff Ergodic Theorem~\ref{t:birkhoff}, $\ml{R}_g(b)(x,\xi)=\mu_{x,\xi}(b)$ on $S^*M$ where $\mu$ is any 
invariant measure on $S^*M$. In the following, we will make a small abuse of notations and denote by $\ml{R}_g(V)$ the function 
$\ml{R}_g(V\circ\pi)$ where $\pi:T^*M\mapsto M$ is the canonical projection.
\end{rema}
Using these conventions, it was proved 
by Maci\`a and the author~\cite{MaciaRiviere16} that semiclassical measures enjoy some additional regularity properties:
\begin{theo}\label{t:zoll} Let $(M,g)$ be a Zoll manifold. Then, 
there exists a smooth function $q_0:T^*M\setminus 0\rightarrow\IR$ depending only on $(M,g)$ which is $0$-homogeneous and 
$\varphi^t$-invariant, and such that, for any smooth function $V:M\rightarrow \IR$, for any $\mu$ in 
$\ml{M}(\Delta_g-2V)$ and for any $a\in\ml{C}^{0}(S^*M)$, one has
$$\forall t\in\IR,\quad\int_{S^*M}a\circ\varphi^{t}_{q_0+\ml{R}_g(V)}d\mu=\int_{S^*M}ad\mu.$$ 
In particular, for every $(x,\xi)\in S^*M$,  
$$d_{x,\xi}(q_0+\ml{R}_g(V))\neq 0\Longrightarrow \mu\left(\{\varphi^t(x,\xi):0\leq t\leq l\}\right)=0.$$
\end{theo}
In the case of $\IS^n$, one has $q_0=0$. This result was recently extended to semiclassical harmonic oscillators by Arnaiz and Maci\`a~\cite{ArnaizMacia19}. 
In terms of observability, the above Theorem has the following consequence:
\begin{coro}[Improved observability on Zoll manifolds]\label{c:zoll} Let $(M,g)$ be a Zoll manifold, let $V:M\rightarrow \IR$ be a smooth function and let $\omega$ be 
an open subset of $M$. Suppose that, for every geodesic $\gamma$ of $(M,g)$, there exists some $t\in\IR$ such that
\footnote{Here, $\gamma$ is identified with its lift on $S^*M$ and $\pi_1:S^*M\rightarrow M$ is the 
restriction of $\pi$ to $S^*M$.}
$$\pi_1\circ \varphi^{t}_{q_0+\ml{R}_g(V)}(\gamma)\cap\omega\neq\emptyset.$$
Then, there exists some contant $c_{\omega,V,M,g}>0$ such 
that, for every $\psi_{\lambda}$ solution to\footnote{Note that $\lambda$ may be complex for finitely many eigenvalues.}
$$(-\Delta_g+2V)\psi_{\lambda}=\lambda^2\psi_{\lambda},\quad\|\psi_{\lambda}\|_{L^2}=1,$$
 one has 
 $$\int_{\omega}|\psi_\lambda(x)|^2d\operatorname{Vol}_g(x)\geq c_{\omega,V,M,g}>0.$$ 
\end{coro}
The proof is similar to the one given for Theorem~\ref{t:observability} and we omit it. Finally, in dimension $2$, 
another application of Theorem~\ref{t:zoll} is the following~\cite{MaciaRiviere16}:
\begin{theo}\label{t:schrodinger-zoll} Suppose that $n=2$ and that the assumptions of Theorem~\ref{t:zoll} 
are verified. Then, any $\nu$ in $\ml{N}(\Delta_g-2V)$ can be decomposed as follows:
 $$\nu = f \operatorname{Vol}_g + \alpha \nu_\mathrm{sing}$$
where $f\in L^1(M,d\operatorname{Vol}_g)$, $\alpha\in[0,1]$ and $\nu_\mathrm{sing}$
belongs to $\ml{N}_{\operatorname{Crit}}(V+q_0)$ which is by definition the closed convex hull 
of the set of probability measures $\delta_\gamma$, where 
$d_{x,\xi}(\ml{R}_g(V)+q_0)=0$ for some $(x,\xi)$ generating the closed geodesic $\gamma$. 
\end{theo}
This Theorem should be read in the light of the case where $V\equiv 0$ and where $\IS^2$ is endowed with its 
canonical metric $g_0$. In that case, we saw that eigenfunctions can concentrate on any closed geodesics and that 
$\ml{N}(\IS^2)=\ml{N}(\Delta_{g_0})$. Here, this result states that 
quantum limits have in fact more regularity when $V$ is nonconstant. We shall not prove this result and we refer 
to~\cite[Cor.~4.4]{MaciaRiviere16} for details.

\subsection{Examples} Before giving the proof of Theorem~\ref{t:zoll}, let us give some examples of Zoll surfaces 
and of smooth potential where these results give nontrivial conclusion.

\subsubsection{Surfaces of revolution with $V\equiv 0$}
Theorem~\ref{t:zoll} shows that Laplace eigenfunctions cannot concentrate on closed geodesics that are not critical points 
of the function $q_0$. In particular, one needs $q_0$ to be nonconstant on $S^*M$ for the Theorem (and its Corollary) 
not being empty when $V\equiv 0$. To that aim, we need to explain 
where the function $q_0$ comes from. As shown in~\cite{MaciaRiviere16}, it can be expressed in terms of the subprincipal symbol 
that appears when the Laplace-Beltrami is put under a quantum normal form as in the works of 
Duistermaat-Guillemin~\cite{DuistermaatGuillemin75} and Colin de Verdi\`ere~\cite{ColindeVerdiere79}. The 
dependence of $q_0$ in terms of the metric is a little bit intricate. It was computed explicitely by 
Zelditch~\cite{Zelditch96, Zelditch97} in the case of Zoll metrics on $\IS^2$ all of whose geodesics 
are of length $2\pi.$ For instance, examples of such metrics can be obtained in the following manner. Let $\sigma$ be a 
smooth odd function on $[-1,1]$ satisfying $\sigma(1)=0$ and such that
$$g_{\sigma}(\theta,\phi)=(1+\sigma(\cos\theta))^2d\theta^2+\sin^2\theta d\phi^2$$
is a smooth metric on $\IS^2$. In~\cite{Besse78}, these manifolds are referred as Tannery surfaces and 
are shown to verify the Zoll property. 
In~\cite{MaciaRiviere16}, we implemented this metric in the formulas obtained by Zelditch and, together with 
Theorem~\ref{t:zoll}, we were able to deduce that
\begin{prop} Suppose that $\sigma'(0)\neq 0$. Then, $q_0$ is nonconstant on $S^*\IS^2$ and one has
$$\ml{M}(\Delta_{g_{\sigma}})\neq\ml{M}(\varphi^t).$$ 
\end{prop}
This result contrasts with~\eqref{e:scm-sphere}. In particular, it shows that the fact that all invariant 
measures are semiclassical measures for $\Delta_g$ is not due to the periodicity of the flow. It is rather 
a consequence of the specific structure of the spectrum of the round sphere, e.g. high multiplicity of the eigenvalues.

\subsubsection{Schr\"odinger operators on the round sphere}

Let us now consider the sphere $\IS^2$ endowed with its canonical metric $g_0$. In that case, $q_0$ is constant 
according to~\eqref{e:scm-sphere} and we can discuss what happens depending on the properties of $V$. Denote by 
$G(\IS^2)\simeq S^*\IS^2/\IS^1$ the space of closed geodesics on $\IS^2$, which is 
a smooth symplectic manifold~\cite{Besse78}. This space $G(\IS^2)$ can be identified with $\IS^2$~\cite[p.~54]{Besse78}. 
This can be easily seen as follows. Take an oriented closed geodesic $\gamma$. It belongs to an unique 
$2$-plane in $\IR^3$ which can be oriented via the orientation of the geodesic, and $\gamma$ can be identified with the 
unit vector in $\IS^2$ which is directly orthogonal to this oriented $2$-plane. Hence, the Radon transform 
$\ml{R}_{g_0}(V)$ can be identified with a function of $G(\IS^2)\simeq\IS^2,$ and the flow $\varphi_{\ml{R}_{g_0}(V)}^t$ 
can be viewed as an Hamiltonian flow on the symplectic manifold $G(\IS^2)$. In that framework, one finds that, if 
$\gamma\in G(\IS^2)$ verifies $d_{\gamma}\ml{R}_{g_0}(V)\neq 0$, then, for every $\nu\in\ml{N}(\Delta_{g_0}-2V)$, 
$\nu(\gamma)=0$. In other words, eigenfunctions of Schr\"odinger operators can only concentrate on closed 
geodesics which are critical points of $\ml{R}_{g_0}(V)$. Recall from the works of Guillemin that $\ml{R}_{g_0}:V\in\ml{C}^{\infty}_{\text{even}}(\IS^2)\mapsto \ml{R}_{g_0}(V)\in\ml{C}^{\infty}_{\text{even}}(\IS^2)$ is an 
isomorphism~\cite{Guillemin76}.  Note also that, if $V$ is an odd function, then its Radon 
transform vanishes but we can still formulate a statement by replacing $\ml{R}_{g_0}(V)$ with some appropriate 
function of $V$~\cite{MaciaRiviere17}.

Corollary~\ref{c:zoll} can also be rewritten as follows. Suppose that
$$K_{\omega,V}:=\left\{\gamma\in G(\IS^d):\forall t\in\IR,\ \varphi_{\ml{R}_{g_0}(V)}^t(\gamma)\cap\omega=\emptyset\right\}=\emptyset.$$
Then, there exists some contant $c_{\omega,V}>0$ such 
that, for every $\psi_{\lambda}$ solution to
$$(-\Delta_{g_0}+2V)\psi_{\lambda}=\lambda^2\psi_{\lambda},\quad\|\psi_{\lambda}\|_{L^2}=1,$$
 one has 
 $$\int_{\omega}|\psi_\lambda(x)|^2d\operatorname{Vol}_g(x)\geq c_{\omega,V}>0.$$
Let us explain how to construct $\omega$ and $V$ such that $K_{\omega,V}=\emptyset$ while $K_{\overline{\omega}}\neq\emptyset$.
Write $\IS^2:=\left\{(x,y,z):x^2+y^2+z^2=1\right\}$ 
and set$$\omega=\left\{(x,y,z):x^2+y^2+z^2=1\ \text{and}\ z>\eps\right\},$$ with $\eps>0$ small enough. In particular, there are infinitely many geodesics 
which belong to $K_{\overline{\omega}}\subset K_{\omega}$, i.e. the geometric control condition $K_\omega=\emptyset$ fails. In the space of geodesics $G(\IS^2)\simeq\IS^2$, the geodesics 
belonging to $K_{\overline{\omega}}$ correspond to a small neighborhood of the two poles $(0,0,-1)$ and $(0,0,1)$ of $\IS^2$. Hence, if one 
chooses $V\in\ml{C}^{\infty}_{\text{even}}(\IS^2)$ in such a way that $\ml{R}_{g_0}(V)$ has no critical points in a slightly bigger 
neighborhood\footnote{This is possible thanks to Guillemin's result.}, then one finds that $K_{\omega,V}=\emptyset$. Indeed, in that case, the 
Hamiltonian flow $\varphi_{\ml{R}_{g_0}(V)}^t$ has no critical point inside $K_{\omega}$. Thus, it transports the uncontrolled geodesics of $K_{\omega}$ 
to geodesics which are geometrically controlled by the geodesic flow. Note that the condition of having no critical points inside 
$K_{\omega}$ is a priori nongeneric among smooth functions.

Finally, in order to illustrate Theorem~\ref{t:schrodinger-zoll}, we note that $\ml{R}_{g_0}(V)$ 
can always be identified with a smooth function on the real projective plane $\IR P^2$. Hence, for a generic 
choice of potential $V$, the set $\ml{N}_{\operatorname{Crit}}(V)$ is the convex hull of finitely many measures 
carried by closed geodesics depending only on $V$.

\subsection{Proof of Theorem~\ref{t:zoll}}

For the sake of simplicity, we will only give the proof in the case of $\IS^d$ endowed with its canonical metric. 
The argument is based on the quantum averaging method that goes back to the works of Weinstein on spectral asymptotics 
of Schr\"odinger operators on the sphere~\cite{Weinstein77} -- see also~\cite{DuistermaatGuillemin75, ColindeVerdiere79} for 
more general geometric frameworks.

\subsubsection{Weinstein's averaging method} First, we need to fix some conventions and to 
collect some well-known facts on the spectral properties of the Laplace-Beltrami operator on $\IS^d$. 
Recall that the eigenvalues of $-\Delta_{g_0}$ are of the form 
$$E_k=\left(k+\frac{d-1}{2}\right)^2-\frac{(d-1)^2}{4},$$
where $k$ runs over the set of nonnegative integer. In particular, we can write
\begin{equation}\label{e:decompose-Laplace}
 -\Delta_{g_0}= A^2-\left(\frac{d-1}{2}\right)^2,
\end{equation}
where $A$ is a selfadjoint pseudodifferential operator of order $1$ with principal symbol $\|\xi\|_x$ and satisfying
\begin{equation}\label{e:period-quantum}
e^{2i\pi A}=e^{i\pi(d-1)}\text{Id}.
\end{equation}
Given $a$ in $\ml{C}^{\infty}_c(T^*\IS^d\setminus 0)$, we then set, by analogy with the Radon transfom of $a$,
$$\ml{R}_{\text{qu}}(\Oph(a)):=\frac{1}{2\pi}\int_0^{2\pi}e^{-is A}\Oph(a)e^{isA}ds.$$
An important observation which is due to Weinstein~\cite{Weinstein77} is that the following exact commutation relation holds:
$$\left[\ml{R}_{\text{qu}}(\Oph(a)),A\right]=0.$$
In particular, from~\eqref{e:decompose-Laplace}, one has
\begin{equation}\label{e:commute}
\left[\ml{R}_{\text{qu}}(\Oph(a)),\Delta_{g_0}\right]=0.
\end{equation}
Finally, a variant\footnote{The principal symbol of the Hamiltonian is now $\|\xi\|$ instead of $\|\xi\|^2/2$.} 
of the Egorov Theorem~\eqref{e:egorov} allows to relate the operator $\ml{R}_{\text{qu}}(\Oph(a))$ to the classical Radon transform as follows:
\begin{equation}\label{e:egorov-sphere}
 \ml{R}_{\text{qu}}(\Oph(a))=\Oph(\ml{R}_{g_0}(a))+\hbar R,
\end{equation}
where $R$ is a pseudodifferential operator in $\Psi^{-\infty}(\IS^d)$.

\subsubsection{Extra invariance properties on $\IS^d$} Let us now apply these properties to derive some invariance properties of the elements in 
$\ml{M}(\Delta_{g_0}-2V)$. We fix $\mu$ in $\ml{M}(\Delta_{g_0}-2V)$ which is generated by a sequence $(\psi_{\hbar})_{\hbar\rightarrow 0^+}$ and 
$a$ in $\ml{C}^{\infty}_c(T^*\IS^d\setminus 0)$. We use the eigenvalue equation to write
$$\left\la \psi_{\hbar},\left[-\frac{1}{2}\hbar^2\Delta_{g_0}+\hbar^2V,\ml{R}_{\text{qu}}(\Oph(a))\right]\psi_{\hbar}\right\ra=0.$$
According to~\eqref{e:commute}, this implies that
$$\left\la \psi_{\hbar},\left[V,\ml{R}_{\text{qu}}(\Oph(a))\right]\psi_{\hbar}\right\ra=0.$$
Combining~\eqref{e:egorov-sphere} with the commutation formula for pseudodifferential operators~\eqref{e:commute} 
and with the Calder\'on-Vaillancourt Theorem~\eqref{e:calderon}, we get
$$\frac{\hbar}{i}\left\la \psi_{\hbar},\Oph\left(\{V,\ml{R}_{g_0}(a)\}\right)\psi_{\hbar}\right\ra=\ml{O}(\hbar^2).$$
Hence, after letting $\hbar$ go to $0$, one finds that
$$\mu(\{V,\ml{R}_{g_0}(a)\})=0.$$
Applying the invariance by the geodesic flow twice, one finally has 
\begin{equation}\label{e:main-result}\mu(\{\ml{R}_{g_0}(V),a\})=\mu(\{\ml{R}_{g_0}(V),\ml{R}_{g_0}(a)\})=0.\end{equation}
This is valid for any smooth test function $a$ in $\ml{C}^{\infty}_c(T^*\IS^d\setminus 0)$. Thus, we have just proved that 
any $\mu$ in $\ml{M}(\Delta_{g_0}-2V)$ 
is invariant by the Hamiltonian flow $\varphi_{\ml{R}_{g_0}(V)}^t$ of $\ml{R}_{g_0}(V)$ which is well defined on 
$S^*\IS^d\subset T^*\IS^d\setminus 0$.


\subsection{Equidistribution on the round sphere}

We saw in the previous sections that $\ml{M}(\Delta_{g})=\ml{M}(\varphi^t)$ when $(M,g)=(\IS^2,g_0)$ is 
the $2$-sphere endowed with its canonical metric. Equivalently, there exists sequences of eigenfunctions that 
concentrate along closed geodesic. In this last paragraph related to Zoll manifolds, we would like to see that 
one can still obtain equidistribution as in Theorem~\ref{t:QE} provided we make appropriate assumptions on 
the sequence of eigenfunctions we consider.

\subsubsection{Equidistribution by probabilistic averaging}
In relation with the random waves model of Berry, 
one can in fact show that typical orthonormal basis of Laplace eigenfunctions are equidistributed on the $2$-sphere 
following the works of Zelditch~\cite{Zelditch92}, Van der Kam~\cite{Vanderkam97} and Burq-Lebeau~\cite{BurqLebeau13}. 
To that aim, recall that 
$$L^2(\IS^2)=\bigoplus_{l=0}^{+\infty}\text{Ker}\left(\Delta_{g_0}+l(l+1)\right),$$
and that each eigenspace $E_l:=\text{Ker}\left(\Delta_{g_0}+l(l+1)\right)$ is of 
dimension $2l+1$. We denote by $S_l$ the unit sphere inside $E_l$. The set $\mathcal{B}_l$ of orthonormal basis inside $E_l$ 
can be identified with the unitary group $U(2l+1)$ which is endowed with a natural probability measure 
given by the Haar measure that we denote here by $\mathbb{P}_l$. Given an orthonormal basis 
$b_l=(e_{-l},\ldots, e_0,\ldots,e_{l})$ in $\mathcal{B}_l$ and $-l\leq m\leq l$, we define the map 
$T_m:b_l\mapsto e_m$. These maps induce measures on $S_l$ that can be identified with the uniform 
measure $\tilde{\mathbb{P}}_l$ on $S_l$. The next Theorem shows 
that a typical orthonormal basis on $S_l$ is equidistributed in the semiclassical limit 
$l\rightarrow+\infty$~\cite{Zelditch92, Vanderkam97, BurqLebeau13}:
\begin{theo}\label{t:proba} Let $a$ be a continuous function on $\mathbb{S}^2$. Then, there exists some constant $C_a>0$ 
such that, for every $l\geq 0$,
$$\mathbb{P}_l\left(\left\{b_l\in \ml{B}_l:\ \sup_{-l\leq m\leq l}
\left|\int_{\IS^2}a|e_m|^2d\operatorname{Vol}_{g_0}-
\int_{\IS^2}a\frac{d\operatorname{Vol}_{g_0}}{\operatorname{Vol}_{g_0}(\IS^2)}\right|\geq 
l^{-\frac{1}{8}}
\right\}\right)\leq C_a e^{-C_a\sqrt{l}}.$$
\end{theo}
Hence, for a generic orthonormal basis inside $\ml{B}_l$, eigenfunctions will tend to be equidistributed as 
$l\rightarrow+\infty$. More precisely, if we endow the set of Laplace orthonormal basis with the product measure 
$\mathbb{P}=\prod_{l\geq 0}\mathbb{P}_l$, then this Theorem implies that, for $\mathbb{P}$-a.e. orthonormal basis of Laplace 
eigenfunctions $(\psi_j)_{j\geq 0}$, one has 
$$\forall a\in\mathcal{C}^0(\IS^2),\ \lim_{j\rightarrow+\infty}\int_{\IS^2}a|\psi_j|^2d
\operatorname{Vol}_{g_0}=
\int_{\IS^2}a\frac{d\operatorname{Vol}_{g_0}}{\operatorname{Vol}_{g_0}(\IS^2)}.$$
In particular, compared with the Quantum Ergodicity results of Section~\ref{s:observability} (e.g. Theorem~\ref{t:QE}), 
one does not need to extract a density $1$ subsequence to obtain equidistribution of Laplace eigenfunctions on $\IS^2$ 
(for a generic choice of orthonormal basis). We also emphasize that no ergodic behaviour is avalaible on $\IS^2$. 
Instead of ergodicity, the averaging property comes from a probabilistic phenomenon called concentration of 
measure~\cite{Ledoux01}. 

\begin{rema} For the sake of simplicity, we stated the Theorem in the case of observables depending only 
on the configuration space but the argument can be extended to more general observables using the notion of 
Wigner distributions as we did before -- see~\cite{Zelditch92, BurqLebeau13} for more details. Note that, 
in some sense, this result gives a kind of negative answer on $\IS^2$ to {\v{S}}nirel'man's 
question raised at the end of Section~\ref{s:observability}.
\end{rema}

\begin{proof} First, without loss of generality, we can suppose that $a$ is real valued and that $\int_{\IS^2}ad\operatorname{Vol}_{g_0}=0$. We have 
then
$$\mathbb{P}_l\left(\left\{b_l\in \ml{B}_l:\ \sup_{-l\leq m\leq l}
\left|\int_{\IS^2}a|e_m|^2d\operatorname{Vol}_{g_0}\right|\geq 
l^{-\frac{1}{8}}
\right\}\right)$$
$$\hspace{4cm}\leq (2l+1)\tilde{\mathbb{P}}_l\left(\left\{e\in S_l:\ \left|\int_{\IS^2}a|e|^2
d\operatorname{Vol}_{g_0}\right|\geq 
l^{-\frac{1}{8}}
\right\}\right)$$
Hence, it is sufficient to show that there exists some constant $C_a>0$ 
such that, for every $l\geq 1$,
\begin{equation}\label{e:key-proba}\tilde{\mathbb{P}}_l\left(\left\{e\in S_l:\ \left|\int_{\IS^2}a|e|^2d\operatorname{Vol}_{g_0}\right|\geq 
l^{-\frac{1}{8}}
\right\}\right)\leq C_a e^{-C_al^{\frac{3}{4}}}.\end{equation}
To that aim, we introduce the random variable 
$$Y_l:e\in S_l\mapsto\int_{\IS^2}a|e|^2d\operatorname{Vol}_{g_0}\in\mathbb{R},$$
which is a Lipschitz map with constant $2\|a\|_{\infty}$. We can now use the main probabilistic argument 
needed for our proof. This is a property of the Lebesgue measure on spheres of large dimension 
which is referred as concentration of the measure~\cite[Ch.~1]{Ledoux01}. This property~\cite[Eq.~(1.13), Th.~(2.3)]{Ledoux01} implies that, for every $r>0$
$$\tilde{\mathbb{P}}_l\left(\left\{e\in S_l:\ \left|Y_l-m_l\right|\geq 
r
\right\}\right)\leq 2e^{-\frac{2l r^2}{(2\|a\|_{\infty})^2}},$$
where $m_l$ is the median of the random variable $Y_l$. Letting $3r=l^{-\frac{1}{8}}$ yields an 
upper of the form $2e^{-\frac{l^{\frac{3}{4}}}{18\|a\|_{\infty}^2}}$. Using this inequality, we find that
$$\left|\int_{S_l}Y_ld\tilde{\mathbb{P}}_l-m_l\right|\leq 
\frac{l^{-\frac{1}{8}}}{3}+4\|a\|_{\infty}e^{-\frac{l^{\frac{3}{4}}}{18\|a\|_{\infty}^2}}.$$
Hence, there exists some constant $C>0$ (depending only on $a$) such that
$$\tilde{\mathbb{P}}_l\left(\left\{e\in S_l:\ \left|Y_l-\int_{S_l}Y_ld\tilde{\mathbb{P}}_l\right|\geq 
l^{-\frac{1}{8}}
\right\}\right)\leq Ce^{-Cl^{\frac{3}{4}}}.$$
It now remains to verify that the expectation of $Y_l$ with respect to $\tilde{\mathbb{P}}_l$ is equal to $0$ and we 
will be done. To see this, we recall that the dimension of $E_l$ is 
equal to $2l+1$ and we write
\begin{eqnarray*}\int_{S_l}Y_ld\tilde{\mathbb{P}}_l &=&
\int_{\IS^2}a(x)\left(\int_{S_l}|e(x)|^2d\tilde{\mathbb{P}}_l(e)\right)
d\operatorname{Vol}_{g_0}(x)\\
&=&\int_{\IS^2}a(x)\left(\frac{1}{2l+1}\int_{\ml{B}_l}\sum_{e\in b}|e(x)|^2d\mathbb{P}_l(b)\right)
d\operatorname{Vol}_{g_0}(x). \end{eqnarray*}
Observe now that $K_l(x,y):=\sum_{e\in b}e(x)\overline{e(y)}$ is the kernel of the orthogonal projector on the 
space $E_l$, hence independent of the choice of $b$. Invariance by rotation also tells us that 
$K_l(Rx,Ry)=K_l(x,y)$ for any element $R\in SO(3)$. Thus, $K_l(x,x)$ is constant and one finds that 
this constant is equal to $\frac{2l+1}{\text{Vol}_{g_0}(\IS^2)}$ as $\text{dim}(E_l)=2l+1$. From this, we 
can deduce that $\int_{S_l}Y_ld\tilde{\mathbb{P}}_l=0$ as we supposed that the average of $a$ is equal 
to $0$. This concludes the proof of Theorem~\ref{t:proba}. 
\end{proof}

\subsubsection{Equidistribution by Hecke averaging}
There is still another way to obtain equidistribution of Laplace eigenfunctions on the canonical $2$-sphere $\IS^2$ 
by introducing operators commuting with the Laplacian. This approach was initiated by B\"ocherer, Sarnak and 
Schulze-Pillot~\cite{BochererSarnakSchulze03} who conjectured that joint eigenfunctions of the Laplace operator on 
the $2$-sphere and of certain Hecke type operators must be equidistributed as $\hbar\rightarrow 0^+$. Hecke 
type operators on the $2$-sphere are defined as follows. For $N\geq 2$, consider a finite set of rotations 
$R_1,\ldots, R_N$ in $SO(3)$ and define the operator
$$T_N\psi(x)=\frac{1}{\sqrt{2N-1}}\sum_{j=1}^{N}(\psi(R_jx)+\psi(R_j^{-1}x)).$$
This operator commutes with $\Delta_{g_0}$. Hence, it makes sense to consider joint eigenfunctions of these 
two operators. When the rotations correspond to certain special elements in an order of a quaternion division algebra, 
then B\"ocherer, 
Sarnak and Schulze-Pillot~\cite{BochererSarnakSchulze03} conjectured that joint eigenfunctions of $T_N$ and $\Delta_{g_0}$ 
must be equidistributed as in Theorem~\ref{t:QE} without extracting a subsequence. This conjecture remains open. In a 
recent work~\cite{BrooksLemassonLindenstrauss16}, Brooks, Le Masson and Lindenstrauss showed that, 
as soon as $(R_1,\ldots,R_N)$ generate a free subgroup, joint eigenfunctions of $T_N$ and $\Delta_{g_0}$ verify the property of 
Theorem~\ref{t:QE-torus}, i.e. they equidistribute on $\IS^2$ provided that we extract a density $1$ subsequence.

\section{The case of the torus}\label{s:torus}

We move on to a case where the geodesic flow is not periodic but where the flow is still quite regular, namely the case of 
flat tori. In that case, we completely described the dynamics of the geodesic flow in paragraph~\ref{sss:torus}. Geodesic flows on 
flat tori are in some sense the simplest examples of nondegenerate completely integrable systems. In fact, $H_0$ can be 
written as $\frac{1}{2}\sum_{j=1}^nf_j^2$ where $f_j(x,\xi)=\xi_j$ and one has $\{f_i,f_j\}=0$ for 
every $(i,j)$ and $df_1\wedge df_2\ldots\wedge df_n\neq 0$. This exactly means that $F=(f_1,f_2,\ldots,f_n)$ is a nondegenerate 
completely integrable system on $T^*\IT^n\setminus 0$. The geodesic flow preserves the level sets of the moment map $F$, and we 
would like to analyze how this regular dynamics influences the structure of the set $\ml{M}(\Delta)$. Another important 
feature of this example is that the quantum Hamiltonian $-\hbar^2\Delta$ can be written as
$$-\hbar^2\Delta=\hat{F}_1^2+\hat{F}_2^2+\ldots+\hat{F}_n^2,$$
where $\hat{F}_j=\Op_\hbar^w(f_j)=-i\hbar\partial_{x_j}$. In particular, $[\hat{F}_i,\hat{F}_j]=0$ for every $(i,j)$ and we say that 
$(\hat{F}_1,\hat{F}_2,\ldots,\hat{F}_n)$ is a quantum completely integrable system. For an introduction on these topics, we 
refer the reader to the book of V{\~u} Ng{\d{o}}c~\cite{VuNgoc06}. Another crucial property of the torus is its 
arithmetic structure that we already encountered in the proof of Theorem~\ref{t:zygmund} for instance. 
In this section, we will try to illustrate both aspects of this problem. First, we will show how an arithmetic approach allows to improve 
Theorem~\ref{t:QE-torus}. Then, we will make use of semiclassical methods and of the complete integrability to improve (in the case of the torus) the 
observability estimate from Theorem~\ref{t:observability}.

\subsection{Improving Theorem~\ref{t:QE-torus} by arithmetic means}

We begin with the following improvement of Theorem~\ref{t:QE-torus} which was proved in~\cite{HezariRiviere18}:
\begin{theo}\label{t:quant-QE-torus} Suppose that $\IT^n$ is endowed with the Euclidean metric. Then, there exists 
some constant $C_n>0$ such that for every orthonormal basis of Laplace eigenfunctions $(\psi_{\hbar_j})_{j\geq 0}$ 
(i.e. solutions to~\eqref{e:eigenvalue}) and, for every $a$ in $L^2(\IT^n)$,
$$\forall\hbar>0,\quad \frac{1}{N_\hbar}\sum_{j:\hbar_j\geq\hbar}\left|\int_{\IT^n}a(x)|\psi_{\hbar_j}(x)|^2dx-\int_{\IT^n}a(x)\frac{dx}{(2\pi)^n}\right|^2\leq C_n\|a\|_{L^2(\IT^n)}^2\hbar.$$
\end{theo}
The fact that this quantity goes to $0$ follows from Theorem~\ref{t:QE-torus} and the improvement here is that we 
can provide an explicit rate of convergence. If we had followed a semiclassical strategy, then the best bound that 
can be obtained seems to be of order $\hbar^{\frac{2}{3}}$~\cite{HezariRiviere16b}. In~\cite{LesterRudnick17}, 
Lester and Rudnick showed how to make use of arithmetic results to improve\footnote{They also prove such a bound 
for moments of order $1$.} this up to the threshold $\hbar$ modulo some 
constant that depends on a certain number of derivatives of $a$. Then, in~\cite{HezariRiviere18}, together with Hezari, 
we showed how to combine both approach to get the above result. The proof being rather elementary, we will present it 
as an illustration of what can be gained from these arithmetic methods.

\begin{proof} For the sake of simplicity, let $a$ be a smooth real valued function. First, it can be remarked that, the functions $\psi_{\hbar_j}(x)$ being trigonometric polynomials, we can write
 $$S(a,\hbar):=\sum_{j:\hbar_j=\hbar}\left|\int_{\IT^n}a(x)|\psi_{\hbar_j}(x)|^2dx
 -\int_{\IT^n}a(x)\frac{dx}{(2\pi)^n}\right|^2
 =\sum_{j:\hbar_j=\hbar}\left|\int_{\IT^n}a_{\hbar}(x)|\psi_{\hbar_j}(x)|^2dx\right|^2,$$
 where 
 $$a_{\hbar}(x)=\sum_{p:1\leq\|p\|\leq2\hbar^{-1}}\widehat{a}_pe^{ip.x}.$$
In some sense, we are microlocalizing in a compact part of phase space as we did before. 
Then, we can perform the same trick as in the proof of Theorem~\ref{t:weak-QE}, i.e. average by the Schr\"odinger 
flow $e^{it\Delta}$. For every $T>0$, it gives us
$$S(a,\hbar)=
\sum_{j:\hbar_j=\hbar}\left|\left\la\psi_{\hbar_j},
\left(\frac{1}{T}\int_0^Te^{-it\Delta}a_\hbar e^{it\Delta}dt\right)\psi_{\hbar_j}\right\ra\right|^2\leq
\sum_{j:\hbar_j=\hbar}\left\|
\left(\frac{1}{T}\int_0^Te^{-it\Delta}a_\hbar e^{it\Delta}dt\right)\psi_{\hbar_j}\right\|^2,$$
where the second inequality follows from the Cauchy-Schwarz inequality. As the trace is independent of the choice of basis, 
this leads to
$$S(a,\hbar)\leq
\sum_{k:\|k\|\hbar=1}\left\|
\left(\frac{1}{T}\int_0^Te^{-it\Delta}a_\hbar e^{it\Delta}dt\right)e_k\right\|^2,$$
where $e_k(x)=e^{2i\pi k.x}$ for $k\in\IZ^n$. We can now write
$$\left(\frac{1}{T}\int_0^Te^{-it\Delta}a_\hbar e^{it\Delta}dt\right)e_k
=\sum_{p:1\leq\|p\|\leq2\hbar^{-1}}\widehat{a}_p\left(\frac{1}{T}\int_0^Te^{it(\|k+p\|^2-\|k\|^2)}dt\right)e_{k+n}.$$
Then, from Plancherel's inequality, we find
\begin{equation}\label{e:plancherel}S(a,\hbar)\leq (2\pi)^n
\sum_{k:\|k\|\hbar=1}\sum_{p:1\leq\|p\|\leq2\hbar^{-1}}|\widehat{a}_p|^2\left|
\frac{1}{T}\int_0^Te^{it(\|k+p\|^2-\|k\|^2)}dt\right|^2.\end{equation}
Up to this point we followed exactly the steps of our proof of Theorem~\ref{t:weak-QE} and we made everything 
explicit thanks to the structure of the torus. In the proof of Theorem~\ref{t:weak-QE}, we then took the limits 
$\hbar\rightarrow 0^+$ and $T\rightarrow+\infty$ (in this order). Here, we proceed differently and we first take the 
limit $T\rightarrow+\infty$ which is possible due to the specific structure of our problem. The 
drawback of doing this is that there will not be any dynamical interpretation of this limit. Yet, the advantage is that it 
will reduce the problem to a diophantine problem that can be handled. More precisely, after letting $T\rightarrow+\infty$, 
one finds
\begin{equation}\label{e:key-variance-torus}S(a,\hbar)\leq
\sum_{p:1\leq\|p\|\leq2\hbar^{-1}}|\widehat{a}_p|^2
\left|\left\{k\in\IZ^n:\ \|k\|\hbar=1\ \text{and}\ \|k\|=\|k+p\|\right\}\right|.\end{equation}
This inequality holds for eigenfunctions corresponding to the eigenvalue $\hbar^{-2}$. 
If we sum over all eigenvalues $\leq\hbar^{-2}$, then we find
\begin{eqnarray*}V(a,\hbar) & := &\frac{1}{N_\hbar}\sum_{j:\hbar_j\geq\hbar}\left|\int_{\IT^n}a(x)|\psi_{\hbar_j}(x)|^2dx-\int_{\IT^n}a(x)\frac{dx}{(2\pi)^n}\right|^2\\
& \leq &\frac{1}{N_{\hbar}}\|a\|_{L^2}^2
\sup_{1\leq\|p\|\hbar\leq 2}\left|\left\{k\in\IZ^n:\ \|k\|\hbar\leq1\ \text{and}\ \|k\|=\|k+p\|\right\}\right|. 
\end{eqnarray*}
Recall from the Weyl law~\eqref{e:weyl} that $N_{\hbar}$ is of order $\hbar^{-n}.$ Moreover, for every $1\leq\|p\|\hbar\leq 2$, 
the number of lattice points we want to count is the number of lattice points $k\in\IZ^n$ such that $\|k\|\hbar\leq 1$ and 
$\la p,p-2k\ra=0$. This means that $p-2k$ lies on some hyperplane orthogonal to $p$. The number of lattice points 
with $\|k\|\hbar\leq 1$ in this hyperplane is of order $\ml{O}(\hbar^{1-n})$ for some constant that can be made 
uniform in $p$ and this concludes the proof of the Theorem.
\end{proof}
\begin{rema}
When $n=2$, the reader can check that we can recover Theorem~\ref{t:zygmund} from~\eqref{e:key-variance-torus} -- 
see~\cite{HezariRiviere18} for details.
\end{rema}

\subsection{Observability on flat tori by semiclassical means}\label{ss:semiclassical-torus}

Theorem~\ref{t:zygmund} told us that quantum limits are absolutely continuous with respect to the Lebesgue measure. 
Another important property is that they are uniformly observable on any nonempty open set $\omega$~\cite{Jaffard90, Jakobson97, BurqZworski03, Macia10, 
AnantharamanMacia14, BurqZworski12, BourgainBurqZworski13}:
\begin{theo}[Unconditional observability on $\IT^n$]\label{t:observability-torus} 
Suppose that $\IT^n$ is endowed with the Euclidean metric and that $\omega$ is a nonempty open set. Then, 
there exists some contant $c_{\omega}>0$ such 
that, for every $\psi_{\lambda}$ solution to
$$-\Delta\psi_{\lambda}=\lambda^2\psi_{\lambda},\quad\|\psi_{\lambda}\|_{L^2}=1,$$
 one has 
 $$\int_{\omega}|\psi_\lambda(x)|^2dx\geq c_{\omega}>0.$$
\end{theo}
In dimension $2$, this Theorem was first proved by Jaffard~\cite{Jaffard90} using the theory 
of lacunary Fourier series of Kahane and it can also be recovered by Jakobson's result 
on the structure of quantum limits in dimension $2$~\cite{Jakobson97}. Still in dimension $2$, 
Burq and Zworski recovered this result by more semiclassical methods~\cite{BurqZworski03, BurqZworski12}. 
Then, Anantharaman and Maci\`a extended this result in higher dimension~\cite{Macia10, AnantharamanMacia14} by making a crucial use of 
the complete integrability of the system. Their proof also allowed to recover the absolute continuity of 
quantum limits on the torus (see Theorem~\ref{t:zygmund}) and to encompass the case of Schr\"odinger operators $-\frac{1}{2}\Delta+V$ 
with $V$ that may be only in $L^{\infty}$ -- see also~\cite{BourgainBurqZworski13} for potentials with low regularity in dimension $2$. 
A key ingredient of Anantharaman-Maci\`a's proof is a second microlocalization procedure that was 
subsequently used for more general completely integrable 
systems~\cite{AnantharamanFermanianMacia15, AnantharamanLeautaudMacia16}. 
Note that second microlocalization procedure were also 
used in previous works of Vasy and Wunsch to study the wavefront set properties of eigenfunctions for general completely integrable systems 
in dimension $2$~\cite{Wunsch08, VasyWunsch09, Wunsch12}.

The end of this section will be devoted to the proof of Theorem~\ref{t:observability-torus} following the strategy of 
Anantharaman and Maci\`a~\cite{Macia10, AnantharamanMacia14}. In dimension $2$, it is not the fastest way to prove 
Theorem~\ref{t:observability-torus}. For instance, a simple arithmetic proof follows from the works of Bourgain and 
Rudnick~\cite[Eq.~(1.13-18)]{BourgainRudnick12} by replacing in that reference the normalized arc-length measure $d\sigma(x)$ by any 
smooth (normalized) density $\rho(x)dx$. Yet our semiclassical approach 
has the advantage that it can be generalized in higher dimensions and that it illustrates 
the classical-quantum correspondence in a rather explicit manner. Compared with~\cite{Macia10, AnantharamanMacia14}, we 
make use of an (elementary) arithmetic argument in the final step in order to simplify the exposition but this is not crucial -- see~\cite[\S 7]{AnantharamanMacia14} 
for a more analytical approach.

In order to prove this Theorem, we already observed that it is sufficient to prove the 
following proposition:
\begin{prop} Suppose that $\IT^2$ is endowed with the Euclidean metric and that $\omega$ is a nonempty open set. Then, 
there exists some contant $c_{\omega}>0$ such that, for every $\mu\in\ml{M}(\Delta)$,
$$\mu(\omega\times\IS^1)\geq c_{\omega}.$$ 
\end{prop}
\begin{proof}
We proceed by contradiction. By compactness and up to a diagonal argument, it is sufficient to suppose 
that there exists some $\mu\in\ml{M}(\Delta)$ such that 
$\mu(\omega\times\IS^1)=0$ and to obtain a contradiction. We denote by $(\psi_{\hbar})_{\hbar\rightarrow 0^+}$ the sequence of 
solutions to~\eqref{e:eigenvalue} used to construct $\mu$. 

Using the conventions of example~\ref{ex:torus}, recall that
$$\Lambda_{\xi}:=\{p\in\IZ^2:p.\xi=0\}.$$
Hence, we can split 
$\IS^1$ in two parts by setting
$$\Omega:=\{\xi\in \IS^1:\Lambda_{\xi}=\{0\}\}.$$
One has then the following decomposition of $\mu$ into two invariant measures: 
$$\mu=\mu|_{\IT^2\times\Omega}+\mu|_{\IT^2\times \IS^1\setminus \Omega}.$$
From the ergodic properties of the geodesic flow on $\IT^2\times\Omega$ described in example~\ref{ex:torus}, 
one knows that $\mu(\omega\times\IS^1)=0$ implies that $\mu|_{\IT^2\times\Omega}=0$. Hence, the semiclassical 
measure is concentrated in the rational direction in momentum space:
$$\mu=\sum_{\xi\in\IS^1:\text{rk}\Lambda_{\xi}=1}\mu|_{\IT^2\times\{\xi\}},$$
which is a countable sum. Now, on the one hand, as $\mu$ is a probability measure on $\IT^2\times\IS^1$, there exists some 
$\xi_0\in\IS^1$ such that $\text{rk}\Lambda_{\xi_0}=1$ and such that $\mu|_{\IT^2\times\{\pm\xi_0\}}$ is nontrivial. 
On the other hand, one knows that $\mu(\omega\times\IS^1)=0$. Still using the conventions of example~\ref{ex:torus}, 
this implies that 
\begin{equation}
 \label{e:contradiction}
 \mu\left((\omega+\IT(\xi_0))\times\{\pm\xi_0\}\right)=0.
\end{equation}
Now, we can fix some smooth function 
$b_{\xi_0}(x)$ such that
\begin{itemize}
 \item it is compactly supported inside the open set $\omega+\IT(\xi_0)$,
 \item its Fourier decomposition is of the form
\begin{equation}\label{e:b-xi0}b_{\xi_0}(x)=\sum_{k\in\Lambda_{\xi_0}}\widehat{b}_ke^{ik.x},\end{equation}
 \item it takes values in $[0,1]$ and it does not identically vanish.
\end{itemize}
In particular, one has
\begin{equation}
 \label{e:contradiction-2}
\int_{\IT^2\times\{\pm\xi_0\}} b_{\xi_0}(x)d\mu(x,\xi)=0.
\end{equation}
In order to get our contradiction, we only need to analyze the semiclassical measure in the direction of $\xi_0$ and to test 
it against functions whose Fourier coefficients are in $\Lambda_{\xi_0}$. To that aim, we need to come back at the quantum level $\hbar>0$ and to analyze 
the semiclassical (Wigner) distribution more precisely:
$$\la w_{\hbar},b_{\xi_0}\ra=\left\la \psi_{\hbar},\Oph^w(b_{\xi_0})\psi_\hbar\right\ra.$$
We introduce a smooth cutoff function $\chi:\IR\rightarrow [0,1]$ which is equal to $1$ near $0$ and to $0$ 
outside a slightly larger neighborhood. We also fix some $R>0$ (that will tend to $+\infty$ in the end) and we rewrite
\begin{eqnarray*}\la w_{\hbar},b_{\xi_0}\ra &= &\left\la \psi_{\hbar},\Oph^w\left(b_{\xi_0}(x)\chi^2\left(\frac{\xi.\xi_0^{\perp}}{R\hbar}
\right)\right)\psi_\hbar\right\ra\\ &+ &\left\la \psi_{\hbar},\Oph^w\left(b_{\xi_0}(x)\left(1-\chi^2\left(\frac{\xi.\xi_0^{\perp}}{R\hbar}
\right)\right)\right)\psi_\hbar\right\ra,\end{eqnarray*}
where $\xi_0^{\perp}$ is the unit vector which is directly orthogonal to $\xi_0$. This operation splits the Wigner distribution 
$w_\hbar$ into two parts. The first part allows to analyze the mass of the eigenfunction in a neighborhood of size 
$R\hbar$ of the $\IR$-vector space $\la\Lambda_{\xi_0}\ra$ generated by $\Lambda_{\xi_0}$. We 
say that we are performing a second microlocalization on the submanifold $\la\Lambda_{\xi_0}\ra$. The second term 
describes what happens away from this submanifold. Again, we will apply the Egorov Theorem which is exact in that 
case -- see~\eqref{e:egorov-torus}. Let us start with the first term
$$\la w_{\hbar}^{\xi_0},b_{\xi_0}\ra:=\left\la \psi_{\hbar},\Oph^w\left(b_{\xi_0}(x)\left(1-\chi^2\left(\frac{\xi.\xi_0^{\perp}}{R\hbar}
\right)\right)\right)\psi_\hbar\right\ra,$$
that we rewrite in the following manner
\begin{eqnarray*}\la w_{\hbar}^{\xi_0},b_{\xi_0}\ra &= &\sum_{k\in\Lambda_{\xi_0}}\widehat{b}_k\left\la \psi_{\hbar},
\Oph^w\left(e^{ik.x}\left(1-\chi^2\left(\frac{\xi.\xi_0^{\perp}}{R\hbar}
\right)\right)\right)\psi_\hbar\right\ra\\
& = &\sum_{k\in\Lambda_{\xi_0}}\widehat{b}_k\left\la \psi_{\hbar},\Op_{R^{-1}}^w\left(e^{ik.x}\left(1-\chi^2\left(\xi.\xi_0^{\perp}
\right)\right)\right)\psi_\hbar\right\ra\\
& = &\sum_{k\in\Lambda_{\xi_0}}\widehat{b}_k\left\la \psi_{\hbar},
\Op_{R^{-1}}^w\left(\left(\frac{1}{R}\int_0^Re^{ik.(x+t\xi)}dt\right)\left(1-\chi^2\left(\xi.\xi_0^{\perp}
\right)\right)\right)\psi_\hbar\right\ra,
\end{eqnarray*}
where the last equality follows from the exact Egorov formula~\eqref{e:egorov-torus} applied with $\hbar=R^{-1}$ and 
$t=R$. For $k\neq 0$ and $\xi$ in the support of $(1-\chi(\xi.\xi_0^{\perp}))$, there exists some constant $c>0$ 
such that $|k.\xi|\geq c\|k\|$. Hence, from the Calder\'on-Vaillancourt Theorem~\eqref{e:calderon}, one has
$$\left\|\Op_{R^{-1}}^w\left(\left(\frac{1}{R}\int_0^Re^{ik.(x+t\xi)}dt\right)\left(1-\chi\left(\xi.\xi_0^{\perp}
\right)\right)\right)\right\|_{L^2\rightarrow L^2}\leq C\sum_{|\alpha|\leq D}
R^{-\frac{|\alpha|}{2}}\frac{2\|k\|^{|\alpha|}}{cR\|k\|}C_{\chi},$$
where $C_{\chi}>0$ depends only on a finite number of derivatives of $\chi$. 
Implementing this in our computation of $\la w_{\hbar}^{\xi_0},b_{\xi_0}\ra$ and using that $|\widehat{b}_k|=\ml{O}(\|k\|^{-\infty})$, 
we find that
\begin{equation}\label{e:2-microlocal-reduction}\la w_{\hbar},b_{\xi_0}\ra = \widehat{b}_0+\left\la \psi_{\hbar},\Oph^w\left(\left(b_{\xi_0}(x)-\widehat{b}_0\right)\chi^2\left(\frac{\xi.\xi_0^{\perp}}{R\hbar}
\right)\right)\psi_\hbar\right\ra+\ml{O}(R^{-1}).\end{equation}
In some sense, we have been able to use one more time the averaging by the geodesic flow as, thanks to our truncation, 
the momentum variable $\xi$ remained far from the direction $\xi_0^{\perp}$ where $b_{\xi_0}$ has all its Fourier 
coefficients. We now set 
$$\la b_{\xi_0}\ra(x)=b_{\xi_0}(x)-\widehat{b}_0=\sum_{k\in\Lambda_{\xi_0}\setminus 0}\widehat{b}_ke^{ik.x},$$
and we want to study the properties of
\begin{eqnarray*}\la w_{\hbar,\xi_0},b_{\xi_0}\ra &:= &\left\la \psi_{\hbar},\Oph^w\left(\left\la b_{\xi_0}\right\ra(x)\chi^2\left(\frac{\xi.\xi_0^{\perp}}{R\hbar}
\right)\right)\psi_\hbar\right\ra\\
& =&\left\la \psi_{\hbar},\Op_{R^{-1}}^w\left(\left\la b_{\xi_0}\right\ra(x)\chi^2\left(\xi.\xi_0^{\perp}
\right)\right)\psi_\hbar\right\ra.
\end{eqnarray*}
In other words, we study the properties of the Wigner distribution $w_{\hbar}$ in a neighborhood of size $\sim R\hbar$ of 
$\la\Lambda_{\xi_0}\ra.$ Here, we will not be able to use invariance by the geodesic flow: we have to proceed 
differently. Thanks to the composition rule~\eqref{e:composition} for pseudodifferential operators, one first observes that
$$\la w_{\hbar,\xi_0},b_{\xi_0}\ra=
\int_{\IT^2}\left\la b_{\xi_0}\right\ra(x)\left|\chi\left(-\frac{i\xi_0^{\perp}.\partial_x}{R}\right)\psi_{\hbar}(x)
\right|^2dx+\ml{O}(R^{-1}).$$
Let us now decompose $\chi\left(\frac{ -i\xi_0^{\perp}.\partial_x}{R}\right)\psi_{\hbar}(x)$, i.e.
\begin{equation}\label{e:truncated-eigenfunction}\chi\left(-\frac{i\xi_0^{\perp}.\partial_x}{R}\right)\psi_{\hbar}(x)=\sum_{k:\|k\|\hbar=1}
\chi\left(\frac{k.\xi_0^{\perp}}{R}\right)\widehat{\psi}_{\hbar}(k)e^{ik.x}.\end{equation}
This means that we have a trigonometric polynomial whose coefficients lie on two arcs of lenght $\asymp R$ 
(as $\hbar\rightarrow 0^+$) on the circle of radius $\hbar^{-1}$ centered at $0$. There are in fact at most $4$ 
lattice points having this property. 
\begin{rema} This last observation follows from a classical geometric fact due to Jarnik~\cite{Jarnik21}. Let $k_1$, $k_2$ and $k_3$ be three 
distinct elements of $\IZ^2$ lying on a circle of radius $T$. They form a triangle whose area $\ml{A}$ 
is at least equal to $1/2$. Then, one has the classical geometric formula relating the area of a triangle with the radius $T$ 
of its circumscribed triangle:
$$2T\leq4\ml{A}T=\|k_1-k_2\|\|k_2-k_3\|\|k_3-k_1\|\leq\max\{\|k_i-k_j\|\}^3.$$
 In particular, there at most two lattice points on an arc of length $(2T)^{1/3}$.
\end{rema}
Hence, there are at most four terms in the sum~\eqref{e:truncated-eigenfunction}. Thus, using Riemann-Lebesgue lemma and
up to an extraction, we can suppose that $\left|\chi\left(\frac{\la -i\partial_x,\xi_0^{\perp}\ra}{R}\right)\psi_{\hbar}(x)
\right|^2dx$ converges weakly (as $\hbar\rightarrow 0^+$) to some finite (nonnegative) measure of the form
$$d\nu_{\xi_0,R}(x)=\left(\widehat{\nu}_0(\xi_0,R)+\widehat{\nu}_l^+(\xi_0,R)\cos(l.x)
+\widehat{\nu}_{l}^-(\xi_0,R)\sin(l.x)\right)dx,$$
where $\|l\|\neq 0$ is a lattice point of length $\ml{O}(R)$ and $\widehat{\nu}_{l}^*(\xi_0,R)\in\IR$ for every $l$. Then, up to another 
extraction, this measure converges (as $R\rightarrow+\infty$) to some finite and nonnegative measure of the form
$$d\nu_{\xi_0}(x)=\left(\widehat{\nu}_0(\xi_0)+\widehat{\nu}_l^+(\xi_0)\cos(l.x)
+\widehat{\nu}_{l}^-(\xi_0)\sin(l.x)\right)dx,$$
with $l$ that may be different from before. We now pass to the limits $\hbar\rightarrow0^+$ and $R\rightarrow+\infty$ (in 
this order) in~\eqref{e:2-microlocal-reduction}:
$$\int_{\IT^2\times\IS^1}b_{\xi_0}(x)d\mu(x,\xi)=\widehat{b}_0(1-\nu_{\xi_0}(\IT^2))+\int_{\IT^2}b_{\xi_0}(x)d\nu_{\xi_0}(x).$$
As we already observed it, $\mu$ is carried by the rational directions in the $\xi$-variable. Hence,
$$\sum_{\xi:\text{rk}\Lambda_{\xi}=1}\int_{\IT^2\times\{\xi\}}b_{\xi_0}(x)d\mu(x,\xi)=\widehat{b}_0(1-\nu_{\xi_0}(\IT^2))+\int_{\IT^2}b_{\xi_0}(x)d\nu_{\xi_0}(x).$$
Using the dynamical properties of the geodesic flow described in example~\ref{ex:torus} 
and the $\varphi^t$-invariance of $\mu|_{\IT^2\times\{\xi\}}$, we can deduce that there are only two terms in the sum on the left 
hand side, i.e.
$$0=\int_{\IT^2\times\{\pm\xi_0\}}b_{\xi_0}(x)d\mu(x,\xi)=\widehat{b}_0(1-\nu_{\xi_0}(\IT^2))+\int_{\IT^2}b_{\xi_0}(x)d\nu_{\xi_0}(x).$$
As $\nu_{\xi_0}$ is a nonnegative measure, as $b_{\xi_0}\geq 0$ and as $\widehat{b}_0\neq 0$, we deduce that $\nu_{\xi_0}$ is a 
probability measure. Moreover, we know that its density is a trigonometric polynomial which contradicts the fact 
that $\int_{\IT^2}b_{\xi_0}(x)d\nu_{\xi_0}(x)=0$. Hence, one cannot have $\mu(\omega\times\IS^1)=0$ which 
concludes the proof of the Proposition.
\end{proof}

\section{The ergodic case}\label{s:QE}

We will now consider the ergodic case which has attracted the more attention since the seminal works 
of {\v{S}}nirel'man~\cite{Snirelman73, Snirelman74}, Zelditch~\cite{Zelditch87} and Colin de Verdi\`ere~\cite{ColindeVerdiere85} notably through the Quantum Unique 
Ergodicity Conjecture of Rudnick and Sarnak~\cite{RudnickSarnak94}. Questions in that direction are part of 
the more general question of Quantum Chaos which tries to describe the manifestation of the chaotic 
features of a system from classical mechanics on its quantum counterpart. This kind of 
problems can be traced back to the early days of quantum mechanics where the question of the quantization of 
nonintegrable systems was already raised~\cite{Einstein17}.

\subsection{Quantum (unique) ergodicity}

Recall from Theorem~\ref{t:QE} that, if the Liouville measure is ergodic, then, for any orthonormal basis 
of Laplace eigenfunctions, most of the eigenfunctions are equidistributed inside $S^*M$ in the 
semiclassical limit. This raises naturally the question of determining if there could be other accumulation points. 
In~\cite{ColindeVerdiere85}, Colin de Verdi\`ere states the following:

\emph{Le probl\`eme le plus int\'eressant 
est de savoir si on peut s'affranchir de la condition d'extraire une sous-suite de densit\'e $1$: quelles sont 
les valeurs d'adh\'erence de la suite des $w_{\hbar_j}$ pour la convergence vague ? Ces mesures sont n\'ecessairement 
invariantes par le flot g\'eod\'esique mais il y a une foule de telles mesures singuli\`eres (par exemple, si $\gamma$ 
est une g\'eod\'esique ferm\'ee). Il me semble que ces mesures pour des g\'eod\'esiques ferm\'ees des surfaces de 
Riemann \`a courbure $-1$ ne peuvent pas \^etre des limites vagues d'une suite $\psi_{\hbar_j}$. De toutes les fa\c{c}ons, 
il y a d'autres mesures port\'ees par des ensembles de dimension de Hausdorff $>1$.}

A few years later, Colin de Verdi\`ere and Parisse showed that Laplace eigenfunctions cannot 
concentrate too fast on a closed hyperbolic geodesic of $\varphi^t$~\cite{ColindeVerdiereParisse94}. Recall that a 
closed hyperbolic geodesic $\gamma$ of period $T_{\gamma}$ is a closed orbit of the geodesic flow 
$\varphi^t:S^*M\rightarrow S^*M$ such that $d\varphi^{T_\gamma}|\gamma$ has $1$ as an eigenvalue of multiplicity $1$ and such that all its 
other eigenvalues are of modulus $\neq 1$. This result was generalized by Toth-Zelditch~\cite{TothZelditch03} 
and Christianson~\cite{Christianson07, Christianson11}. More precisely, following these references, one has
\begin{theo}\label{t:cdv-parisse} Let $\gamma$ be closed hyperbolic geodesic of $\varphi^t$ and let 
$a\in\ml{C}^{\infty}_c(T^*M,[0,1])$ which is equal to $0$ near $\gamma$ and to $1$ outside a 
slightly bigger neighborhood (depending on $\gamma$). Then, there exists some $c_{a,\gamma}>0$ such that, for any sequence 
$(\psi_{\hbar})_{\hbar\rightarrow 0^+}$ of solutions to~\eqref{e:eigenvalue},
$$\liminf_{\hbar\rightarrow 0^+}|\log\hbar|\la\psi_{\hbar},\Oph(a)\psi_{\hbar}\ra\geq c_{a,\gamma}>0.$$ 
\end{theo}
Colin de Verdi\`ere and Parisse also showed that this Theorem is sharp by constructing an explicit example where 
a sequence of eigenfunctions concentrate on a closed hyperbolic geodesic at this rate~\cite{ColindeVerdiereParisse94}.
Note that this Theorem does not require anything on the global properties of the flow. In particular, the Liouville measure is 
not necessarly ergodic and it may happen that the Quantum Ergodicity property is not satisfied on the manifold. Below, we will give a 
proof of this result which is valid for more general invariant hyperbolic subsets satisfying some appropriate 
smallness conditions that can be expressed in terms of classical dynamical quantities.

At the same period, this kind of questions was also addressed by Rudnick and Sarnak in the context of arithmetic 
surfaces~\cite{RudnickSarnak94}. In that case, one can consider a subset $\ml{M}_{\text{Hecke}}(\Delta_g)$ 
of $\ml{M}(\Delta_g)$ which is made of semiclassical measures obtained through a sequence of joint eigenfunctions 
of $\Delta_g$ and of the Hecke operators. Recall that, on an arithmetic surface, 
one can define a family of operators $(T_m)_{m\in\IZ_+}$ acting on $L^2(M)$ and commuting with the Laplacian. 
They are called Hecke operators and, as they commute with the Laplacian, it is then natural 
to consider a joint basis of eigenfunctions and to expect that they have stronger properties. 
In that direction, Rudnick and Sarnak proved~\cite{RudnickSarnak94}:
\begin{theo}\label{t:rudnicksarnak} Let $M$ be a compact arithmetic surface and let $\mu$ be an element of 
$\ml{M}_{\operatorname{Hecke}}(\Delta_g)$. Then, for any closed orbit $\gamma$ of the geodesic flow, one has
$$\mu(\gamma)=0.$$ 
\end{theo}
This result relies on arithmetic methods and it shows that joint eigenfunctions of $\Delta_g$ and of the Hecke 
operators cannot concentrate on closed orbits of the geodesic flow. In the same article, Rudnick and Sarnak 
also conjectured that a much stronger property holds:
\begin{conj}[Quantum Unique Ergodicity]\label{co:QUE} Let $(M,g)$ be a negatively curved manifold. Then,
\begin{equation}\label{e:QUE}\ml{M}(\Delta_g)=\{L\}.\end{equation}
\end{conj}
This conjecture remains widely open even if important progress have been made over the last twenty years. 
The end of these notes is devoted to a description of these results. The level of technicality being much more important 
than in the previous sections, we will mostly (but not always) give sketch of proofs in order to illustrate some of the main 
ideas behind these results. We refer the reader to the literature for more details.

\subsection{Exceptional subsequences}

We start with negative results which illustrate the 
importance of the geometric assumptions for the result to hold. First, we will begin with simple toy models 
of quantum mechanics which indicate that~\eqref{e:QUE} cannot be only a consequence of the hyperbolic 
proprerties of the geodesic flow on negatively curved manifolds. Then, we will show that ergodicity of $L_1$ 
is not enough to get~\eqref{e:QUE}.

\subsubsection{Quantum maps}

A popular model in Quantum Chaos is the one of quantum maps. The idea is to consider a 
compact symplectic manifold $(M_0,\omega_0)$ and a (quantizable) symplectomorphism $T_0$ acting on it. The difference 
with the geodesic flow acting on $T^*M$ is that now the phase space is compact which makes the analysis in some 
sense easier even if Laplace eigenfunctions were asymptotically concentrated on a compact part of phase space. 
The simplest example is the one given by a symplectic linear map $A$ in $SL_{2n}(\IZ)$ acting on 
the torus $\IR^{2n}/\IZ^{2n}$ endowed with its canonical symplectic structure~\cite{BouzouinaDeBievre96}. Examples 
of such maps are 
$$A:=\left(\begin{array}{cc} 2 & 1\\
      1 & 1
     \end{array}\right)\quad\text{and}\quad A:=\left(\begin{array}{cc} B & 0\\
      0 & (B^T)^{-1}
     \end{array}\right),$$
where $B\in SL_n(\IZ).$ As before, we shall denote by $\ml{M}(A)$ the set of $A$-invariant measures, e.g. the 
Lebesgue on $\IR^{2n}/\IZ^{2n}$ or the Dirac measure at a fixed point of $A$. If $1$ is not an eigenvalue of $A$, 
there is a natural way to quantize the symplectic dynamical system $(\IR^{2n}/\IZ^{2n},A)$~\cite{BouzouinaDeBievre96}. 
For every positive integer $N=(2\pi\hbar)^{-1}$, this gives rise to a finite dimensional Hilbert space $\ml{H}_N$ of 
dimension $N^n$ on which one can define a quantization procedure $$a\in\ml{C}^{\infty}(\IR^{2n}/\IZ^{2n})\mapsto
\Oph^w(a)\in\ml{L}(\ml{H}_N)$$ satisfying the same properties as in 
paragraph~\ref{ss:semiclassical}. Instead of quantizing the geodesic flow, we are now quantizing $A$ which yields 
some unitary matrix $\widehat{A}_N$ acting on $\ml{H}_N$ and verifying an Egorov Theorem which is exact, i.e.
$$\forall a\in\ml{C}^{\infty}(\IR^{2n}/\IZ^{2n}),\quad \widehat{A}_N^{-1}\Oph^w(a)\widehat{A}_N=\Oph^w(a\circ A).$$ 
Eigenfunctions of the Laplacian are replaced by the eigenvectors of $\widehat{A}_N$ and, if $A$ is ergodic (e.g. if $A$ is an hyperbolic matrix), one can prove 
an analogue of the Quantum Ergodicity Theorem~\cite{BouzouinaDeBievre96} and raise similar questions. As 
above, we can define
\begin{itemize}
 \item the set $\ml{M}(\widehat{A})$ of accumulation points of the Wigner distributions as $N\rightarrow+\infty$. Again, one has 
 $\ml{M}(\widehat{A})\subset\ml{M}(A)$.
 \item the subset $\ml{M}_{\text{Hecke}}(\widehat{A})$ of (Hecke) semiclassical measures which are obtained through a sequence 
 of joint eigenfunctions of $\widehat{A}_N$ and of a certain group of isometries associated 
 with $A$~\cite{KurlbergRudnick00}.
\end{itemize}
If the matrix $A$ is hyperbolic, the dynamical situation is close to the one of negatively curved manifolds (ergodicity, 
strong mixing properties) and one can ask if the Quantum Unique Ergodicity property is true in that framework. In 
dimension $2$, Kurlberg and Rudnick proved that Hecke eigenfunctions satisfy this property~\cite{KurlbergRudnick00}. On 
the other hand,  De Bi\`evre, Faure and Nonnenmacher constructed explicit examples of non-Hecke eigenfunctions for which 
this property fails~\cite{FaureNonnenmacherDeBievre03}. More precisely, these two results read
\begin{theo}\label{t:catmap} Suppose that
 $$A=\left(\begin{array}{cc} 2 & 1\\
      1 & 1
     \end{array}\right).$$
 Then, one has 
 \begin{enumerate}
  \item $\ml{M}_{\text{Hecke}}(\widehat{A})=\{\operatorname{Leb}\},$
  \item $\frac{1}{2}\delta_0+\frac{1}{2}\operatorname{Leb}\in\ml{M}(\widehat{A})$. In particular, 
  $$\ml{M}(\widehat{A})\neq\{\operatorname{Leb}\}.$$
 \end{enumerate}
\end{theo}
This Theorem shows that the chaotic properties of $A$ (e.g. its hyperbolicity) are not enough to obtain Quantum 
Unique Ergodicity. In particular, if Conjecture~\ref{co:QUE} is true, then it has to use other features of 
geodesic flow and of the Laplace eigenfunctions.
\begin{rema}\label{r:FNDB}
The counterexample obtained in~\cite{FaureNonnenmacherDeBievre03} 
is constructed through a subsequence of semiclassical parameters $\hbar_N$ for which $\widehat{A}_N$ has large 
spectral multiplicity in its spectrum. More precisely, $\widehat{A}_N^{T_N}=\text{Id}_{\ml{H}_N}$ for some 
$T_N\simeq 2\ln N/\chi_{\max}\ll\text{dim}\ml{H}_N$ where $\chi_{\max}$ is the positive Lyapunov exponent of $A$. This periodicity 
is exploited by the authors to propagate a coherent state $|0,N\ra$ microlocalized in a small ball of radius 
$N^{-1/2}$ centered at the fixed point $0$ of $A$. Then, they set
$$\psi_N:=\sum_{k=-T_N/2}^{T_N/2}\widehat{A}_N^k|0,N\ra.$$
The time $T_N/4$ is referred as the Ehrenfest time in the physics literature and it corresponds to the scale of times 
for which the semiclassical approximation remains valid. Here, it means that, on the interval $[-T_N/4,T_N/4]$, the propagated 
coherent state remains microlocalized near $0$ which is a fixed point of $A$. This part of the coherent state will 
contribute to the part $\frac{1}{2}\delta_0$ of the semiclassical measure. On the interval $[T_N/4,T_N/2]$ (similarly on 
$[-T_N/2,-T_N/4]$), the coherent state will begin to spread out over the torus and this will lead to the part 
$\frac{1}{2}\text{Leb}$ of the measure. Note that the choice of the fixed point $0$ is rather arbitrary and we 
would choose a point $x_0$ belonging to some periodic orbit of $A$. In that case, the singular part of the measure 
will be the Dirac measure along this orbit.
\end{rema}

Regarding Remark~\ref{r:FNDB}, it could be thought that spectral degeneracies would be the reason for the failure 
of Quantum Ergodicity Theorem. In fact, the Hecke operators of Kurlberg and Rudnick 
allow to reduce the multiplicity of the eigenvalues by considering a joint spectrum and this leads to an unique 
semiclassical measure. However, the situation is not as simple as was shown by Kelmer in the case of higher 
dimensional tori~\cite{Kelmer10}:
\begin{theo}\label{t:kelmer} Suppose that $$A=\left(\begin{array}{cc} E & F\\
      G & H
     \end{array}\right)$$ is an hyperbolic and symplectic matrix with distinct eigenvalues.
Suppose also that $EF^T=GH^T=0\ \operatorname{mod}\ 2$. Then, a necessary and sufficient condition     
for having
$$\ml{M}_{\operatorname{Hecke}}(\widehat{A})=\{\operatorname{Leb}\}$$
is that there is no $A$-invariant subspace $\{0\}\neq E_0\subset\mathbb{Q}^{2n}$ which is isotropic\footnote{It means that 
the restriction of the symplectic form to $E_0$ vanishes.} with respect to the symplectic form. Moreover, for any $A$-invariant subspace $E_0\subset\mathbb{Q}^{2n}$ which is isotropic with respect to 
the symplectic form, $$\operatorname{Leb}_{X_0}\in\ml{M}_{\operatorname{Hecke}}(\widehat{A}),$$
  where
  $$X_0:=\left\{x\in\IR^{2n}/\IZ^{2n}:\ \forall p\in E_0\cap\IZ^{2n},\ p.x\in\IZ\right\}.$$
\end{theo}

\begin{rema}
In order to illustrate this Theorem, one can take the example of $$A=\left(\begin{array}{cc} B & 0\\
      0 & (B^T)^{-1}
     \end{array}\right),$$ 
where $B$ is an hyperbolic matrix in $SL_n(\IZ)$. In that case, the Theorem does not exactly apply as eigenvalues 
may have multiplicity $\geq 2$. Yet, by a direct calculation, one can still verify that
$\ml{M}_{\operatorname{Hecke}}(\widehat{A})\neq\{\operatorname{Leb}\}$ 
and that $\text{Leb}_{\IR^n/\IZ^n\times\{0\}}\in\ml{M}_{\operatorname{Hecke}}(\widehat{A})$.
\end{rema}
This Theorem shows that Hecke eigenfunctions can also fail to equidistribute. 
Moreover, it gives examples of semiclassical measures that are 
concentrated on a submanifold of $\IR^{2n}/\IZ^{2n}$. In particular, in the case of (chaotic) quantum maps, 
one cannot obtain an improvement of Proposition~\ref{p:observability} that would state observability for any open 
set $\omega\neq\emptyset$.

\subsubsection{Bunimovich stadium}

In their Quantum Unique Ergodicity conjecture~\ref{co:QUE}, Rudnick and Sarnak carefully stated that it should hold on negatively 
curved manifolds. This allows to exclude the case of certain nonpositively curved manifolds where the 
Liouville measure is ergodic for the geodesic flow. In that case, motivated by numerical experiments on the 
related model of the Bunimovich stadium, it was in fact widely believe that Quantum Unique Ergodicity 
should fail. Let us recall this model. For later purpose, we fix some parameter\footnote{The choice of the interval $[1,2]$ is arbitrary.} $\omega\in[1,2]$ and we define 
the Bunimovich stadium as the following domain of $\IR^2$:
$$D_\omega:=\left([-\omega,\omega]\times [-\pi/2,\pi/2]\right)\cup\{(x,y)\in\IR^2:(x\pm \omega)^2+y^2\leq \pi^2/4\}.$$
The analogue of the geodesic flow is the billiard flow $\varphi_\omega^t$ acting on $D_\omega\times\IS^1$ with Descartes' laws for the 
reflection at the boundary. This flow preserves the Lebesgue measure $L_{\omega}$ on $D_\omega\times\IS^1$ which is ergodic thanks to the 
work of Bunimovich~\cite{Bunimovich79}. In fact, such a flow falls in the category of nonuniformly hyperbolic flows (as the 
geodesic flow on nonpositively curved surfaces). One can then consider the 
eigenvalue problem
\begin{equation}\label{e:dirichlet}-\hbar^2\Delta\psi_{\hbar}=\psi_{\hbar},\quad\psi_{\hbar}|_{\partial D_\omega}=0,\quad\|\psi_{\hbar}\|_{L^2}=1,\end{equation}
and define the set of semiclassical measures $\ml{M}(\Delta_{D_\omega})\subset\ml{M}(\varphi_{\omega}^t)$~\cite{GerardLeichtnam93}. In 
that case, G\'erard-Leichtnam and Zelditch-Zworski proved that the Quantum Ergodicity property holds~\cite{GerardLeichtnam93, ZelditchZworski96}. Moreover, 
it was proved in~\cite{GerardLeichtnam93, HassellZelditch04} that, for a sequence of eigenfunctions 
$(\psi_\hbar)_{\hbar\rightarrow 0^+}$ verifying the Quantum Ergodicity property, then the boundary values 
$|\hbar\partial_n\psi_{\hbar}(s)|^2d\sigma(s)$ become equidistributed on $\partial D_\omega$, where $\sigma$ is the Lebesgue 
measure on $\partial D_\omega$. In that framework, Hassell proved the following~\cite{Hassell10}:
\begin{theo}\label{t:Hassell} For a.e. $\omega\in[1,2]$, one has
$$\ml{M}(\Delta_{D_\omega})\neq\{L_{\omega}\}.$$ 
\end{theo}
In particular, this Theorem shows that ergodicity is not enough to expect Quantum Unique Ergodicity to hold. This 
result was generalized by Hassell and Hillairet to other contexts where the manifold with (or without) boundary has a 
rectangular (or cylindrical) part with various boundary conditions~\cite{Hassell10}.

\begin{proof}[Sketch of Proof] The idea is to apply arguments from analytic perturbation theory. More precisely, one 
can order the eigenvalues $\hbar_j(\omega)^{-2}$ in a nondecreasing order and verify that they are piecewise smooth. Then, one of 
the key observation is the Hadamard variational formula which states that, for every $\omega$ where it makes sense,
$$\hbar_j(\omega)^2\frac{d}{d\omega}\left(\hbar_j(\omega)^{-2}\right)
=-\int_{\partial D_{\omega}}\rho_{\omega}|\hbar_j(\omega)\partial_n\psi_{\hbar_j}(\omega)|^2d\sigma$$
where $(\psi_{\hbar_j}(\omega))_{j\geq 1}$ is the corresponding orthonormal basis and where 
$$\rho_{\omega}(x,y)=\text{sgn}(x)n_{\omega}^1(x,y),$$ with 
$n_{\omega}(x,y)=n_{\omega}^1(x,y)\partial_x+n_{\omega}^2(x,y)\partial_y$ being the normal 
to $\partial D_\omega$ at $(x,y)$. Hence, according to~\cite{GerardLeichtnam93, HassellZelditch04}, if we 
suppose that $\ml{M}(\Delta_{D_\omega})=\{L_{\omega}\}$ for some $\omega\in[1,2]$, then the boundary values 
are equidistributed and we can deduce that
$$\limsup_{j\rightarrow+\infty }\hbar_j(\omega)^2\frac{d}{d\omega}\left(\hbar_j(\omega)^{-2}\right)\leq -c'<0,$$
for some $c'$ independent of $\omega$.
Now, from the Weyl law~\eqref{e:weyl}, one knows that $h_j(\omega)\approx c(\omega)j^{-1/2}$ for some positive $c(\omega)>0$. This 
implies that, for some uniform constant $c>0$ and for $j$ large enough, 
\begin{equation}\label{e:hassell}\frac{d}{d\omega}\hbar_j(\omega)^{-2}\leq -cj.\end{equation} Introduce now
$$G:=\left\{\omega\in[1,2]:\ml{M}(\Delta_{D_\omega})=\{L_{\omega}\}\right\},$$
and, for $n\in\IZ_+^*$ and $\alpha>0$
$$N_\omega(\alpha,n):=\left|\left\{j:\hbar_j(\omega)^{-2}\in[n^2-\alpha,n^2+\alpha)\right\}\right|.$$
Writing things a little bit informally, one has
$$\int_GN_\omega(\alpha,n)d\omega\lesssim\sum_{c_1 n^2\leq j\leq c_2n^2}
\text{Leb}\left(\left\{\omega\in G:\hbar_j(\omega)^{-2}\in[n^2-\alpha,n^2+\alpha)\right\}\right),$$
for some positive constant $c_1,c_2$ related to the constant appearing in the Weyl law. Now from~\eqref{e:hassell}, 
one knows that $\hbar_j(\omega_2)^{-2}-\hbar_j(\omega_1)^{-2}\gtrsim n^2(\omega_1-\omega_2)$ which yields
$$\int_GN_\omega(\alpha,n)d\omega\lesssim\sum_{c_1 n^2\leq j\leq c_2n^2}\frac{2\alpha}{n^2}=\ml{O}_{\alpha}(1),$$
where the involved constant is uniform in $n$ even if we can note that it depends on $\alpha$. If the Lebesgue measure of 
$G$ was positive, then it would mean that, on a set of positive measure inside $G$, $N_\omega(\alpha,n_k)$ is uniformly 
bounded along a certain subsequence of integer $n_k\rightarrow+\infty$ -- see~\cite[p.~613]{Hassell10} for details. 
In other words, for any $\alpha$, one can find some $\omega\in G$ and some sequence of integers $n_k\rightarrow+\infty$ 
such that the number of eigenvalues of $\Delta_{D_{\omega}}$ inside $[n_k^2-\alpha,n_k^2+\alpha]$ is 
uniformly bounded as $k\rightarrow +\infty$ by some constant $M_\alpha$.

Following earlier arguments of Heller-O'Connor~\cite{OConnorHeller88} and Zelditch~\cite{Zelditch04} -- see also~\cite[\S 1]{Hassell10}, 
this leads to a contradiction with the fact that $\ml{M}(\Delta_{D_\omega})=\{L_{\omega}\}$. To see this, we define
$$\psi_{\hbar_k}(x,y):=\sin(n_k y)\chi(x)\ \text{if $n_k$ is even and }\psi_{\hbar_k}(x):=\cos(n_k y)\chi(x)\ \text{otherwise},$$
where $\chi$ is a smooth function which is compactly supported in $[-1,1]$ and which is chosen such that 
$\|\psi_{\hbar_k}\|_L^2=1$. This does not define a sequence of solutions to~\eqref{e:dirichlet} but, letting $\hbar_k=n_k^{-1}$, 
one has the quasimode equation
$$\|-\hbar_k^2\Delta\psi_{\hbar_k}-\psi_{\hbar_k}\|\leq \frac{\alpha\hbar_k^2}{2}.$$
Even if it is not an eigenmode, one can write the spectral decomposition of the Laplacian to get
$$\left\|\mathbf{1}_{[1-\alpha\hbar_k^2,1+\alpha\hbar_k^2]}(-\hbar_k^2\Delta_{D_\omega})\psi_{\hbar_k}\right\|_{L^2}^2\geq 
\frac{3}{4}.$$
In particular, there exist some exact eigenmodes $\tilde{\psi}_{k}$ such that 
$|\la\psi_{\hbar_k},\tilde{\psi_k}\ra|\geq (3/4M_{\alpha})^{1/2}$ for every $k\geq 1$. We denote by $h_k^{-2}$ 
the corresponding eigenvalue. Following example~\ref{ex:lagrangian}, one can now compute the semiclassical measure of the quasimodes 
$(\psi_{\hbar_k})_{k\geq 1}$. Their limit 
density is given by $$\mu(x,y,\xi,\eta)=c|\chi(x)|^2\delta_0(\xi,\eta\pm 1),$$ 
for some positive constant $c>0$. Such states are called \emph{bouncing ball modes}, in the sense that they are 
localized on the closed orbits of the billiard flow located in the rectangular part of the domain. Now 
take a symbol $a$ which is supported in a neighborhood of the suppport of $\mu$ (meaning near the vertical direction) and 
which is equal to $1$ near this support. Hence, for an appropriate choice of such $a$, one has, for $k$ large enough,
$$\la\psi_{\hbar_k},\Op_{h_k}(|1-a|^2)\psi_{\hbar_k}\ra\leq\frac{3}{16M_\alpha}.$$
We now write
\begin{eqnarray*}|\la\tilde{\psi}_{k},\Op_{h_k}(|a|^2)\tilde{\psi}_{k}\ra| &= &\|\Op_{h_k}(a)\tilde{\psi}_{k}\|_{L^2}^2+o(1)\\
 &\geq &|\la\Op_{h_k}(a)\tilde{\psi}_{k},\psi_{\hbar_k}\ra|^2+o(1)\\
 &\geq &|\la\tilde{\psi}_{k},\psi_{\hbar_k}\ra-\la\tilde{\psi}_{k},(\text{Id}-\Op_{h_k}(a))\psi_{\hbar_k}\ra|^2+o(1)\\
 & \geq &\left(\sqrt{\frac{3}{4M_\alpha}}-|\la\psi_{\hbar_k},\Op_{h_k}(|1-a|^2)\psi_{\hbar_k}\ra|^{\frac{1}{2}}\right)^2+o(1)\geq\frac{3}{16M_{\alpha}}+o(1), 
\end{eqnarray*}
where the first and the last line are consequences of the composition rule~\eqref{e:composition} for 
pseudodifferential operators and where we use the Cauchy-Schwarz inequality on the second and the last line. 
This shows that the sequence $(\tilde{\psi}_k)_{k\geq 1}$ puts asymptotically some mass 
in the vertical direction and this yields the expected contradiction as $D_{\omega}$ is supposed to satisfy the 
Quantum Unique Ergodicity property.
\end{proof}

\subsection{Towards quantum unique ergodicity}

Now that we have described these several obstructions to Quantum Unique Ergodicity, we can come back 
to the case of negatively curved manifolds where this conjecture was initially stated. 
Before discussing the current state of the art, we need 
to begin with the introduction of tools from hyperbolic dynamical systems.

\subsubsection{Hyperbolic dynamics, entropy and pressure}

In the following, we fix a Riemannian metric on $S^*M$, e.g. the Sasaki metric $g_S$ induced by $g$~\cite{Ruggiero07}. 
One can define the notion of hyperbolic subset~\cite{AbrahamMarsden78}:
\begin{def1} Let $\Lambda$ be a compact subset of $S^*M$ which is invariant by the geodesic flow $\varphi^t$. 
We say that $\Lambda$ is \textbf{hyperbolic} (for the flow $\varphi^t$) 
if there exist some $C>0$, $\chi>0$ and a familly of subspaces $E_u(\rho),E_s(\rho)\subset T_{\rho}S^*M$ 
(for every $\rho=(x,\xi)\in\Lambda$) satisfying the following properties, for every $\rho$ in $\Lambda$ and for every $t\geq 0$,
\begin{enumerate}
 \item $T_{\rho}S^*M=\IR X_{H_0}(\rho)\oplus E_u(\rho)\oplus E_s(\rho)$ with 
 $X_{H_0}(\rho)=\frac{d}{d}(\varphi^t(\rho))\rceil_{t=0}$,
 \item $d_{\rho}\varphi^t E_{u/s}(\rho)=E_{u/s}(\varphi^t(\rho))$,
 \item for every $v$ in $E_u(\rho)$, $\|d_{\rho}\varphi^{-t}v\|\leq Ce^{-\chi t}\|v\|$,
 \item for every $v$ in $E_s(\rho)$, $\|d_{\rho}\varphi^{t}v\|\leq Ce^{-\chi t}\|v \|$.
\end{enumerate}
If $\Lambda=S^*M$, we say that the flow $\varphi^t$ has the \textbf{Anosov property}.
\end{def1}
These assumptions roughly tell us that two points which are close will tend to separate 
from each other at an exponential rate. In that sense, the dynamical system $\varphi^t|_{\Lambda}$ is very sensitive 
to perturbations and it is a good model of a classical chaotic system. Note that the simplest example of such an 
hyperbolic subset is given by an hyperbolic closed geodesic as in Theorem~\ref{t:cdv-parisse}.
The main example where $\Lambda=S^*M$ is the geodesic flow on a negatively curved manifold. In that case, 
the flow has strong chaotic properties: ergodicity and mixing of the Liouville measure~\cite{Anosov67}. In that framework, 
the Liouville measure was also shown to be exponentially mixing~\cite{Liverani04}.


We can define the \textbf{unstable Jacobian} at a point $(x,\xi)\in S^*M$  as follows\footnote{The choice of the 
time $1$ and of the metric are arbitrary and the reader can check that it will not change the following.}:
$$J^u(x,\xi):=\left|\det\left(d_{\varphi^1(x,\xi)}\varphi^{-1}\rceil_{E^u(\varphi^1(x,\xi))}\right)\right|,$$
where we endow the unstable spaces $E^u(x,\xi)$ and $E^u(\varphi^1(x,\xi))$ with the Sasaki metric in order to 
defined the Jacobian. It defines a continuous function on $\Lambda$~\cite{KatokHasselblatt95}. For every $\eps>0$ and every $T\geq0$, we also define the 
Bowen ball centered at $(x,\xi)\in S^*M$
$$B_g(x,\xi;\eps,T):=\left\{(x',\xi')\in S^*M:\ \forall 0\leq t\leq T,\ d_g(\varphi^t(x,\xi),\varphi^t(x',\xi'))< \eps\right\},$$
where $d_g$ is the Riemannian distance induced by $g$ on $S^*M$. These are exactly the points of $S^*M$ whose 
trajectories remain close to the one of $(x,\xi)$ up to time $T$. A subset $F\subset \Lambda$ is said to be 
$(\eps,T)$-separated if for every $(x,\xi)$ and $(x',\xi')$ in $F$, $(x',\xi')\in B_g(x,\xi;\eps,T)$ implies 
$(x,\xi)=(x',\xi')$. Our goal is now to measure the size of an hyperbolic set 
$\Lambda$ by taking into account the dynamical properties of the flow. More precisely, for every $0\leq s\leq 1$, we define 
the \textbf{topological pressure} of the set $\Lambda$ 
(with respect to $-s\log J^u$) as follows~\cite{Walters82}:
$$P_{\text{top}}(\Lambda,s):=
\lim_{\eps\rightarrow 0}\limsup_{T\rightarrow +\infty}\frac{1}{T}\log\sup_{F}
\left\{\sum_{\rho\in F}\exp\left(s\int_0^T\log J^u\circ\varphi^t(\rho)dt\right)\right\},$$
where the supremum is taken over all the $(\eps,T)$-separated sets $F$. In order to understand 
the meaning of this function, we can note that the quantity 
$\exp\left(s\int_0^T\log J^u\circ\varphi^t(\rho)dt\right)$ roughly measures the Riemannian volume (to the power $s$) of the 
Bowen ball centered at $(x,\xi)$. Hence, this quantity looks a little bit like the ones used to define the Hausdorff 
dimension of a set~\cite{Pesin97} except that we consider dynamical balls rather than classical balls. 
\begin{ex}
A simple example 
of calculation is given by a closed orbit $\gamma$ of the geodesic flow of minimal period $T_{\gamma}$. 
In that case, one can verify that $P_{\text{top}}(\Lambda,s)=s\log|\det d\varphi^{-T_{\gamma}}|_{E^u(\gamma)}|$. 
\end{ex}
One can show that the map 
$s\mapsto P_{\text{top}}(\Lambda,s)$ is continuous and convex~\cite{Ruelle78}. In the case $s=0$, this defines a quantity 
which is referred as the \textbf{topological entropy} of $\Lambda$ that is denoted by $h_{\text{top}}(\Lambda)\geq 0$. In 
the case of a manifold with constant negative (sectional) curvatures equal to $-1$, one has 
$P_{\text{top}}(\Lambda,s)=h_{\text{top}}(\Lambda)-s(n-1).$ 

If we now fix an invariant probability measure $\mu$ on $\Lambda$, we can also define a related quantity which is called the 
\textbf{Kolmogorov-Sina\u{\i} entropy}~\cite{Walters82}. This quantity was originally defined by using partition 
of the space~\cite{Kolmogorov58, Kolmogorov59, Sinai59} and 
ideas coming from information theory~\cite{Shannon48}. We will introduce it via an equivalent definition due to Brin and 
Katok~\cite{BrinKatok83}. Precisely, they showed that, if $\mu$ is ergodic, the following limit exists for $\mu$-a.e. 
$(x,\xi)$ on $\Lambda$:
$$h_{KS}(\mu)=\lim_{\eps\rightarrow 0}\limsup_{T\rightarrow+\infty}\frac{1}{T}\log\mu(B(x,\xi;\eps,T)),$$
and it is equal to the entropy as originally defined by Kolmogorov and Sina\u{\i}. It measures in some sense the complexity 
of the geodesic flow from the point of view of the ergodic measure $\mu$. More generally, if $\mu$ is not ergodic, 
the Kolmogorov-Sina\u{\i} entropy is defined as
$$h_{KS}(\mu):=\int_{\Lambda}h_{KS}(\mu_{x,\xi})d\mu(x,\xi),$$
where $\mu=\int_{\Lambda}\mu_{x,\xi}d\mu(x,\xi)$ is the ergodic decomposition of $\mu$ -- see Remark~\ref{r:ergodic}. 
The entropy map $\mu\mapsto h_{KS}(\mu)$ is an affine function, i.e. 
$h_{KS}(t\mu_1+(1-t)\mu_2)=th_{KS}(\mu_1)+(1-t)h_{KS}(\mu_2)$ for every $t\in[0,1]$ and 
every $(\mu_1,\mu_2)\in\ml{M}(\varphi^t)$. 
\begin{ex}In 
the case of the Lebesgue measure along a closed orbit, one can verify that the entropy is equal to $0$. From the volume 
Lemma of Bowen and Ruelle~\cite{BowenRuelle75}, one also knows that, for an Anosov geodesic flow, 
$h_{KS}(L_1)=-\int_{S^*M}\log J^u dL_1.$ 
\end{ex}
More generally Ruelle proved that, for any invariant measure 
on $\Lambda$, one has~\cite{Ruelle78}
\begin{equation}\label{e:ruelle}
 h_{KS}(\mu)\leq-\int_{\Lambda}\log J^ud\mu,
\end{equation}
with equality in the case of Anosov flow if and only if $\mu=L_1$~\cite{LedrappierYoung85}.

\begin{rema}\label{r:variational}
The topological pressure is related to the Kolmogorov-Sina\u{\i} entropy via the variational principle~\cite{Walters82}:
$$P_{\text{top}}(\Lambda,s)=\text{sup}\left\{h_{KS}(\mu)+s\int_{\Lambda}\log J^ud\mu: \mu\in\ml{M}(\varphi^t)\ \text{and}\ \mu(\Lambda)=1\right\}.$$
In particular, the topological entropy is the supremum of the Kolmogorov-Sina\u{\i} entropy over all invariant measures on $\Lambda$. 
Combined with Ruelle's inequality~\eqref{e:ruelle}, this variational principle shows that 
$PP_{\text{top}}(\Lambda,1)\leq 0$. One can in fact prove that the Bowen's equation
$$P_{\text{top}}(\Lambda,s)=0$$
admits an unique solution $s_{\Lambda}\in[0,1]$. In analogy with the definition of the Hausdorff dimension, this number 
$s_{\Lambda}$ measures in some sense the size of $\Lambda$ along the unstable direction. 
\end{rema}

\subsubsection{The arithmetic case}

Using this notion of Kolmogorov-Sina\u{\i} entropy, Bourgain and Lindenstrauss refined Theorem~\ref{t:rudnicksarnak} 
in the following manner~\cite{BourgainLindenstrauss03}:
\begin{theo}\label{t:bourgainlindenstrauss} Let $M$ be a compact arithmetic surface. Then, there exists some constant 
$h_0>0$ such that, for every $\mu\in\ml{M}_{\operatorname{Hecke}}(\Delta_g)$ and for $\mu$-a.e. $(x,\xi)\in S^*M$, 
one has 
$$h_{KS}(\mu_{x,\xi})\geq h_0>0.$$
In particular, if $\Lambda$ is an hyperbolic subset such that $h_{\operatorname{top}}(\Lambda)<h_0$, then 
$$\mu(\Lambda)=0.$$
\end{theo}
Hence, not only Hecke eigenfunctions cannot concentrate on closed geodesics, but also they cannot concentrate on subsets 
with small topological entropy. In~\cite{Lindenstrauss06}, Lindenstrauss proved that elements 
$\mu\in\ml{M}_{\text{Hecke}}(\Delta_g)$ verify an extra invariance property that he calls \emph{Hecke}-recurrent and that 
follows from the fact that they are Hecke eigenfunctions. We denote by $\ml{M}_{\text{Hecke}}(\varphi^t)$ 
the subset of Hecke-recurrent measure in $\ml{M}(\varphi^t)$. Lindenstrauss proved the following general 
classification result for measures in $\ml{M}_{\text{Hecke}}(\varphi^t)$ which
yields an answer to the Quantum Unique Ergodicity conjecture for \emph{Hecke} Laplace eigenfunctions~\cite{Lindenstrauss06}:
\begin{theo}\label{t:lindenstrauss} Let $M$ be a compact arithmetic surface and let 
$\mu\in\ml{M}_{\operatorname{Hecke}}(\varphi^t)$. If $h_{KS}(\mu_{x,\xi})>0$ for $\mu$-a.e. $(x,\xi)\in S^*M$, then $\mu=L_1$. 
In particular, from Theorem~\ref{t:bourgainlindenstrauss},
$$\ml{M}_{\operatorname{Hecke}}(\Delta_g)=\{L_1\}.$$
\end{theo}
This Theorem can be generalized to arithmetic surfaces of finite volume that may be noncompact like the modular 
surface $M=\mathbb{H}^2/PSL_2(\IZ)$~\cite{Lindenstrauss06, Soundararajan10}. More recently, it was also 
extended to the case where we consider Laplace eigenfunctions which are eigenfunctions of only \emph{one} (and not all) 
Hecke operator~\cite{BrooksLindenstrauss14}. We will not give proofs of these results which rely on techniques rather 
different from the semiclassical ones presented in these lectures. We refer the interested reader to the 
above references for more details on these arithmetic aspects of the Quantum Unique Ergodicity conjecture.

\subsubsection{Entropic results in variable curvature}\label{sss:entropy}

In the general framework, progress have also been made using the notion of Kolmogorov-Sina\u{\i} entropy. 
This was pioneered by Anantharaman and Nonnenmacher~\cite{Anantharaman08, AnantharamanNonnenmacher07b} and 
it lead to quantitative dynamical constraints on the set of semiclassical measures. More precisely, 
Anantharaman proved~\cite{Anantharaman08}
\begin{theo}\label{t:anantharaman} Suppose that $\varphi^t$ has the Anosov property. 
Then, for every
 $$h_0<\frac{1}{2}\max_{S^*M}|\log J^u|,$$ there exists some constant $c(h_0)>0$ such that, 
 for every $\mu\in\ml{M}(\Delta_g)$, one has
$$\mu\left(\{(x,\xi)\in S^*M:h_{KS}(\mu_{x,\xi})\geq h_0\}\right)\geq c(h_0)>0.$$
In particular, one has
$$h_{KS}(\mu)>0,$$
and, for any invariant subset $\Lambda$,
$$h_{\operatorname{top}}(\Lambda)\leq h_0\Longrightarrow\mu(\Lambda)\leq 1-c(h_0)<1.$$
%

\end{theo}
This Theorem says that Laplace eigenfunctions cannot concentrate on hyperbolic subsets which are too small in 
terms of topological entropy. For instance, it tells us that, if $\gamma$ is a closed orbit of $\varphi^t$, then 
$\mu(\gamma)<1$ for every semiclassical measure $\mu$. Compared with Theorem~\ref{t:bourgainlindenstrauss}, 
this result is valid in full generality for Anosov geodesic flows without any constraints on the sequence 
of Laplace eigenfunctions that one chooses. Of course, it leads to a weaker conclusion. We shall give some 
ideas of this proof in paragraph~\ref{ss:proof-pressure}.

In a subsequent work 
with Nonnenmacher~\cite{AnantharamanNonnenmacher07b}, they developped a slightly different approach based on an entropic 
uncertainty principle~\cite{MaassenUffink88} and they obtained more explicit lower bounds on $h_{KS}(\mu)$. For instance, 
in the case of \emph{constant negative curvature} equal to $-1$, they obtained the following explicit lower bound
$$\forall\mu\in\ml{M}(\Delta_g),\quad h_{KS}(\mu)\geq\frac{n-1}{2},$$
which is half of Ruelle's upper bound~\eqref{e:ruelle}. Showing the Quantum Unique Ergodicity conjecture 
would then amount to remove the $1/2$ in this lower bound thanks to the works of Ledrappier and Young~\cite{LedrappierYoung85}. 
This argument was optimized by Anantharaman, Koch and Nonnenmacher and it lead to the following general bound~\cite{AnantharamanKochNonnenmacher09}:
\begin{theo}\label{t:anantharamankochnonnenmacher} Suppose that $\varphi^t$ has the Anosov property. Set 
$$\chi_{\max}:=\lim_{t\rightarrow+\infty}\frac{1}{t}\log\sup\{\|d_{x,\xi}\varphi^t\|:(x,\xi)\in S^*M\}.$$
 Then, for every $\mu\in\ml{M}(\Delta_g)$, one has
$$h_{KS}(\mu)\geq -\int_{S^*M}\log J^ud\mu-\frac{n-1}{2}\chi_{\max}.$$
\end{theo}
Theorem~\ref{t:anantharaman} gave a rather unexplicit lower bound $h_0c(h_0)$ on the Kolmogorov-Sina\u{\i} entropy. 
Here, the lower bound
is more explicit but the drawback of the proof is that this lower bound may be negative (and thus the result empty). 
This only happens when the expansion rate $\chi_{\max}$ is too large compared with the averaged Lyapunov exponents 
with respect to $\mu$. This term $\chi_{\max}$ comes out because their proof relies on an application of the Egorov 
Theorem~\eqref{e:egorov} for range of times that tend to $+\infty$ as $\hbar\rightarrow 0$. 
In~\cite{AnantharamanNonnenmacher07b, DyatlovGuillarmou14}, it is shown that the optimal 
range of times for which the Egorov property~\eqref{e:egorov} remains true (up to a $o(1)$ remainder) is 
$$T_E(\hbar)=\frac{|\log\hbar|}{2\chi_{\max}},$$
which is known as the Ehrenfest time. This is the time scale until which the semiclassical approximation remains valid. 
Understanding the semiclassical properties of a quantum system beyond this time scale is a difficult and subtle 
question that one has to face in order to make improvements on our understanding of the set $\ml{M}(\Delta_g)$.

Motivated by this Theorem and by the counterexamples of Theorems~\ref{t:catmap} and~\ref{t:kelmer} -- see 
also~\cite{AnantharamanNonnenmacher07a}, Anantharaman and Nonnenmacher also formulated that the following 
general result should hold~\cite{AnantharamanNonnenmacher07b}:
\begin{conj}\label{t:entropy} Suppose that $\varphi^t$ has the Anosov property.
Then, for every $\mu\in\ml{M}(\Delta_g)$, one has
$$h_{KS}(\mu)\geq -\frac{1}{2}\int_{S^*M}\log J^ud\mu.$$
\end{conj}
This lower bound is in fact sharp in the counterexamples that we have described for quantum maps, and 
proving the Quantum Unique Ergodicity conjecture would amount to remove the $1/2$ in this inequality. In the case 
of certain quantum maps with variable Lyapunov exponents in dimension $2$, this conjecture was proved by 
Gutkin~\cite{Gutkin2010}. Then, in the case of Anosov geodesic flows, one gets~\cite{Riviere10a, Riviere12}
\begin{theo}\label{t:riviere-entropy} Suppose that $\varphi^t$ has the Anosov property. Then, for every
 $P_0>0$, there exists some constant $c(P_0)>0$ such that, for every $\mu\in\ml{M}(\Delta_g)$, one has
 $$\mu\left(\left\{(x,\xi):h_{KS}(\mu_{x,\xi})\geq -\frac{1}{2}\int_{S^*M}\log J^ud\mu_{x,\xi}-P_0\right\}\right)
\geq c(P_0)>0.$$
In particular, one has, for any invariant subset $\Lambda$,
$$P_{\operatorname{top}}\left(\Lambda,\frac{1}{2}\right)\leq -P_0\ \Longrightarrow\ \mu(\Lambda)\leq 1-c(P_0)<1.$$
Moreover, if $\operatorname{dim}\ M=2$, then, one has, for every $\mu\in\ml{M}(\Delta_g)$,
 $$h_{KS}(\mu)\geq -\frac{1}{2}\int_{S^*M}\log J^ud\mu.$$
\end{theo}
This Theorem answers Anantharaman-Nonnenmacher's question in dimension $2$ and, in that case, it can also be extended 
to nonpositively curved manifold provided that we choose an appropriate definition for the unstable 
Jacobian~\cite{Riviere10b}. In higher dimensions, this conjecture remains open and, besides Theorem~\ref{t:anantharamankochnonnenmacher}, 
the best explicit bounds we are aware of are 
for locally symmetric spaces~\cite{AnantharamanSilberman13} and for symplectic linear maps of the multidimensional 
torus~\cite{Riviere11}.

The proof of these different results is slightly technical and it involves tools from ergodic theory, hyperbolic 
dynamical systems and semiclassical analysis. Instead of doing that, we will explain how these methods can be 
used to extend Theorem~\ref{t:cdv-parisse} to more general hyperbolic subsets $\Lambda$. 
More precisely, together with Nonnenmacher, we proved~\cite{Riviere14}:
\begin{theo}\label{t:nonnenmacher-riviere}Let $\Lambda$ be an hyperbolic subset for $\varphi^t$ such that
$$P_{\operatorname{top}}\left(\Lambda,\frac{1}{2}\right)<0.$$
Let $a\in\ml{C}^{\infty}_c(T^*M,[0,1])$ which is equal to $0$ near $\Lambda$ and to $1$ outside a 
slightly bigger neighborhood (depending on $\Lambda$). Then, there exists some $c_{a,\Lambda}>0$ such that, for any sequence 
$(\psi_{\hbar})_{\hbar\rightarrow 0^+}$ of solutions to~\eqref{e:eigenvalue},
$$\liminf_{\hbar\rightarrow 0^+}|\log\hbar|\la\psi_{\hbar},\Oph(a)\psi_{\hbar}\ra\geq c_{a,\Lambda}>0.$$ 
 
\end{theo}
We emphasize that, compared with Theorems~\ref{t:anantharaman},~\ref{t:anantharamankochnonnenmacher} 
and~\ref{t:riviere-entropy}, no global assumption is made on the manifold. We will give a sketch of proof in paragraph~\ref{ss:proof-pressure} and we will also briefly explain which steps 
can be optimized when we make some global assumptions on the manifold (e.g. Anosov property of the geodesic flow).

\subsubsection{Unconditional observability in dimension $2$}

We can now come back to the question of observability in the case of an Anosov geodesic flow. Given an open set $\omega$ of $S^*M$ (not necessarly of $M$), one can define the hyperbolic subset
$$\Lambda_{\omega}:=\bigcap_{t\in\IR}\varphi^t(S^*M\setminus \omega).$$
The entropic Theorems from paragraph~\ref{sss:entropy} imply that if 
$$P_{\operatorname{top}}\left(\Lambda_\omega,\frac{1}{2}\right)<0,$$
then $\mu(\omega)\geq c_{\omega}>0$ for any $\mu\in\ml{M}(\Delta_g)$. Hence, if $\omega$ is an open set of $S^*M$ such 
that $\Lambda_{\omega}$ is not too big in a dynamical sense, then eigenfunctions are uniformly observables on 
$\omega$. In the case where $\omega$ is an open set of $M$, these results implies an observability Theorem like 
Theorem~\ref{t:observability-torus} provided that $P_{\operatorname{top}}\left(\Lambda_{S^*\omega},\frac{1}{2}\right)<0$. 
In the case of hyperbolic surfaces, this condition was removed by 
Dyatlov and Jin~\cite{DyatlovJin17} using a fractal uncertainty principle due to Bourgain and Dyatlov~\cite{BourgainDyatlov18} -- see also~\cite{DyatlovZahl16}: 
\begin{theo}\label{t:FUP} Suppose that $M$ is an hyperbolic surface. Then, 
for any open set $\omega\neq\emptyset$ inside $S^*M$, there exists some constant $c_{\omega,M,g}>0$ such that, for any $\mu\in\ml{M}(\Delta_g)$,
$$\mu(\omega)\geq c_{\omega,M,g}>0.$$
\end{theo}
Compared with the case of flat tori, we emphasize that this Theorem is valid for any nonempty open set inside $S^*M$ (and not 
only on $M$). An extension of this result to variable negative curvature was recently announced by Dyatlov, Jin and 
Nonnenmacher~\cite{DyatlovJinNonnenmacher19}. A natural question is to understand what can happen in higher dimensions. In that case, one has to keep in mind 
Kelmer's construction (see Theorem~\ref{t:kelmer}) which shows that such a result cannot be a consequence of 
hyperbolicity only. 
\begin{rema} If we fix an invariant subset $\Lambda\neq S^*M$, then Theorem~\ref{t:FUP} tells us that 
$$\mu(\Lambda)\leq 1-c(\Lambda)<1,$$ 
for some constant $c(\Lambda)>0$ depending on $\Lambda$. Compared with 
Theorems~\ref{t:anantharaman} and~\ref{t:riviere-entropy}, this is valid for \emph{any} subset $\Lambda$. Yet, the 
constant depends on $\Lambda$ while the constant appearing in Theorems~\ref{t:anantharaman} 
and~\ref{t:riviere-entropy} depends only on the topological pressure of $\Lambda$ provided it is small enough. 
\end{rema}
For the sake of completeness, let us finally record the following remarkable Corollary of Theorem~\ref{t:FUP}:
\begin{coro}[Unconditional observability on negatively curved surfaces]
Suppose that $M$ is an hyperbolic surface and let $\omega$ be a nonempty open set of $M$. Then, there exists some contant $c_{\omega,g}>0$ such 
that, for every $\psi_{\lambda}$ solution to
$$-\Delta\psi_{\lambda}=\lambda^2\psi_{\lambda},\quad\|\psi_{\lambda}\|_{L^2}=1,$$
 one has 
 $$\int_{\omega}|\psi_\lambda(x)|^2d\operatorname{Vol}_g(x)\geq c_{\omega,g}>0.$$
 
\end{coro}

\subsection{Sketch of proof of Theorem~\ref{t:nonnenmacher-riviere}}\label{ss:proof-pressure}

The proof of Theorem~\ref{t:nonnenmacher-riviere} being slightly simpler than the ones of 
Theorem~\ref{t:anantharaman} or~\ref{t:riviere-entropy}, we will focus on this case without 
getting too much into the details. In the end, we will explain how a global assumption on 
the dynamics allows to improve the argument.

\subsubsection{Cover of the hyperbolic subset $\Lambda$}
Let $(V_a)_{a\in W}$ be an open cover of $\Lambda$ by some open subsets of
$T^*M$ adapted to the dynamics. In order to avoid complications, we remain vague on this aspect and 
we refer to~\cite[Par.~2.2]{Riviere14} for more details. We will just say that they are constructed by 
mimicking the construction of Bowen balls. Instead of considering balls of radius $\eps$ to construct 
dynamical balls, we consider a fixed open cover that we refine by the dynamics to obtain a cover 
adapted to the dynamics. In that manner, we obtain the following estimate
\begin{equation}\label{e:pressure-estimate}\sum_{a\in W}\sup_{\rho\in V_a} \exp\left(\frac{1}{2}\int_0^{T_0}\log J^u\circ\varphi^t(\rho)dt\right)\leq 
e^{T_{0} (P_{\text{top}}(\Lambda,1/2)+\eps_0)},\end{equation}
where $\eps_0>0$ is some small parameter which is fixed in advance and $T_0>0$ is large enough. We 
complete this cover by setting $V_{\infty}$ to be an open set not intersecting $\Lambda$ (or its image in the 
dilates of $\Lambda$ near $S^*M$). In the following, we set 
$\overline{W}=W\cup\{\infty\}$. Our goal is to prove that at least a fraction $|\log\hbar|^{-1}$ of the Wigner
distribution associated to the sequence $(\psi_{\hbar})_{0<\hbar\leq 1}$ is contained inside $V_{\infty}$. 
To that aim, we introduce a partition of unity $(P_{a})_{a\in \overline{W}}$ associated to the open cover 
$(V_a)_{a\in\overline{W}}$, i.e. 
$$\forall a\in\overline{W},\ P_a\in\ml{C}^{\infty}_c(V_a,[0,1])\ \text{and}\ \sum_{a\in\overline{W}} P_a(\rho)=1,$$ 
for every $\rho$ in a small neighborhood $S^*M$. Again, the reason for working near $S^*M$ is that eigenfunctions 
are asymptotically concentrated on $S^*M$.

\subsubsection{Quantum partitions of identity}
We now quantize this smooth partition into a family of pseudodifferential operators 
$(\pi_a)_{a\in\overline{W}}$ 
so that $\pi_a^*=\pi_a$, $\sigma(\pi_a)=P_a$ and
$$\sum_{a\in\overline{W}}\pi_a=\text{Id}_{L^2(M)}+\ml{O}(\hbar^{\infty}),$$
for data microlocalized near $S^*M$, e.g. for our sequence $(\psi_{\hbar})_{0<\hbar\leq1}$ of solutions 
to~\eqref{e:eigenvalue}. We say that we have a quantum partition of identity and these play a central role 
in the above entropic results. We can now construct a quantum analogue of the Bowen balls, and to that aim we set
$$\pi_a(t)=e^{-\frac{it\hbar\Delta_g}{2}}\pi_ae^{\frac{it\hbar\Delta_g}{2}},$$ 
We still have the relation
$$\sum_{a\in\overline{W}}\pi_a(t)=\text{Id}_{L^2(M)}+\ml{O}(\hbar^{\infty}),$$
for data microlocalized near $S^*M$. Hence, for $N\geq 0$, one can write
\begin{equation}\label{e:partition-identity}
 \sum_{\alpha\in\overline{W}^N}\left\la \psi_{\hbar},\pi_{\alpha_{N-1}}((N-1)T_0)\ldots\pi_{\alpha_1}(T_0)\pi_{\alpha_0}(0)\psi_{\hbar}\right\ra=1+\ml{O}(\hbar^{\infty}).
\end{equation}
Each of these operators can be thought as a semiclassical analogue of the characteristic function of a Bowen ball. In fact, for $N$ fixed and 
using~\eqref{e:composition} and~\eqref{e:egorov}, 
the principal symbol of this operator is 
$P_{\alpha_{N-1}}\circ\varphi^{(N-1)T_0}\times \ldots\times P_{\alpha_1}\circ\varphi^{T_0}\times P_{\alpha_0}$. The points inside 
the support of this symbol are exactly the points $\rho$ which at time $0$ are inside $V_{\alpha_0}$, at time $T_0$ 
inside $V_{\alpha_1}$, etc. The Egorov Theorem~\eqref{e:egorov} remains true for times of order 
$N=\kappa|\log\hbar|$ with $\kappa>0$ small enough up to some small error of order $\ml{O}(\hbar^{\delta})$ with 
$\delta>0$ depending on $\kappa$~\cite[Th.~11.12]{Zworski12}. Hence, the symbol remains of this form 
for this range of $N$ and we will use this observation several times in the following.

\subsubsection{Hyperbolic dispersion estimates}

The first key ingredient is an hyperbolic dispersion estimate which, in the present case, is due to Nonnenmacher and 
Zworski~\cite{NonnenmacherZworski09}. This estimate will allow us to control the terms in the 
sum~\eqref{e:partition-identity} which correspond to words $\alpha$ in $W^N$, i.e. points whose trajectory remains 
close to $\Lambda$. More precisely, for every $\ml{K}>0$, they show the existence of 
$C_{\ml{K}}>0$ and $\hbar_{\ml{K}}>0$ such that, for every $0<\hbar\leq \hbar_{\ml{K}}$, for every 
$N\leq \ml{K}|\log\hbar|$ and for every $\alpha$ in $W^N$,
\begin{equation}\label{e:hyp-dispersive-estimate}
\left\|\pi_{\alpha_{N-1}}((N-1)T_0)\ldots\pi_{\alpha_1}(T_0)\pi_{\alpha_0}\right\|_{L^2\rightarrow L^2}\leq C_{\ml{K}}
\hbar^{-\frac{n}{2}}(1+\ml{O}(\eps))^{NT_0}\prod_{k=0}^{N-1}\sup_{\rho\in V_{\alpha_k}} 
\exp\left(\frac{1}{2}\int_0^{T_0}\log J^u\circ\varphi^t(\rho)dt\right).
\end{equation}
where $\eps>0$ is an upper bound on the diameter of the open sets $(V_{\alpha})_{\alpha\in W}$. This is in fact the only place of 
the proof where the hyperbolicity is used. Proving this upper bound is quite subtle and it is one of the major steps 
in the proofs of~\cite{Anantharaman08, AnantharamanNonnenmacher07a}. We will skip the argument that would require 
to introduce more background material related to Fourier integral operators and we will just explain how it 
can be used here.

If we implement this upper bound and the estimate~\eqref{e:pressure-estimate} in~\eqref{e:partition-identity}, we can deduce
that
\begin{equation}\label{e:partition-identity2}
 \sum_{\alpha\in\overline{W}^N-W^N}\left\la \psi_{\hbar},\pi_{\alpha_{N-1}}((N-1)T_0)\ldots\pi_{\alpha_1}(T_0)\pi_{\alpha_0}(0)\psi_{\hbar}\right\ra
 =1+\ml{O}\left(\hbar^{-\frac{n}{2}}e^{NT_{0} (P_{\text{top}}(\Lambda,1/2)+\eps_0)}\right).
\end{equation}
Now, using the fact that $P_{\text{top}}(\Lambda,1/2)<0$ and taking $N=[\ml{K}|\log\hbar|]$ with $\ml{K}>0$ large enough, 
one finds that
\begin{equation}\label{e:partition-identity3}
 \sum_{\alpha\in\overline{W}^N-W^N}\left\la \psi_{\hbar},\pi_{\alpha_{N-1}}((N-1)T_0)\ldots\pi_{\alpha_1}(T_0)\pi_{\alpha_0}(0)\psi_{\hbar}\right\ra
 =1+o(1).
\end{equation}
Hence, we have shown that those elements $\alpha$ of the partition that correspond to trajectories that spend some time 
away from $\Lambda$ have a total mass which is asymptotically of order $1$. The difficulty is now to show that there is indeed 
some fraction of the mass outside $\Lambda$. This will follow from a combinatorial work whose ideas go back to the works of 
Anantharaman~\cite{Anantharaman08} who showed that these quantities verify some subadditive properties.

Before explaining that, let us observe that~\eqref{e:partition-identity3} can be rewritten as
\begin{equation}\label{e:partition-identity4}
 \left\la \psi_{\hbar},\sum_{\alpha\in\overline{W}^N-W^N}\pi_{\alpha_{N-1}}(T_0(N-1))\ldots\pi_{\alpha_1}(T_0)\pi_{\alpha_0}(0)\psi_{\hbar}\right\ra
 =1+o(1).
\end{equation}
Thus, by Cauchy-Schwarz inequality,
\begin{equation}\label{e:partition-identity5}
 \left\|\sum_{\alpha\in\overline{W}^N-W^N}\pi_{\alpha_{N-1}}((N-1)T_0)\ldots\pi_{\alpha_1}(T_0)\pi_{\alpha_0}(0)\psi_{\hbar}\right\|
 \geq 1+o(1).
\end{equation}

\subsubsection{Subadditivity property: from long logarithmic times to short ones}

Let us write $N=kN_1\simeq\mathcal{K}|\log\hbar|$ with 
$N_1\simeq[\kappa|\log\hbar|]$ so that the semiclassical rules of paragraph~\ref{ss:semiclassical} applies up to small error 
terms (e.g. the Egorov Theorem or the composition rule) at the scales $N_1$. We can already note that $k$ is bounded in terms of $\ml{K}$ and 
$\kappa$. Hence, using the semiclassical rules of paragraph~\ref{ss:semiclassical}, one finds
$$\left\|\sum_{\gamma\in W^{N_1}}\pi_{\gamma_{N_1-1}}((N_1-1)T_0)\ldots\pi_{\gamma_1}(T_0) \pi_{\gamma_0}\right\|_{L^2\rightarrow L^2}
=\ml{O}(1).$$
Observe now that
$$\overline{W}^N-W^N=\bigsqcup_{p=1}^k\left\{\alpha=(\overline{\gamma},\gamma,\tilde{\gamma}):
\overline{\gamma}\in\overline{W}^{N_1(p-1)},\ \gamma\in \overline{W}^{N_1}-W^{N_1},\ 
\tilde{\gamma}\in W^{N_1(k-p)}\right\}.$$
This allows us to rewrite~\eqref{e:partition-identity5} under the form
\begin{equation}\label{e:partition-identity6}
 1+o(1)\leq\ml{O}(1) \sum_{p=1}^k\left\|
 \sum_{\gamma\in\overline{W}^{N_1}-W^{N_1}}\pi_{\gamma_{N_1-1}}e^{\frac{i\hbar T_0\Delta_g}{2}}\ldots 
 \pi_{\gamma_1}e^{\frac{i\hbar T_0\Delta_g}{2}}\pi_{\gamma_0}e^{\frac{i\hbar T_0 N_1 p\Delta_g}{2}}\psi_{\hbar}\right\|.
\end{equation}
Then, using that $e^{\frac{i\hbar t\Delta_g}{2}}$ is unitary and that the $(\psi_{\hbar})_{0<\hbar\leq 1}$ are solutions 
to~\eqref{e:eigenvalue},
\begin{equation}\label{e:partition-identity7}
 1+o(1)\leq\ml{O}(1)\left\|\sum_{\gamma\in\overline{W}^{N_1}-W^{N_1}}\pi_{\gamma_{N_1-1}}((N_1-1)T_0)\ldots 
 \pi_{\gamma_1}(T_0)\pi_{\gamma_0}(0)\psi_{\hbar}\right\|.
\end{equation}
Taking the square of this inequality, we obtain the existence of a constant $C>0$ such that
 \begin{equation}\label{e:partition-identity8}
 C\leq \left\|\sum_{\gamma\in\overline{W}^{N_1}-W^{N_1}}\pi_{\gamma_{N_1-1}}((N_1-1)T_0)\ldots 
 \pi_{\gamma_1}(T_0)\pi_{\gamma_0}(0)\psi_{\hbar}\right\|^2.
\end{equation}
Using Egorov property~\eqref{e:egorov} for short logarithmic times (recall that 
$N_1\simeq[\kappa|\log\hbar|]$ with $\kappa>0$ small enough) and the composition rule, this inequality can be rewritten as
\begin{equation}\label{e:partition-identity9}
 C\leq \left\la\psi_{\hbar},\Oph\left(\left(\sum_{\gamma\in\overline{W}^{N_1}-W^{N_1}}P_{\gamma_{N_1-1}}\circ\varphi^{(N_1-1)T_0}
 \times\ldots\times P_{\gamma_1}\circ\varphi^{T_0}\times P_{\gamma_0}\right)^2\right)\psi_{\hbar}\right\ra+o(1).
\end{equation}
Now we can make use of the G\aa{}rding inequality~\eqref{e:garding} and of the fact that the principal symbol is
$\leq 1$ to deduce:
\begin{equation}\label{e:partition-identity10}
 C\leq \left\la\psi_{\hbar},\Oph\left(\sum_{\gamma\in\overline{W}^{N_1}-W^{N_1}}P_{\gamma_{N_1-1}}\circ\varphi^{(N_1-1)T_0}
 \times\ldots\times P_{\gamma_1}\circ\varphi^{T_0}\times P_{\gamma_0}\right)\psi_{\hbar}\right\ra+o(1).
\end{equation}

\subsubsection{Subadditivity again: from short logarithmic times to finite times}

We can now reiterate the same procedure to go down to finite times. We note that
$$\overline{W}^{N_1}-W^{N_1}=\bigsqcup_{p=1}^{N_1-1}\left\{\alpha=(\overline{\gamma},\infty,\tilde{\gamma}):
\overline{\gamma}\in\overline{W}^{p-1},\ \tilde{\gamma}\in W^{N_1-p}\right\}.$$
Using the fact that we have a partition of unity near $S^*M$ (and that $\psi_\hbar$ is microlocalized near $S^*M$)
and the G\aa{}rding inequality, we can deduce from~\eqref{e:partition-identity10} that
\begin{equation}\label{e:partition-identity11}
 C\leq \sum_{p=1}^{N_1-1}\left\la\psi_{\hbar},\Oph\left(P_{\infty}\circ\varphi^{pT_0}\right)\psi_{\hbar}\right\ra+o(1).
\end{equation}
Using one last time Egorov theorem for short logarithmic times and the fact that the $(\psi_{\hbar})_{0<\hbar\leq 1}$ 
are Laplace eigenfunctions, we can conclude that
$$C\leq N_1\left\la\psi_{\hbar},\Oph\left(P_{\infty}\right)\psi_{\hbar}\right\ra+o(1).$$
Recall now that $N_1$ is of order $|\log\hbar|$ which concludes the proof of Theorem~\ref{t:nonnenmacher-riviere}.

\subsubsection{Adding a global assumption on the dynamics}

As the reader can see, the proof of Theorem~\ref{t:nonnenmacher-riviere} is already quite involved but it illustrates 
the general scheme that was initiated by Anantharaman in~\cite{Anantharaman08}. In order to obtain the results 
for Anosov flows (e.g. Theorems~\ref{t:anantharaman}), one needs to use an earlier version of the 
hyperbolic dispersion estimate~\eqref{e:hyp-dispersive-estimate} due to Anantharaman and 
Nonnenmacher~\cite{Anantharaman08, AnantharamanNonnenmacher07b} and valid for any sequence $\gamma$ associated 
with a partition of unity. In particular, it still holds for sequences which correspond to trajectories that 
are away from $\Lambda$: this is where one uses the global assumption on the dynamics (i.e. the 
Anosov property). Implementing this, we can be less restrictive when we cut the sum in two pieces at the 
beginning of the proof. In particular, we can allow sequences $\gamma$ corresponding to trajectories that are 
away from $\Lambda$ for a small fraction of times. Proceeding like this, we have less terms in the remaining sums 
and we can proceed to a better subadditive argument that will not end with this $|\log\hbar|$ factor -- 
see~\cite{Anantharaman08, Riviere12} for more details.

\subsection{Sketch of proof of Theorem~\ref{t:FUP}}

To end this section on Laplace eigenfunctions for chaotic geodesic flows, we give a sketch of the proof of 
Theorem~\ref{t:FUP}. This is based on a fractal 
uncertainty principle~\cite{DyatlovZahl16, BourgainDyatlov18} that can be used to replace the hyperbolic 
dispersion estimate~\eqref{e:hyp-dispersive-estimate} used in the previous argument. We follow the presentation of~\cite{Dyatlov17} and we refer the reader 
to this reference for an account on these recent developments. We fix a sequence of Laplace eigenfunctions 
$(\psi_{\hbar})_{\hbar\rightarrow 0^+}$ (i.e. solutions to~\eqref{e:eigenvalue}) having a single 
semiclassical measure. We let $\omega$ be a nonempty open set of $S^*M$ and $\omega_1\neq\emptyset$ 
such that $\overline{\omega_1}\subset\omega$. As before, we introduce a partition 
$(\pi_1,\pi_2)$ of identity, i.e
$$\pi_1+\pi_2=\text{Id}_{L^2(M)}+\ml{O}(\hbar),$$
for sequences of initial data which are microlocalized on $S^*M$. We also make the assumption that the principal symbol $P_2$ 
of $\pi_2$ does not intersect $\omega_1$ (and its dilation) in a small neighborhood of $S^*M$ while the 
operator $\pi_1$ is microlocalized in $\omega$ near $S^*M$. We now fix some logarithmic 
time
$$N=[\ml{K}|\log\hbar|],$$
where $0<\ml{K}<1$ is a parameter which is close to $1$. In that geometric context, the Egorov Theorem is only 
valid up to times of order $\frac{1}{2}|\log\hbar|$~\cite{BouzouinaRobert02, 
AnantharamanNonnenmacher07b, DyatlovGuillarmou14}. Thus, $N_1$ is larger than the scale where we can expect the 
semiclassical rules of paragraph~\ref{ss:semiclassical} to work. In order to overcome this difficulty, Dyatlov and Jin 
make use of two quantization procedures that are adapted respectively to symbols whose 
derivatives are bounded in $\hbar$ along the stable manifold (resp. along the unstable manifold) but may blow 
in $h^{-\kappa}$ when one differentiates in other directions. These kind of quantization procedures adapted to a Lagrangian 
foliation were constructed by Dyatlov and Zahl~\cite[$\S 3$]{DyatlovZahl16} and they are denoted 
respectively by $\Oph^{L_s}$ and $\Oph^{L_u}$ in the present case. They correspond to different class of symbols and they are a 
priori not compatible with each other. In particular, $\Oph^{L_s}(b_1)\Oph^{L_u}(b_2)$ is not necessarly a 
pseudodifferential operator. Yet, it has the advantage that, up to the time scale $N$, we can apply 
semiclassical rules provided we consider the appropriate sense of times. More precisely, one has among other things
\begin{equation}\label{e:fup-pi-unstable}
\Pi_+:=\pi_2(-1)\pi_2(-2)\ldots\pi_2(-N)=\Oph^{L_u}\left(\prod_{j=1}^{N}P_2\circ\varphi^{-j}\right)+\ml{O}(\hbar^{1-\ml{K}}),
\end{equation}
and
\begin{equation}\label{e:fup-pi-stable}
 \Pi_-:=\pi_2(N-1)\ldots\pi_2(1)\pi_2(0)=\Oph^{L_s}\left(\prod_{j=0}^{N-1}P_2\circ\varphi^{j}\right)+\ml{O}(\hbar^{1-\ml{K}}),
\end{equation}
as we had in the previous proof but for much smaller scales of times, i.e. $\kappa|\log\hbar|$ with some fixed $\kappa\ll 1.$ 
One can then use a version of the fractal uncertainty principle as it appears in the work of 
Bourgain and Dyatlov~\cite{BourgainDyatlov18} 
in order to show that
\begin{equation}\label{e:FUP}
 \left\|\Oph^{L_s}\left(\prod_{j=0}^{N-1}P_2\circ\varphi^{j}\right)\Oph^{L_u}\left(\prod_{j=1}^{N}P_2\circ\varphi^{-j}\right)\right\|_{L^2\rightarrow L^2}\leq 
 C\hbar^{\beta},
\end{equation}
for some $\beta>0$ depending only on the open set $\omega$. This is exactly the estimate that will replace the 
hyperbolic dispersion estimate~\eqref{e:hyp-dispersive-estimate} that we used in our previous 
proof. Again these last steps are the ones that crucially use the hyperbolic structure. Let us now try to implement 
this new information. As before, we use the quantum partition of unity to write 
$$\sum_{\gamma\in\{1,2\}^{2N}}
\left\la\psi_{\hbar}
,\pi_{\gamma_{N-1}}(N-1)\ldots\pi_{\gamma_1}(1)\pi_{\gamma_0}(0)\pi_{\gamma_{-1}}(-1)\ldots\pi_{\gamma_{-N}}(-N)\psi_{\hbar}\right\ra=1+o(1).$$
Implementing~\eqref{e:FUP}, we find that
$$\sum_{\gamma\in\{1,2\}^{2N}\setminus\{2\}^{2N}}
\left\la\psi_{\hbar}
,\pi_{\gamma_{N-1}}(N-1)\ldots\pi_{\gamma_1}(1)\pi_{\gamma_0}(0)\pi_{\gamma_{-1}}(-1)\ldots\pi_{\gamma_{-N}}(-N)\psi_{\hbar}\right\ra=1+o(1).$$
We now proceed to a subadditive argument as before. First, we set $N_1=[\kappa|\log\hbar|]$ with $\kappa\ll 1$
Reorganizing the sums as we did before and applying the Cauchy-Schwarz inequality, 
we find that there exists some constant $C>0$ such that
$$C\leq\left\|\sum_{\gamma\in\{1,2\}^{N_1}\setminus\{2\}^{N_1}}\pi_{\gamma_{N_1-1}}(N_1-1)\ldots\pi_{\gamma_1}(1)\pi_{\gamma_0}(0)
\psi_{\hbar}\right\|.$$
As we are now handling much smaller time scales, we can proceed as in the previous proof to find that
$$C^2\leq \sum_{p=1}^{N_1-1}\langle\psi_{\hbar},\Oph(P_1\circ\varphi^p)\psi_{\hbar}\rangle+o(1).$$
Applying one last time the Egorov Theorem and using the fact that $\psi_{\hbar}$ is a normalized Laplace 
eigenfunction, we find that
$$\frac{C^2+o(1)}{N_1}\leq\la\psi_\hbar,\pi_1\psi_\hbar\ra.$$
Recall now that the pseudodifferential operator $\pi_1$ is microlocalized inside the open set $\omega$ we are interested in. 
Hence, one has a fraction $|\log\hbar|^{-1}$ which is contained inside this open set. As before, in order to remove this 
$|\log\hbar|^{-1}$ factor, we have to be slightly more careful in the definition of $\Pi_{\pm}$ and to allow that there is a 
small fraction of terms with $\gamma=1$, meaning that it corresponds to trajectories that send a small amount of 
times away from the open set $\omega_1$. If this fraction is small enough in terms of the exponent $\beta>0$ appearing in 
the fractal uncertainty principle~\eqref{e:FUP}, then one can still apply the estimate to these extra terms. The 
remaining sums are again much smaller and performing an argument similar to the one from~\cite{Anantharaman08} allows to 
remove the logarithmic factor -- see~\cite{DyatlovJin17} for details.

\subsection{Some comments on the key estimates~\eqref{e:hyp-dispersive-estimate} and~\eqref{e:FUP}}

The two proofs that we have described follows a similar scheme that was initially 
introduced in~\cite{Anantharaman08}. In particular, both proofs require at some point an estimate on the $L^2$-norm 
of elements inside our quantum partition associated with classical trajectories that avoid the region 
where we want to show that the quantum state is observable. Hence, one needs to show that this norm is asymptotically 
small and this is where the chaotic features of the flow are crucially used in both proofs. Yet, the proofs of these two 
estimates are of rather different nature even if they both consist in estimating some oscillatory integral. 
Observe that an important difference between the two estimates is that~\eqref{e:hyp-dispersive-estimate} yields 
a nontrivial information for symbols $P_i$ supported in a subset of small diameter while~\eqref{e:FUP} 
holds for \emph{any} symbol $P_2$. This subtle difference between the two statements is responsible for the fact that the 
entropic results yields observability for big enough subsets of $S^*M$ while Theorem~\ref{t:FUP} gives it for 
any subset (under more restrictive assumptions on $M$).

The hyperbolic dispersion estimates of Anantharaman, Nonnenmacher and Zworski is based on WKB type asymptotics. 
More precisely, after some Fourier decomposition which is responsible for the loss $\hbar^{-\frac{n}{2}}$ 
in~\eqref{e:hyp-dispersive-estimate}, the proof amounts to a precise description of the action 
of elements inside the quantum partition on Lagrangian states $(a(x)e^{\frac{iS(x)}{\hbar}})_{\hbar>0}$ (as the ones from example~\eqref{ex:lagrangian}). 
Then, the decay follows from the exact expression obtained in the WKB asymptotics and from the fact that the 
Lagrangian submanifolds $\ml{L}_S=\{(x,d_xS)\}$ become closer to the unstable direction under the classical evolution.

On the other hand, using the hyperbolic structure, Dyatlov and Jin show how to reduce the crucial upper 
bound~\eqref{e:FUP} to some question from harmonic analysis. This question is referred 
as a fractal uncertainty principle in the sense that it gives an estimate on the $L^2$ mass of a function 
inside a subset having a ``fractal'' structure when its Fourier transform is itself supported in a subset 
with similar properties. More precisely, relying on the unique ergodicity of the 
horocycle flow\footnote{This is the flow generated by the vector field in the unstable direction.}, 
they show that the supports of the symbols involved in~\eqref{e:FUP} have holes 
of uniform size\footnote{The size depends only on $\omega$.} $0<\delta<1$ 
along the stable (resp. unstable) direction at all scales 
between $\hbar^{\ml{K}}$ and $1$ -- see~\cite[Lemma~3.2]{Dyatlov17} for a precise statement. 
The supports are said to be $\delta$-porous at these scales. Then, one can reduce 
the problem to a one dimensional problem where the stable and unstable directions are in some sense 
dual to each other in terms of Fourier transform. Thus, proving the upper bound~\eqref{e:FUP} 
amounts to prove some subtle estimate on the properties of the one-dimensional Fourier transform~\cite{BourgainDyatlov18}. 
Roughly speaking, if the Fourier transform $\widehat{f}$ of $f\in L^2(\IR)$ is supported in a set $Y$ having 
holes of size $0<\delta<1$ at each scale between $\hbar^{\ml{K}-1}$ and $\hbar^{-1}$, 
then the $L^2$-mass of $f$ inside a subset
having holes of size $\delta$ at each scale between $\hbar^{\ml{K}}$ and $1$ is of order $\hbar^{\beta-2(1-\ml{K})}\|f\|_{L^2}$ for some 
$\beta>0$ (depending only on $\delta$ thus on $\omega$). We refer the reader to the recent review of Dyatlov~\cite{Dyatlov19} 
for background and ideas behind these fractal uncertainty principles.

\section{Application to nodal sets and $L^p$ norms}\label{ss:application}

We mostly focused on the study of eigenfunctions via their semiclassical measures which are natural from the 
point of view of quantum mechanics. Yet, from the mathematical perspective, there are several other interesting 
quantities that can be used to describe the singularities and the concentration properties of Laplace 
eigenfunctions. To end up these lectures, we would like to mention at least two of them and to explain how they are 
related to the questions dicussed in these notes.

\subsection{Bound on $L^p$ norms}

Without any assumption on the Riemannian manifold $(M,g)$, the best bounds that one can expect 
on the $L^p$ norms of Laplace eigenfunctions are due to H\"ormander for $p=+\infty$~\cite{Hormander68} and 
to Sogge~\cite{Sogge88} for $p<+\infty$:
\begin{theo}\label{t:sogge} Let $2\leq p\leq +\infty$. There exists $C_{M,g,p}>0$ such that, for every $\psi_{\lambda}$ solution to
$$-\Delta_g\psi_{\lambda}=\lambda^2\psi_{\lambda},$$
one has
$$\|\psi_{\lambda}\|_{L^p(M)}\leq C_{M,g,p}(1+\lambda)^{\sigma(p)}\|\psi_{\lambda}\|_{L^2(M)},$$
o\`u 
 $$\sigma(p):=\max\left\{\frac{n-1}{2}\left(\frac{1}{2}-\frac{1}{p}\right),\frac{n-1}{2}-\frac{n}{p}\right\}.$$
\end{theo}
The exponent $p_c=\frac{2(n+1)}{n-1}$ plays an 
important role as it corresponds to the value where the functions defining $\sigma(p)$ coincide. Any 
improvement on the upper bound of the $L^{p_c}(M)$-norm will then give by interpolation improvements 
for every $2<p<+\infty$. This result is optimal if we do not make any extra assumptions on $(M,g)$ as these upper bounds are saturated on the 
sphere endowed with its canonical metric~\cite{Sogge15}. Yet, it is natural to wonder if appropriate geometric 
assumptions allow to improve this result. We will not try to review the vast literature on the subject and we will 
just make a simple heuristic calculation that allows to understand the relation with the Wigner distributions 
we have encountered in the previous sections. This computation is at the origin of a joint work with 
Hezari~\cite{HezariRiviere16} and it was improved by Sogge in~\cite{Sogge16}.

Let us consider an orthonormal sequence of Laplace eigenfunctions $(\psi_k)_{k\in S}$ and a sequence of radii $r_k>0$ 
such that one can show
\begin{equation}\label{e:comparable-L2}\forall x_0\in M,\ \forall k \in S,\quad \frac{1}{\operatorname{Vol}_g(B(x_0,r_k))}\int_{B(x_0,r_k)}
|e_k(x)|^2d\operatorname{Vol}_g(x)\leq c_0,\end{equation}
for some uniform constant $0<c_0$. This kind of behaviour typically occurs when one has a sequence of 
eigenfunctions verifying the Quantum Ergodicity property. Let us now try to explain some heuristic idea to 
improve Sogge's estimate. We fix $x_0$ in $M$ and we work in a small geodesic chart centered at $x_0$. 
We rescale the function $\psi_k$ near $x_0$ by setting $\tilde{\psi}_k(y)=\psi_k(r_k y)$. 
This new function verifies near $0$:
$$\tilde{\Delta}_g\tilde{\psi}_k\approx r_k^2\lambda_k^2\tilde{\psi}_k$$
and, thanks to~\eqref{e:comparable-L2},
$$\int_{B(0,1)}|\tilde{\psi}_k(y)|^2dy\lesssim c_0.$$
Hence, locally near $x_0$, one has a quasimode of the Laplacian to which we can apply Sogge's estimates\footnote{These estimates are valid 
for not too bad quasimodes.}. This implies that for $p=p_c$,
$$\int_{B(0,1)}|\tilde{\psi}_k(y)|^{p_c}dy\lesssim r_k\lambda_k\left(\int_{B(0,1)}|\tilde{\psi}_k(y)|^2dy\right)^{\frac{p_c}{2}},$$
where we recall that $\sigma(p_c)=1/p_c$. We now make the change of variables in the other direction, i.e. $x=r_ky$, and we get
$$\frac{1}{\operatorname{Vol}_g(B(x_0,r_k))}\int_{B(x_0,r_k)}
|\psi_k(x)|^{p_c}d\text{Vol}_g(x)\lesssim r_k\lambda_k
\left(\frac{1}{\operatorname{Vol}_g(B(x_0,r_k))}\int_{B(x_0,r_k)}|\psi_k(y)|^2d\text{Vol}_g(x)\right)^{\frac{p_c}{2}}.$$
Using property~\eqref{e:comparable-L2}, we obtain
$$\int_{B(x_0,r_k)}
|\psi_k(x)|^{p_c}d\text{Vol}_g(x)\lesssim r_k\lambda_kc_0^{\frac{p_c}{2}}\operatorname{Vol}_g(B(x_0,r_k)).$$
Let us now cover $M$ by a ``minimal'' family of balls of radius $r_k$. Summing these inequalities, we find the 
following upper bound:
$$\left(\int_{M}
|\psi_k(x)|^{p_c}d\text{Vol}_g(x)\right)^{\frac{1}{p_c}}\lesssim 
\left(r_k\lambda_k\right)^{\frac{1}{p_c}}c_2^{\frac{1}{2}},$$
which improves Sogge's upper bound by a factor $r_k^{\sigma(p_c)}$. This heuristic argument shows that an upper bound 
like~\eqref{e:comparable-L2} allows to improve Theorem~\ref{t:sogge}. Hence, it transfers the problem of estimating 
$L^p$-norms into a problem of controlling the Wigner distribution inside small balls of the configuration 
space. This formal argument can be made rigorous~\cite{HezariRiviere16, Sogge16}:
\begin{theo}\label{t:Lp-Sogge} 
There exists some $C_{M,g}>0$ such that, for every $\psi_{\lambda}$ solution to
$$-\Delta_g\psi_{\lambda}=\lambda^2\psi_{\lambda},\ \|\psi_{\lambda}\|_{L^2}=1,$$
with $\lambda>0$, one has, for every $\lambda^{-1}\leq r\leq \operatorname{Inj}(M,g)$,
$$\|\psi_{\lambda}\|_{L^{p_c}(M)}\leq C_{M,g,p}\lambda^{\sigma(p_c)}\left(r^{-\frac{n+1}{2}}\sup_{x\in M}
\left\{\int_{B(x,r)}
|\psi_{\lambda}(x)|^2d\operatorname{Vol}_g(x)\right\}\right)^{\frac{1}{n+1}}.$$
\end{theo}
In particular, one can apply the Quantum Ergodicity Theorem~\ref{t:QE} to show the following. 
If $L_1$ is ergodic, then one can extract 
a density one subsequence of an orthonormal basis along which one has the upper bound $o(\lambda^{\sigma(p_c)})$~\cite{Sogge16}. In the 
case of negatively curved manifolds, we showed with Hezari how to refine a quantitative version of the 
Quantum Ergodicity Theorem~\cite{Zelditch94} and how to derive from it an upper bound of the 
type~\eqref{e:comparable-L2} for radius of logarithmic size in $\lambda$~\cite{HezariRiviere16} -- see also~\cite{Han15}. 
As a direct Corollary, this gives a logarithmic improvement of Sogge's upper bound at $p=p_c$ along a density one subsequence.

This is not only the manner to relate $L^p$-norms with Wigner distributions and let us single out two other approaches that 
appeared recently in the literature. First, for $p\leq p_c$, Blair and Sogge showed in a series of works how to control $L^p$-norms by 
Kakeya-Nikodym norms -- see~\cite{Sogge11, BlairSogge17} and the references therein. More precisely, Kakeya-Nikodym norms can be 
expressed in terms of the Wigner distribution of Laplace eigenfunctions in small tubes of size 
$\lambda^{-1/2}=\sqrt{\hbar}$ around pieces of geodesics. So, one more time, the problem is transferred to a 
question on Wigner distributions. 
For the case $p=+\infty$, Galkowski-Toth~\cite{GalkowskiToth17} and then 
Canzani-Galkowski~\cite{CanzaniGalkowski18} developped some new approach to control the $L^{\infty}$-norm of 
Laplace eigenfunctions in terms of their semiclassical measures. Roughly speaking, they showed how to control the value of 
$\psi_{\lambda}$ at a point $x_0$ in terms of the mass that the Wigner distribution puts on the ($n$-dimensional) 
submanifold $\cup_{t=-T}^{T}\varphi^t(S_{x_0}^*M)$ for some large enough $T>0$. In particular, if one can exhibit 
some situations where this mass goes to $0$ as $\lambda\rightarrow+\infty$ (e.g. if $L_1$ is ergodic), then it yields an improvement on the growth 
of the $L^{\infty}$-norm.

 \subsection{Complex zeroes}
 
 We conclude with the question of nodal sets which are the vanishing locus of Laplace eigenfunctions. Contrary to the case 
 of $L^p$-norms, there is in general no clear relation between the geometry of these nodal sets and the dynamical 
 properties of the geodesic flow. Yet, we would like to mention a result due to Zelditch~\cite{Zelditch07} who 
 showed how to relate complex zeroes of Laplace eigenfunctions to Wigner distributions and thus to the dynamics 
 of the geodesic flow. We also refer to~\cite{NonnenmacherVoros98} for earlier related results in the case of quantum maps.
 
 In~\cite{Zelditch07}, Zelditch proposed that, on a real analytic manifold $(M,g)$, one should consider the 
 vanishing locus of the complexified eigenfunctions rather than the one of real eigenfunctions. In some very vague sense, one can 
 expect that things become simpler as when one goes from real roots of a polynomial to complex ones. Let us be a little 
 bit more precise. If $(M,g)$ is a real analytic compact manifold, there is a natural complexification of the 
 manifold which can be identified with some ball bundle $B_{\eps_0}^*M:=\{(x,\xi)\in T^*M:\|\xi\|_x<\eps_0\}$~\cite{Zelditch07}.
 This complexification is known as the Grauert tube of $(M,g)$. Building on a seminal work of Boutet de 
 Monvel~\cite{BoutetdeMonvel78} -- see also~\cite{Lebeau18}, Zelditch explained how a solution $\psi_{\lambda}$ to~\eqref{e:Laplace} 
 can be analytically continued to the Grauert tube into a function $\psi_{\lambda}^{\IC}$. 
 
 \begin{ex} In order to understand this procedure, one can take the example of the flat torus $\IT^n$. In that case, 
 the Grauert tube is given by
 $$\{z=x+iy:\ x\in\IT^n,\ y\in\IR^n\},$$
and an eigenfunction like $\cos(n.x)$ naturally extends as $\cos (n.z).$  
 \end{ex}

 Once one has defined these analytic continuations of eigenfunctions, it is natural to consider their ``complex zeros'':
 $$\ml{Z}_{\psi_\lambda^{\IC}}:=\{(x,\xi)\in B_{\eps_0}^*M: \psi_{\lambda}^{\IC}(x,\xi)=0\}.$$
 With these conventions, Zelditch proved~\cite{Zelditch07}:
 \begin{theo} Let $(M,g)$ be a real analytic manifold and let $(\psi_{\lambda})_{\lambda\rightarrow+\infty}$ be a sequence of solutions 
 to~\eqref{e:Laplace} generating a single semiclassical measure $\mu$. Suppose that, for any non empty open subset 
 $\omega$ of $S^*M$, $\mu(\omega)>0$. Then, for any $\theta\in\Omega_c^{n-1,n-1}(B_{\eps_0}^*M\setminus 0)$, one has
$$\lim_{\lambda\rightarrow+\infty}\frac{1}{\lambda}\int_{\ml{Z}_{\psi_\lambda^{\IC}}}\theta=\frac{i}{\pi}
\int_{B_{\eps_0}^*M}\partial\overline{\partial}\|\xi\|\wedge\theta.$$ 
\end{theo}
In other words, this Theorem shows that the complex zeros of a sequence of Laplace eigenfunctions become equidistributed inside the 
Grauert tube under a certain condition on their semiclassical measure. The relation with Wigner distribution becomes slightly more 
clear when one applies Lelong-Poincar\'e formula:
 $$\int_{\ml{Z}_{\psi_\lambda^{\IC}}}\theta
 =\frac{i}{2\pi}\int_{B_{\eps_0}^*M}\partial\overline{\partial}\log|\psi_\lambda^{\IC}|^2\wedge\theta.$$
 Hence one can reduce the problem (up to an integration by part) to the description of the limit of 
 $\frac{1}{\lambda}\log|\psi_\lambda^{\IC}|^2$ in the space of distributions (or currents). To that aim, one can define
 $$U_\lambda^{\IC}(x,\xi):=\frac{\psi_\lambda^{\IC}(x,\xi)}{\left\|\psi_\lambda^{\IC}\right\|_{L^2\left(\partial 
 B_{\|\xi\|}^*M\right)}}.$$
 Hence, the problem amounts to compute the weak limit of
 $$\frac{1}{\lambda}\log|\psi_\lambda^{\IC}|^2=\frac{1}{\lambda}\log|U_\lambda^{\IC}|^2+\frac{2}{\lambda}
 \log\left\|\psi_\lambda^{\IC}\right\|_{L^2\left(\partial 
 B_{\|\xi\|}^*M\right)}.$$
 Going through the analytic continuation procedure, one can verify that the second term converges weakly 
 to $2\|\xi\|$~\cite[\S 4.2]{Zelditch07} and this does not depend on the properties of the semiclassical 
 measure. Then, Zelditch shows that the weak limits of $|U_\lambda^{\IC}|^2|_{\partial B_{\eps}^*M}$ are 
 in fact given by the semiclassical measure $\mu$ of our sequence of eigenfunctions and he proves that 
 $\frac{1}{\lambda}\log|U_\lambda^{\IC}|^2$ is a bounded sequence of subharmonic functions~\cite[\S 5.1]{Zelditch07}. 
 In particular, either the sequence tends uniformly to $-\infty$ on every compact, or it converges to some $v\leq 0$ 
 in $L^1_{\text{loc}}$. The first case cannot occur as $\mu$ is a probability measure. Now, one has to verify 
 that $v$ is equal to $0$. If not, this would mean that one can find a nonempty open subset $\omega$ where the 
 semiclassical measure is $0$ and thus contradict our assumption on $\mu$ -- see~\cite[\S 5.1]{Zelditch07} 
 for details.

 As stated in~\cite{Zelditch07}, one has the following consequence of Theorem~\ref{t:QE}:
 \begin{coro} Let $(M,g)$ be a real analytic manifold such that $L_1$ is ergodic. Then, for any orthonormal basis 
 $(\psi_{j})_{j\geq 1}$ of Laplace eigenfunctions, one can find a density $1$ subset $S$ inside $\IZ_+^*$ 
 such that, for any $\theta\in\Omega_c^{n-1,n-1}(B_{\eps_0}^*M\setminus 0)$, one has
$$\lim_{j\rightarrow+\infty, j\in S}\frac{1}{\lambda_j}\int_{\ml{Z}_{\psi_{j}^{\IC}}}\theta=\frac{i}{\pi}
\int_{B_{\eps_0}^*M}\partial\overline{\partial}\|\xi\|\wedge\theta.$$ 
 \end{coro}
Now regarding the recent developments of Dyatlov and Jin (see Theorem~\ref{t:FUP}), one has also
\begin{coro} Let $(M,g)$ be a real analytic surface with constant negative curvature 
 and let $(\psi_{\lambda})_{\lambda\rightarrow+\infty}$ be a sequence of solutions 
 to~\eqref{e:Laplace}. Then, for any $\theta\in\Omega_c^{n-1,n-1}(B_{\eps_0}^*M\setminus 0)$, one has
$$\lim_{\lambda\rightarrow+\infty}\frac{1}{\lambda}\int_{\ml{Z}_{\psi_\lambda^{\IC}}}\theta=\frac{i}{\pi}
\int_{B_{\eps_0}^*M}\partial\overline{\partial}\|\xi\|\wedge\theta.$$ 
\end{coro}

\section*{Ackowledgements}

I would like to thank S.~Nonnenmacher and J.~Sabin for organizing the spring school ``From quantum to classical''
from which these lecture notes grew up. I also thank S.~Dyatlov, D.~Robert and S.~Zelditch for answering some of 
my questions and for useful feedbacks. Finally, many thanks to M.~L\'eautaud for his careful 
reading and for pointing me several imprecisions in the first version of these notes.

\bibliographystyle{plain}
\bibliography{allbiblio}

\end{document}